\documentclass[11pt]{article}

\usepackage[letterpaper, portrait, margin=1in]{geometry}
\usepackage{hyperref}
\usepackage{tabularx}
\usepackage{algorithm}
\usepackage[noend]{algpseudocode}
\usepackage{booktabs,multirow}
\usepackage{url}
\usepackage{amsmath,amssymb,amsthm}
\usepackage{mathtools}
\usepackage{parskip}
\usepackage{xspace}
\usepackage{verbatim}
\usepackage{mathrsfs}
\usepackage[usenames,dvipsnames,svgnames,table]{xcolor}
\usepackage[noabbrev,capitalise,nameinlink]{cleveref}
\usepackage{pgf}
\usepackage{tikz}
\usepackage[dvipsnames]{xcolor}
\usepackage{todo}

\usetikzlibrary{intersections}

\theoremstyle{plain}
\newtheorem{theorem}{Theorem}[section]
\newtheorem{lemma}[theorem]{Lemma}

\newtheorem{proposition}[theorem]{Proposition}
\newtheorem{claim}[theorem]{Claim}
\newtheorem{corollary}[theorem]{Corollary}

\theoremstyle{definition}
\newtheorem{definition}[theorem]{Definition}

\newtheorem{example}[theorem]{Example}

\usepackage{thm-restate}
\declaretheorem[name=Theorem,numberlike=theorem]{restate-theorem}
\declaretheorem[name=Theorem,unnumbered]{restate-theorem*}

\newcommand{\cA}{\ensuremath{\mathcal{A}}}
\newcommand{\cB}{\ensuremath{\mathcal{B}}}
\newcommand{\cC}{\ensuremath{\mathcal{C}}}
\newcommand{\cD}{\ensuremath{\mathcal{D}}}
\newcommand{\cF}{\ensuremath{\mathcal{F}}}

\newcommand{\cH}{\ensuremath{\mathcal{H}}}
\newcommand{\cK}{\ensuremath{\mathcal{K}}}

\newcommand{\cP}{\ensuremath{\mathcal{P}}}

\newcommand{\cT}{\ensuremath{\mathcal{T}}}

\newcommand{\cX}{\ensuremath{\mathcal{X}}}

\newcommand{\bE}{\ensuremath{\mathbb{E}}}

\newcommand{\bN}{\ensuremath{\mathbb{N}}}
\newcommand{\bP}{\ensuremath{\mathbb{P}}}
\newcommand{\bR}{\ensuremath{\mathbb{R}}}

\newcommand{\bZ}{\ensuremath{\mathbb{Z}}}

\newcommand{\Ex}[1]{\bE \left[ #1 \right]}
\newcommand{\Exu}[2]{\underset{#1} \bE \left[ #2 \right] }

\renewcommand{\Pr}[1]{\bP \left[ #1 \right]} 
\newcommand{\Pru}[2]{\underset{ #1 }\bP \left[ #2 \right]}

\renewcommand{\exp}[1]{\mathrm{exp}\left( #1 \right)}

\newcommand{\floor}[1]{\ensuremath{\lfloor #1 \rfloor}}
\newcommand{\ceil}[1]{\ensuremath{\lceil #1 \rceil}}
\newcommand{\pmset}{\{\pm 1\}}
\newcommand{\zo}{\{0,1\}}
\newcommand{\define}{:=}
\renewcommand{\exp}[1]{\mathrm{exp}\left( #1 \right)}

\DeclareMathOperator{\supp}{supp}   
\DeclareMathOperator{\poly}{poly}
\DeclareMathOperator{\sign}{sign}

\newcommand{\dist}{\mathsf{dist}}

\newcommand{\ind}[1]{\mathbf{1} \left[ #1 \right] }
\newcommand{\inn}[1]{\langle #1 \rangle}

\newcommand{\unif}{\mathsf{unif}}
\newcommand{\bbs}{\mathsf{bbs}}
\newcommand{\block}{\mathsf{block}}
\newcommand{\coarse}{\mathsf{coarse}}
\newcommand{\blockpoint}{\mathsf{blockpoint}}
\newcommand{\TV}{\mathsf{TV}}

\title{Downsampling for Testing and Learning in Product Distributions}
\date{}
\author{Nathaniel Harms\thanks{Partly supported by NSERC\@. Much of this research was done while the
author was visiting the National Institute of Informatics, Japan.} \\ University of
Waterloo, Canada \\ \texttt{nharms@uwaterloo.ca}
\and Yuichi Yoshida\thanks{Supported by JSPS KAKENHI Grant Number JP17H04676 and 18H05291.
} \\ National Institute of Informatics, Japan \\
\texttt{yyoshida@nii.ac.jp}}
\begin{document}

\maketitle

\begin{abstract}
We study distribution-free property testing and learning problems where the unknown probability
distribution is a product distribution over $\bR^d$. For many important classes of functions,
such as intersections of halfspaces, polynomial threshold functions, convex sets, and
$k$-alternating functions, the known algorithms either have complexity that depends on the support
size of the distribution, or are proven to work only for specific examples of product distributions.
We introduce a general method, which we call \emph{downsampling}, that resolves these issues.
Downsampling uses a notion of ``rectilinear isoperimetry'' for product distributions, which further
strengthens the connection between isoperimetry, testing and learning. Using this technique, we
attain new efficient distribution-free algorithms under product distributions on $\bR^d$:
\begin{enumerate}
\item A simpler proof for non-adaptive, one-sided monotonicity testing of functions $[n]^d \to \zo$,
and improved sample complexity for testing monotonicity over unknown product distributions, from
$O(d^7)$ [Black, Chakrabarty, \& Seshadhri, SODA 2020] to $\widetilde O(d^3)$.
\item Polynomial-time agnostic learning algorithms for functions of a
constant number of halfspaces, and constant-degree polynomial threshold functions;
\item An $\exp{O(d\log(dk))}$-time agnostic learning algorithm, and an
$\exp{O(d\log(dk))}$-sample tolerant tester, for functions of $k$ convex sets; and a
$2^{\widetilde O(d)}$ sample-based one-sided tester for convex sets;
\item An $\exp{\widetilde O(k\sqrt{d})}$-time agnostic learning algorithm for
$k$-alternating functions, and a sample-based tolerant tester with the same complexity.
\end{enumerate}
\end{abstract}

\thispagestyle{empty}
\setcounter{page}{0}
\newpage

\section{Introduction}

In property testing and learning, the goal is to design algorithms that use as little information as
possible about the input while still being correct (with high probability). This includes using as
little information as possible about the probability distribution against which correctness is
measured. Information about the probability distribution could be in the form of guarantees on this
distribution (e.g.~it is guaranteed to be uniform, or Gaussian), or in the form of samples from the
distribution. So we want to minimize the requirements on this distribution, as well as the number of
samples used by the algorithm.

Progress on high-dimensional property testing and learning problems is usually made by studying
algorithms for the uniform distribution over the hypercube $\pmset^d$, or the standard Gaussian
distribution over $\bR^d$, as the simplest case. For example, efficiently learning
intersections of halfspaces is a major open problem in learning theory \cite{DKS18,KOS04}, and
progress on this problem has been made by studying the uniform distribution over the hypercube
$\pmset^d$ and the Gaussian as special cases \cite{KKMS08,KOS08,Vem10a}.
%
%
Another important example is the class of degree-$k$ polynomial threshold functions (PTFs). Unlike
intersections of halfspaces, these can be efficiently learned in the PAC model~\cite{KOS04}, but
\emph{agnostic}\footnote{See the Glossary in \cref{glossary} for standard definitions in
learning theory and property testing.} learning is more challenging. Again, progress has been made
by studying the hypercube \cite{DHK+10}. An even more extreme example is the class of convex sets,
which are not learnable in the distribution-free PAC model, because they have infinite VC dimension,
but which become learnable under the Gaussian \cite{KOS08}. The uniform distribution over the
hypercube and the Gaussian are both examples of \emph{product distributions}, so the next natural
question to ask is, can these results be generalized to any \emph{unknown} product distribution?
A partial answer was given by Blais, O'Donnell, \& Wimmer \cite{BOW10} for some of these classes; in
this paper we resolve this question.

Similar examples appear in the property testing literature. Distribution-free property testing and
testing functions with domain $\bR^d$ are emerging trends in the field (e.g.~\cite{BBBY12, DMN19,
Har19, CDS20, FY20, BFH21}). Testing monotonicity is one of the most well-studied problems in
property testing, and recent work \cite{BCS20} has extended this study to product distributions over
domain $\bR^d$. Work of Chakrabarty \& Seshadhri \cite{CS16}, Khot, Minzer, \& Safra \cite{KMS18},
and Black, Chakrabarty, \& Seshadhri \cite{BCS18,BCS20} has resulted in efficient $o(d)$-query
algorithms for the hypercube $\pmset^d$ \cite{KMS18} and the hypergrid $[n]^d$. Black, Chakrabarty,
\& Seshadhri \cite{BCS20} showed that testing monotonicity over unknown product distributions on
$\bR^d$ could be done with $\widetilde O(d^{5/6})$ queries and $O(d^7)$ samples.  Their ``domain
reduction'' method is intricate and specialized for the problem of testing monotonicity.  We
improve\footnote{An early version of this paper proved a weaker result, with two-sided error and
worse sample complexity.} the sample complexity to $\widetilde O(d^3)$ using a much simpler proof.
We also generalize the testers of \cite{CFSS17,CGG+19} for convex sets and $k$-alternating
functions, respectively, and provide new testers for arbitrary functions of convex sets.

This paper provides a general framework for designing distribution-free testing and learning
algorithms under product distributions on $\bR^d$, which may be finite or continuous. An algorithm
is distribution-free under product distributions if it does not require any prior knowledge of the
probability distribution, except the guarantee that it is a product distribution. The technique in
this paper, which we call \emph{downsampling}, improves upon previous methods (in particular,
\cite{BCS20,BOW10}), in a few ways. It is more general and does not apply only to a specific
type of algorithm \cite{BOW10} or a specific problem \cite{BCS20}, and we use it to obtain many
other results.  It is conceptually simpler. And it allows quantitative improvements over both
\cite{BOW10} and \cite{BCS20}.

\paragraph*{Organization.}
In \cref{section:results}, we describe the main results of this paper
in context of the related work. In \cref{section:techniques}, we briefly describe the main
techniques in the paper. \cref{section:downsampling} presents the definitions and lemmas required by
the main results. The following sections present the proofs of the results. For simplicity, the main
body of the paper handles only continuous product distributions; finite distributions are handled in
\cref{section:discrete distributions}.
Definitions of standard terminology in property testing and machine learning can be found in the
Glossary in Section~\ref{glossary}.

\subsection{Results}
\label{section:results}

See \cref{table:testing} for a summary of our results on property testing, and \cref{table:learning}
for a summary of our results on learning.

\subsubsection{Testing Monotonicity}

Testing monotonicity is the problem of testing whether an unknown function $f : X \to \zo$ is
monotone, where $X$ is a partial order. It is one of the most commonly studied problems in the field
of property testing. Previous work on this problem has mostly focused on uniform probability
distributions (exceptions include \cite{AC06,HK07,CDJS17,BFH21}) and finite domains.  However, there
is growing interest in property testing for functions on domain $\bR^d$
(\cite{BBBY12,DMN19,Har19,CDS20,FY20,BFH21}) and \cite{BCS20} generalized the problem to this
domain.

Testing monotonicity under product distributions has been studied a few times.
Ailon \& Chazelle \cite{AC06} gave a distribution-free monotonicity tester for real-valued functions
under product distributions on $[n]^d$, with query complexity $O(\tfrac{1}{\epsilon} d2^d \log n)$.
Chakrabarty \emph{et al.}~\cite{CDJS17} improved this to $O(\tfrac{1}{\epsilon} d \log n)$ and gave
a matching lower bound. This lower bound applies to the \emph{real-valued} case. For the
\emph{boolean-valued} case, monotonicity testers under the uniform distribution on $\pmset^d$
\cite{CS16,KMS18} and $[n]^d$ \cite{BCS18,BCS20} are known with query complexity $o(d)$. In
\cite{BCS20}, an $o(d)$-query tester was given for domain $\bR^d$. That paper showed
that there is a one-sided, non-adaptive, distribution-free monotonicity tester under product
distributions on $\bR^d$, with query complexity
$O\left(\frac{d^{5/6}}{\epsilon^{4/3}}\poly\log(d/\epsilon)\right)$ and sample complexity
$O((d/\epsilon)^7)$. In this paper we improve the sample complexity to $\widetilde
O((d/\epsilon)^3)$, while greatly simplifying the proof.

\begin{restatable}{restate-theorem}{thmmonotonicity}\label{thm:monotonicity}
There is a one-sided, non-adaptive $\epsilon$-tester for monotonicity of functions $\bR^d \to \zo$
that is distribution-free under (finite or continuous) product distributions, using
\[
O\left(\frac{d^{5/6}}{\epsilon^{4/3}}\poly\log(d/\epsilon)\right)
\]
queries and $O(\frac{d^3}{\epsilon^3}\log(d/\epsilon))$ samples.
\end{restatable}

The main result of \cite{BCS20} is a ``domain reduction'' theorem, allowing a change of domain from
$[n]^d$ to $[r]^d$ where $r = \poly(d/\epsilon)$; by applying this theorem together with their
earlier $\widetilde O(\tfrac{d^{5/6}}{\epsilon^{4/3}} \poly\log(dn))$-query tester for the uniform
distribution on $[n]^d$, they obtain a tester for monotone functions with query complexity
independent of $n$. Our result replaces this domain reduction method with a simpler and more general
2-page argument, and gives a different generalization to the distribution-free case. See
\cref{section:monotonicity} for the proofs.

\begin{table}[t]
\caption{Testing results.}\label{table:testing}
\begin{tabularx}{470pt}{ X | p{55pt} | p{60pt} | p{50pt} | X }
  & $\unif(\pmset^d)$ & $\unif([n]^d)$ & Gaussian & $\forall$ Products \\
\hline
1-Sided Testing Monotonicity \newline (Query model)
  & $\widetilde O\left(\frac{\sqrt d}{\epsilon^2}\right)$ \newline
                \cite{KMS18}
  & $\widetilde O\left(\frac{d^{5/6}}{\epsilon^{4/3}}\right)$ \newline
                 \cite{BCS20}
  & $\widetilde O\left(\frac{d^{5/6}}{\epsilon^{4/3}}\right)$ \newline
                 \cite{BCS20}
  & $\widetilde O\left(\frac{d^{5/6}}{\epsilon^{4/3}}\right)$ queries,\hfill \newline
    $\widetilde O\left(\left(\frac{d}{\epsilon}\right)^3\right)$ samples\hfill \newline
     (Thm.~\ref{thm:monotonicity}) \\
\hline
1-Sided Testing Convex Sets \newline (Sample model)
  & --
  & --
  & $\left(\frac{d}{\epsilon}\right)^{(1+o(1))d}$ \newline $2^{\Omega(d)}$ \newline
    \cite{CFSS17}
  & $\left(\frac{d}{\epsilon}\right)^{(1+o(1))d}$ (Thm.~\ref{thm:convex tester}) \\
\hline
Tolerant Testing Functions of $k$ Convex Sets \newline (Sample model)
  & --
  & --
  & --
  & $\left(\frac{dk}{\epsilon}\right)^{O(d)}$ (Thm.~\ref{thm:convex distance approximator}) \\
\hline
Tolerant Testing $k$-Alternating Functions \newline (Sample model)
  & --
  & $\left(\frac{dk}{\tau}\right)^{O\left(\frac{k\sqrt d}{\tau^2}\right)}$ \newline
    $\tau = \epsilon_2-3\epsilon_1$ \newline
    \cite{CGG+19}
  & --
  & $\left(\frac{dk}{\tau}\right)^{O\left(\frac{k\sqrt d}{\tau^2}\right)}$ \newline
    $\tau = \epsilon_2-\epsilon_1$ \newline
    (Thm.~\ref{thm:testing k-alternating})
\end{tabularx}
\end{table}

\begin{table}[t]
\caption{\emph{Learning results. All learning algorithms are agnostic except that of
\cite{Vem10a}. The PTF result for the Gaussian follows from the two cited works but is not stated in
either. All statements are informal, see references for restrictions and
qualifications. For PTFs, $\psi(k,\epsilon) \define
\min\left\{O(\epsilon^{-2^{k+1}}),
2^{O(k^2)}\left(\log(1/\epsilon)/\epsilon^2\right)^{4k+2}\right\}$.}}\label{table:learning}
\begin{tabularx}{470pt}{X | p{55pt} | p{50pt} | p{90pt} | X}
& $\unif(\pmset^d)$ & $\unif([n]^d)$ & Gaussian & $\forall$ Products \\
\hline
Functions of $k$ \newline Convex Sets
  & $\Omega(2^d)$
  & --
  & $d^{O\left(\frac{\sqrt d}{\epsilon^4}\right)}, 2^{\Omega(\sqrt d)}$ \newline
          \cite{KOS08}
  & $O\left(\frac{1}{\epsilon^2}\left(\frac{6dk}{\epsilon}\right)^d\right)$
    (Thm.~\ref{thm:learning convex sets})
           \\
\hline
Functions of $k$ \newline Halfspaces
  & $d^{O\left(\frac{k^2}{\epsilon^4}\right)}$ \newline \cite{KKMS08}
  & $(dn)^{O\left(\frac{k^2}{\epsilon^4}\right)}$ \newline \cite{BOW10}
  & $d^{O\left(\frac{\log k}{\epsilon^4}\right)},$ \newline
    $\poly\left(d, \left(\frac{k}{\epsilon}\right)^k\right)$ \newline \cite{KOS08,Vem10a}
    (Intersections only)
  & $\left(\frac{dk}{\epsilon}\right)^{O\left(\frac{k^2}{\epsilon^4}\right)}$
    (Thm.~\ref{thm:halfspaces}) \\
\hline
Degree-$k$ PTFs
  & $d^{\psi(k,\epsilon)}$ \newline \cite{DHK+10}
  & $(dn)^{\psi(k,\epsilon)}$ \newline \cite{DHK+10,BOW10}
  & $d^{\psi(k,\epsilon)}$ \newline \cite{DHK+10,BOW10}
  & $\left(\frac{dk}{\epsilon}\right)^{\psi(k,\epsilon)}$ (Thm.~\ref{thm:ptfs}) \\
\hline
$k$-Alternating \newline Functions
  & $2^{\Theta\left(\frac{k \sqrt d}{\epsilon}\right)}$ \newline \cite{BCO+15}
  & $\left(\frac{dk}{\tau}\right)^{O\left(\frac{k \sqrt d}{\tau^2} \right)}$ \newline
    (Testing) \cite{CGG+19}
  & --
  & $\left(\frac{dk}{\epsilon}\right)^{O\left(\frac{k\sqrt d}{\epsilon^2}\right)}$
    (Thm.~\ref{thm:k-alternating}) \\
\hline
\end{tabularx}
\end{table}

\subsubsection{Learning Functions of Halfspaces}

Intersections of $k$ halfspaces have VC dimension $\Theta(dk\log k)$~\cite{BEHW89,CMK19}, so the
sample complexity of learning is known, but it is not possible to efficiently find $k$ halfspaces
whose intersection is correct on the sample, unless $\mathsf{P}=\mathsf{NP}$~\cite{BR92}. Therefore
the goal is to find efficient ``improper'' algorithms that output a function other than an
intersection of $k$ halfspaces. Several learning algorithms for intersections of $k$ halfspaces
actually work for arbitrary functions of $k$ halfspaces. We will write $\cB_k$ for the set of all
functions $[k] \to \zo$, and for any class $\cF$ of functions we will write $\cB_k \circ \cF$ as the
set of all functions $x \mapsto g(f_1(x), \dotsc, f_k(x))$ where $g \in \cB_k$ and each $f_i \in
\cF$. Then for $\cH$ the class of halfspaces, Klivans, O'Donnell, \& Servedio~\cite{KOS04} gave a
(non-agnostic) learning algorithm for $\cB_k \circ \cH$ over the uniform distribution on $\pmset^d$
with complexity $d^{O(k^2/\epsilon^2)}$, Kalai, Klivans, Mansour, \& Servedio~\cite{KKMS08}
presented an agnostic algorithm with complexity $d^{O(k^2/\epsilon^4)}$ in the same setting using
``polynomial regression''\!\!.

Polynomial regression is a powerful technique, so it is important to understand when it can be
applied. Blais, O'Donnell, \& Servedio~\cite{BOW10} studied how to generalize it to arbitrary
product distributions. With their method, they obtained an agnostic learning algorithm for $\cB_k
\circ \cH$ with complexity $(dn)^{O(k^2/\epsilon^4)}$ for product distributions $X_1 \times \dotsm
\times X_d$ where each $|X_i|=n$, and complexity $d^{O(k^2/\epsilon^4)}$ for the ``polynomially
bounded'' continuous distributions. This is not a complete generalization, because, for example, on
the grid $[n]^d$ its complexity depends on $n$. This prevents a full
generalization to the domain $\bR^d$. Their algorithm also requires some prior knowledge of the
support or support size. We use a different technique and fully generalize the polynomial regression
algorithm to arbitrary product distributions. See \cref{section:halfspaces} for the proof.

\begin{restatable}{restate-theorem}{thmhalfspaces}\label{thm:halfspaces}
There is a distribution-free, improper agnostic learning algorithm for $\cB_k \circ \cH$ under
(continuous or finite) product distributions over $\bR^d$, with time complexity
\[
\min\left\{
		\left(\frac{dk}{\epsilon}\right)^{O\left(\frac{k^2}{\epsilon^4}\right)},
		O\left(\frac{1}{\epsilon^2}\left(\frac{3dk}{\epsilon}\right)^d\right)
	\right\} \,.
\]
\end{restatable}


\subsubsection{Learning Polynomial Threshold Functions}
Degree-$k$ PTFs are another generalization of halfspaces. A function $f : \bR^d \to \pmset$ is a
degree-$k$ PTF if there is a degree-$k$ polynomial $p: \bR^d \to \bR$ such that $f(x) =
\sign(p(x))$. Degree-$k$ PTFs can be PAC learned in time $d^{O(k)}$ using linear
programming~\cite{KOS04}, but agnostic learning is more challenging. Diakonikolas~\emph{et
al.}~\cite{DHK+10} previously gave an agnostic learning algorithm for degree-$k$ PTFs in the uniform
distribution over $\pmset^d$ with time complexity $d^{\psi(k,\epsilon)}$, where
\[
\psi(k,\epsilon) \define \min\left\{O(\epsilon^{-2^{k+1}}),
2^{O(k^2)}\left(\log(1/\epsilon)/\epsilon^2\right)^{4k+2}\right\} \,.
\]
The main result of that paper is an upper bound on the noise sensitivity of PTFs. Combined with the
reduction of Blais \emph{et al.}~\cite{BOW10}, this implies an algorithm for the uniform
distribution over $[n]^d$ with complexity $(dn)^{\psi(k,\epsilon)}$ and for the Gaussian
distribution with complexity $d^{\psi(k,\epsilon)}$.

Our agnostic learning algorithm for degree-$k$ PTFs eliminates the dependence on $n$ and works for
any unknown product distribution over $\bR^n$, while matching the complexity of~\cite{DHK+10} for
the uniform distribution over the hypercube.  See \cref{section:ptfs} for the proof.

\begin{restatable}{restate-theorem}{thmptfs}\label{thm:ptfs}
There is an improper agnostic learning algorithm for degree-$k$ PTFs under
(finite or continuous) product distributions over $\bR^d$, with time complexity
\[
\min\left\{ \left(\frac{kd}{\epsilon}\right)^{\psi(k,\epsilon)} \;,\;
O\left(\frac{1}{\epsilon^2}\left(\frac{9dk}{\epsilon}\right)^d\right)
\right\} \,.
\]
\end{restatable}

\subsubsection{Testing \& Learning Convex Sets}

One of the first properties (sets) of functions $\bR^d \to \zo$ to be studied in the property
testing literature is the set of indicator functions of convex sets, i.e.~functions $f : \bR^d \to
\zo$ where $f^{-1}(1)$ is convex. Write $\cC$ for this class of functions. This problem has been
studied in various models of testing \cite{Ras03, RV04, CFSS17, BMR19, BB20}. In this paper we
consider the \emph{sample-based} model of testing, where the tester receives only random examples of
the function and cannot make queries.  This model of testing has received a lot of recent attention
(e.g.~\cite{BBBY12, BMR19, BY19, CFSS17, GR16, Har19, RR21, BFH21}), partly because it matches the
standard sample-based model for learning algorithms.

Chen \emph{et al.}~\cite{CFSS17} gave a sample-based tester for $\cC$ under the Gaussian
distribution on $\bR^d$ with one-sided error and sample complexity $(d/\epsilon)^{O(d)}$, along with
a lower bound (for one-sided testers) of $2^{\Omega(d)}$. We match their upper bound while
generalizing the tester to be distribution-free under product distributions. See
\cref{section:convex sets} for proofs.

\begin{restatable}{restate-theorem}{thmconvextester}\label{thm:convex tester}
There is a sample-based distribution-free one-sided
$\epsilon$-tester for $\cC$ under (finite or continuous) product distributions that uses at most
$O\left(\left(\frac{6d}{\epsilon}\right)^d\right)$ samples.
\end{restatable}

A more powerful kind of tester is an \emph{$(\epsilon_1,\epsilon_2)$-tolerant} tester, which must
accept (with high probability) any function that is $\epsilon_1$-close to the property, while
rejecting functions that are $\epsilon_2$-far. Tolerantly testing convex sets has been studied by
\cite{BMR16} for the uniform distribution over the 2-dimensional grid, but not (to the best of our
knowledge) in higher dimensions. We obtain a sample-based tolerant tester (and distance)
approximator for convex sets in high dimension. In fact, recall that $\cB_k$ is the set of all
functions $\zo^k \to \zo$ and $\cB' \subset \cB_k$ any subset, so $\cB' \circ \cC$ is any property
of functions of convex sets. We obtain a distance approximator for any such property:

\begin{restatable}{restate-theorem}{thmconvexdistance}\label{thm:convex distance approximator}
Let $\cB' \subset \cB_k$.  There is a sample-based distribution-free algorithm under (finite or
continuous) product distributions that approximates distance to $\cB' \circ \cC$ up to additive
error $\epsilon$ using $O\left(\frac{1}{\epsilon^2}\left(\frac{3dk}{\epsilon}\right)^d \right)$
samples. Setting $\epsilon = (\epsilon_2-\epsilon_1)/2$ we obtain an
$(\epsilon_1,\epsilon_2)$-tolerant tester with sample complexity
$O\left(\frac{1}{(\epsilon_2-\epsilon_1)^2}\left(\frac{6dk}{\epsilon_2-\epsilon_1}\right)^d\right)$.
\end{restatable}

General distribution-free learning of convex sets is not possible, since this class has infinite VC
dimension. However, they can be learned under the Gaussian distribution. Non-agnostic learning under
the Gaussian was studied by Vempala \cite{Vem10a,Vem10b}. Agnostic learning under the Gaussian was
studied by Klivans, O'Donnell, \& Servedio \cite{KOS08} who presented a learning algorithm with
complexity $d^{O(\sqrt{d}/\epsilon^4)}$, and a lower bound of $2^{\Omega(\sqrt d)}$.

Unlike the Gaussian, there is a trivial lower bound of $\Omega(2^d)$ in arbitrary product
distributions, because any function $f : \pmset^d \to \zo$ belongs to this class. However, unlike
the general distribution-free case, we show that convex sets (or any functions of convex sets) can
be learned under unknown product distributions.

\begin{restatable}{restate-theorem}{thmlearningconvexsets}\label{thm:learning convex sets}
There is an agnostic learning algorithm for $\cB_k \circ \cC$ under
(finite or continuous) product distributions over $\bR^d$, with time complexity
$O\left(\frac{1}{\epsilon^2} \cdot \left(\frac{6dk}{\epsilon}\right)^d\right)$.
\end{restatable}

\subsubsection{Testing \& Learning $k$-Alternating Functions}

A $k$-alternating function $f : X \to \pmset$ on a partial order $X$ is one where for any chain $x_1
< \cdots < x_m$ in $X$, $f$ changes value at most $k$ times. Learning $k$-alternating functions on
domain $\pmset^d$ was studied by Blais \emph{et al.}~\cite{BCO+15}, motivated by the fact that these
functions are computed by circuits with few negation gates. They show that $2^{\Theta(k
\sqrt{d}/\epsilon)}$ samples are necessary and sufficient in this setting.  Canonne \emph{et
al.}~\cite{CGG+19} later obtained an algorithm for $(\epsilon_1,\epsilon_2)$-tolerant testing
$k$-alternating functions, when $\epsilon_2 > 3\epsilon_1$, in the uniform distribution over
$[n]^d$, with query complexity $(kd/\tau)^{O(k \sqrt{d} / \tau^2)}$, where $\tau =
\epsilon_2-3\epsilon_1$.

We obtain an agnostic learning algorithm for $k$-alternating functions that matches the query
complexity of the tester in~\cite{CGG+19}, and nearly matches the complexity of the (non-agnostic)
learning algorithm of~\cite{BCO+15} for the uniform distribution over the hypercube. See
\cref{section:k alternating} for proofs.

\begin{restatable}{restate-theorem}{thmkalternating}\label{thm:k-alternating}
There is an agnostic learning algorithm for $k$-alternating functions under
(finite or continuous) product distributions over $\bR^d$ that runs in time
at most
\[
\min\left\{
	\left(\frac{dk}{\epsilon}\right)^{O\left(\frac{k \sqrt d}{\epsilon^2}\right)},
	O\left(\frac{1}{\epsilon^2}\left(\frac{3kd}{\epsilon}\right)^d\right)
	\right\} \,.
\]
\end{restatable}

We also generalize the tolerant tester of \cite{CGG+19} to be distribution-free under product
distributions, and eliminate the condition $\epsilon_2 > 3\epsilon_1$.

\begin{restatable}{restate-theorem}{thmtestingkalternating}\label{thm:testing k-alternating}
For any $\epsilon_2 > \epsilon_1 > 0$, let $\tau = (\epsilon_2-\epsilon_1)/2$, there is a
sample-based $(\epsilon_1,\epsilon_2)$-tolerant tester for $k$-alternating functions using
$\left(\frac{dk}{\tau}\right)^{O\left(\frac{k\sqrt d}{\tau^2}\right)}$ samples, which is
distribution-free under (finite or continuous) product distributions over $\bR^d$.
\end{restatable}

\subsection{Techniques}\label{section:techniques}

What connects these diverse problems is a notion of rectilinear surface area or isoperimetry that we
call ``block boundary size''\!\!.  There is a close connection between learning \& testing and
various notions of isoperimetry or surface area (e.g.~\cite{CS16,KOS04,KOS08,KMS18}). We show that
testing or learning a class $\cH$ on product distributions over $\bR^d$ can be reduced to testing
and learning on the \emph{uniform} distribution over $[r]^d$, where $r$ is determined by the block
boundary size, and we call this reduction ``downsampling''\!\!. The name \emph{downsampling} is used
in image and signal processing for the process of reducing the resolution of an image or reducing
the number of discrete samples used to represent an analog signal. We adopt the name because our
method can be described by analogy to image or signal processing as the following 2-step process:
\begin{enumerate}
  \item Construct a ``digitized'' or ``pixellated'' image of the function $f : \bR^d \to \pmset$ by
    sampling from the distribution and constructing a grid in which each cell has roughly equal
    probability mass; and
	\item Learn or test the ``low-resolution'' pixellated function.
\end{enumerate}
As long as the function $f$ takes a constant value in the vast majority of ``pixels''\!\!, the low
resolution version seen by the algorithm is a good enough approximation for testing or learning.
The block boundary size is, informally, the number of pixels on which $f$ is not constant.

This technique reduces distribution-free testing and learning problems to the uniform distribution
in a way that is conceptually simpler than in the prior work \cite{BOW10,BCS20}. However, some
technical challenges remain.  The first is that it is not always easy to bound the number of
``pixels'' on which a function $f$ is not constant -- for example, for PTFs. Second, unlike in the
uniform distribution, the resulting downsampled function class on $[r]^d$ is not necessarily ``the
same'' as the original class -- for example, halfspaces on $\bR^d$ are not downsampled to halfspaces
on $[r]^d$, since the ``pixels'' are not of equal size. Thus, geometric arguments may not work,
unlike the case for actual images.

A similar technique of constructing ``low-resolution'' representations of the input has been used
and rediscovered ad-hoc a few times in the property testing literature; in prior work, it was
restricted to the uniform distribution over $[n]^d$ \cite{KR00,Ras03,FR10,BY19,CGG+19} (or the
Gaussian in \cite{CFSS17}).  With this paper, we aim to provide a unified and generalized study of
this simple and powerful technique.

\subsection{Block Boundary Size}

Informally, we define the \emph{$r$-block boundary size} $\bbs(\cH,r)$ of a class $\cH$ of functions
$\bR^d \to \zo$ as the maximum number of grid cells on which a function $f \in \cH$ is non-constant,
over all possible $r \times \dotsm \times r$ grid partitions of $\bR^d$ (which are not necessarily
evenly spaced) -- see \cref{section:downsampling} for formal definitions. Whether downsampling can
be applied to $\cH$ depends on whether
\[
  \lim_{r \to \infty} \frac{\bbs(\cH,r)}{r^d} \to 0 \,,
\]
and the complexity of the algorithms depends on how large $r$ must be for the non-constant blocks to
vanish relative to the whole $r^d$ grid. A general observation is that any function class $\cH$
where downsampling can be applied can be learned under unknown product distributions with a finite
number of samples; for example, this holds for convex sets even though the VC dimension is infinite.

\begin{proposition}[Consequence of \cref{lemma:brute force}]
Let $\cH$ be any set of functions $\bR^d \to \zo$ (measurable with respect to continuous product
distributions) such that
\[
  \lim_{r \to \infty} \frac{\bbs(\cH,r)}{r^d} = 0 \,.
\]
Then there is some function $\delta(d,\epsilon)$ such that $\cH$ is distribution-free learnable
under product distributions, up to error $\epsilon$, with $\delta(d,\epsilon)$ samples.
\end{proposition}

For convex sets, monotone functions, $k$-alternating functions, and halfspaces, $\bbs(\cH,r)$ is
easy to calculate. For degree-$k$ PTFs, it is more challenging. We say that a function $f : \bR^d
\to \zo$ induces a connected component $S$ if for every $x,y \in S$ there is a continuous curve in
$\bR^d$ from $x$ to $y$ such that $f(z)=f(x)=f(y)$ for all $z$ on the curve, and $S$ is a maximal
such set. Then we prove a general lemma that bounds the block boundary size by the number of
connected components induced by functions $f \in \cH$.

\begin{lemma}[Informal, see \cref{lemma:bbs for low-components}]
Suppose that
for any axis-aligned affine subspace $A$ of affine dimension $n \leq d$, and any
function $f \in \cH$, $f$ induces at most $k^n$ connected components in $A$. Then for $r =
\Omega(dk/\epsilon)$, $\bbs(\cH,r) \leq \epsilon \cdot r^d$.
\end{lemma}

This lemma in fact generalizes all computations of block boundary size in this paper (up to constant
factors in $r$). Using a theorem of Warren \cite{War68}, we get that any degree-$k$ polynomial
$\bR^d \to \pmset$ achieves value 0 in at most $\epsilon r^d$ grid cells, for sufficiently large $r
= \Omega(dk/\epsilon)$ (\cref{cor:ptf bbs bound}).

\subsection{Polynomial Regression}

The second step of downsampling is to find a testing or learning algorithm that works for the
uniform distribution over the (not necessarily evenly-spaced) hypergrid. Most of our learning
results use \emph{polynomial regression}.  This is a powerful technique introduced in~\cite{KKMS08}
that performs linear regression over a vector space of functions that approximately spans the
hypothesis class. This method is usually applied by using Fourier analysis to construct such an
approximate basis for the hypothesis class~\cite{BOW10,DHK+10,CGG+19}. This was the method used, for
example, by Blais, O'Donnell, \& Wimmer~\cite{BOW10} to achieve the $\poly(dn)$-time algorithms for
intersections of halfspaces.

We take the same approach but we use the Walsh basis for functions on domain $[n]^d$
(see e.g.~\cite{BRY14}) instead of the bases used in the prior works.  We show that if one can
establish bounds on the noise sensitivity in the Fourier basis for the hypothesis class restricted
to the uniform distribution over $\pmset^d$, then one gets a bound on the number of Walsh functions
required to approximately span the ``downsampled'' hypothesis class. In this way, we establish that
if one can apply standard Fourier-analytic techniques to the hypothesis class over the
\emph{uniform} distribution on $\pmset^d$ and calculate the block boundary size, then the results
for the hypercube essentially carry over to the distribution-free setting for product distributions
on $\bR^d$.

An advantage of this technique is that both noise sensitivity and block boundary size grow at most
linearly during function composition: for functions $f(x) = g(h_1(x),\dotsc,h_k(x))$ where each
$h_i$ belongs to the class $\cH$, the noise sensitivity and block boundary size grow at most
linearly in $k$. Therefore learning results for $\cH$ obtained in this way are easy to extend to
arbitrary compositions of $\cH$, which is how we get our result for intersections of halfspaces.

\section{Downsampling}
\label{section:downsampling}

We will now introduce the main definitions, notation, and lemmas required by our main results. The
purpose of this section is to establish the main conceptual component of the downsampling technique:
that functions with small enough block boundary size can be efficiently well-approximated by a
``coarsened'' version of the function that is obtained by random sampling.
See Figure~\ref{fig:block partition} for an illustration of the following definitions.

\begin{figure}
	\centering
	\includegraphics[scale=0.6]{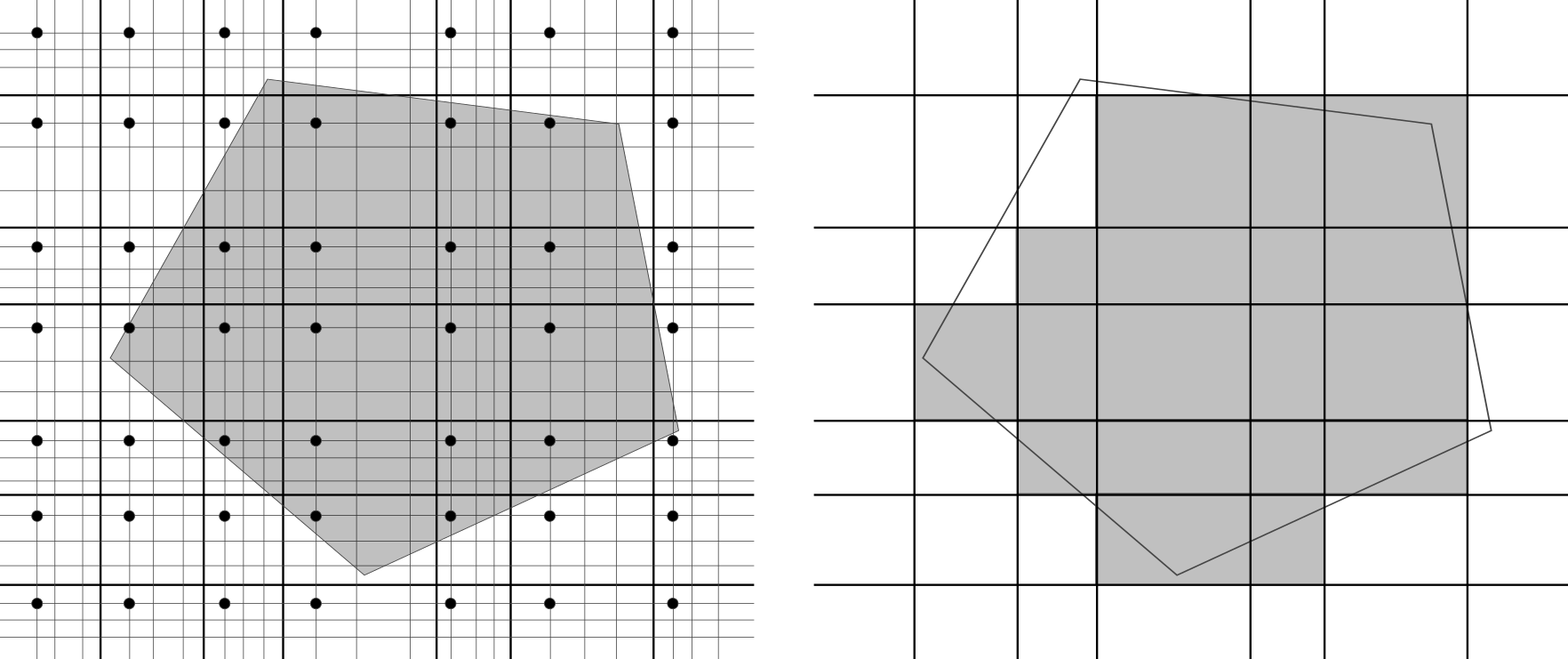}
	\caption{Left: Random grid $X$ (pale lines) with induced block partition (thick lines)
	and $\blockpoint$ values (dots), superimposed on $f^{-1}(1)$ (gray polygon).
	Right: $f^\coarse$ (grey) compared to $f$ (polygon outline).}\label{fig:block partition}
\end{figure}
\begin{definition}[Block Partitions]
An $r$-\emph{block partition} of $\bR^d$ is a pair of functions $\block
: \bR^d \to {[r]}^d$ and $\blockpoint : {[r]}^d \to \bR^d$ obtained as follows.
For each $i \in [d],j \in [r-1]$ let $a_{i,j} \in \bR$ such that $a_{i,j} <
a_{i,j+1}$ and define $a_{i,0}=-\infty, a_{i,r}=\infty$ for each $i$. For each
$i \in [d], j \in [r]$ define the interval $B_{i,j} = (a_{i,j-1},a_{i,j}]$ and a
point $b_{i,j} \in B_{i,j}$. The function $\block : \bR^d \to [r]^d$ is defined
by setting $\block(x)$ to be the unique vector $v \in [r]^d$ such that $x_i \in
B_{i,v_i}$ for each $i \in [d]$. The function $\blockpoint : [r]^d \to \bR^d$ is
defined by setting $\blockpoint(v) = (b_{1,v_1}, \dotsc, b_{d,v_d})$; note that
$\blockpoint(v) \in \block^{-1}(v)$ where $\block^{-1}(v) = \{ x \in \bR^d :
\block(x) = v\}$.
\end{definition}

\begin{definition}[Block Functions and Coarse Functions]
For a function $f : \bR^d \to \pmset$, we define $f^\block : [r]^d \to \pmset$
as $f^\block \define f \circ \blockpoint$ and $f^\coarse : \bR^d \to \bR$ as
$f^\coarse \define f^\block \circ \block = f \circ \blockpoint \circ \block$. For any set
$\cH$ of functions $\bR^d \to \pmset$, we define $\cH^\block \define \{
f^\block \;|\; f \in \cH \}$.  For a distribution $\mu$ over
$\bR^d$ and an $r$-block partition $\block : \bR^d \to [r]^d$ we define
the distribution $\block(\mu)$ over $[r]^d$ as the distribution of
$\block(x)$ for $x \sim \mu$.
\end{definition}

\begin{definition}[Induced Block Partitions]
When $\mu$ is a product
distribution over $\bR^d$, a \emph{random grid} $X$ of
length $m$ is the grid obtained by sampling $m$ points $x_1, \dotsc, x_m \in \bR^d$
independently from $\mu$ and for each $i \in [d], j \in [m]$ defining $X_{i,j}$
to be the $j^\mathrm{th}$-smallest coordinate in dimension $i$ among all sampled
points. For any $r$ that divides $m$ we define an $r$-block partition depending
on $X$ by defining for each $i \in [d], j \in [r-1]$ the point $a_{i,j} =
X_{i,mj/r}$ so that the intervals are $B_{i,j} \define
(X_{i,m(j-1)/r}, X_{i,mj/r}]$ when $j \in \{2,\dotsc,r-1\}$ and $B_{i,1} =
(-\infty, X_{i,m/r}], B_{i,r} = (X_{i,m(r-1)/r},\infty)$; we let the points
$b_{i,j}$ defining $\blockpoint$ be arbitrary. This is the $r$-block
partition \emph{induced} by $X$.
\end{definition}

\begin{definition}[Block Boundary Size]\label{def:block boundary size}
For a block partition $\block : \bR^d \to [r]^d$, a
distribution $\mu$ over $\bR^d$, and a function $f : \bR^d \to \pmset$,
we say $f$ is \emph{non-constant} on a block $v \in [r]^d$ if there are sets
$S,T \subset \block^{-1}(v)$ such that $\forall s \in S, t \in T: f(s) = 1, f(t)=-1$;
and $S,T$ have positive measure (in the product of Lebesgue measures).
For a function $f : \bR^d \to \pmset$ and a number $r$, we define the $r$-block
boundary size $\bbs(f,r)$ as the maximum number of blocks on which $f$
is non-constant, where the maximum is taken over all $r$-block
partitions $\block : \bR^d \to [r]^d$.  For a set $\cH$ of functions
$\bR^d \to \pmset$, we define $\bbs(\cH,r) \define \max\{ \bbs(f,r)
\;|\; f \in \cH\}$.
\end{definition}

The \emph{total variation distance} between two distributions $\mu,\nu$ over a
finite domain $\cX$ is defined as
\[
  \|\mu-\nu\|_\mathsf{TV} \define \frac{1}{2} \sum_{x \in \cX} |\mu(x)-\nu(x)|
  = \max_{S \subseteq \cX} |\mu(S)-\nu(S)| \,.
\]
The essence of downsampling is apparent in the next proposition. It shows that
the distance of $f$ to its coarsened version $f^\coarse$ is bounded by two
quantities: the fraction of blocks in the $r$-block partition on which $f$ is
not constant, and the distance of the distribution $\block(\mu)$ to uniform.
When both quantities are small, testing or learning $f$ can be done by
testing or learning $f^\coarse$ instead. The uniform distribution over a
set $S$ is denoted $\unif(S)$:
\begin{proposition}\label{prop:distance to coarse}
Let $\mu$ be a continuous product distribution over $\bR^d$, let $X$ be a
random grid, and let $\block : \bR^d \to [r]^d$ be the induced
$r$-block partition. Then, for any measurable $f : \bR^d \to \pmset$,
the following holds with probability 1 over the choice of $X$:
\[
  \Pru{x \sim \mu}{f(x) \neq f^\coarse(x)}
  \leq r^{-d}\cdot \bbs(f,r)
    + \|\block(\mu)-\unif([r]^d)\|_\mathsf{TV} \,.
\]
\end{proposition}
\begin{proof}
We first establish that, with probability 1 over $X$ and $x \sim \mu$,
if $f(x) \neq f^\coarse(x)$ then $f$ is non-constant on $\block(x)$.
Fix $X$ and suppose there exists a set $Z$
of positive measure such that for each $x \in Z, f(x) \neq f^\coarse(x)$
but $f$ is not non-constant on $\block(x)$, i.e.~for $V = \block^{-1}(\block(x))$, either
$\mu(V \cap f^{-1}(1)) = \mu(V)$ or $\mu(V \cap f^{-1}(-1)) = \mu(V)$. Then
there is $v \in [r]^d$ such that for $V = \block^{-1}(v)$, $\mu(Z \cap
V) > 0$.  Let $y = \blockpoint(v)$.  If $\mu(V \cap f^{-1}(f(y)) =
\mu(V)$ then $\mu(Z \cap V)=0$, so $\mu(V \cap f^{-1}(f(y)) = 0$. But
for random $X$, the probability that there exists $v \in [r]^d$ such
that $\mu(V \cap f^{-1}(\blockpoint(v))) = 0$ is 0, since
$\blockpoint(v)$ is random within $V$.

Assuming that the above event occurs,
\begin{align*}
\Pru{x \sim \mu}{f(x) \neq f^\coarse(x)}
  &\leq \Pru{x \sim \mu}{f \text{ is non-constant on } \block(x)} \\
  &\leq \Pru{v \sim [r]^d}{f \text{ is non-constant on } v}
    + \|\block(\mu)-\unif([r]^d)\|_\mathsf{TV} \,.
\end{align*}
Since $v \sim [r]^d$ is uniform, the probability of hitting a non-constant block
is at most $r^{-d} \cdot \bbs(f,r)$.
\end{proof}

Next we give a bound on the number of samples required to ensure that $\block(\mu)$ is close to
uniform. We need the following lemma.
\begin{lemma}
\label{lemma:uniform intervals}
Let $\mu$ be continuous probability distribution over $\bR$, $m,r \in \bN$ such that $r$ divides
$m$, and $\delta \in (0,1/2)$. Let $X$ be a set of $m$ points sampled
independently from $\mu$. Write $X = \{ x_1, \dotsc, x_m \}$ labeled such that $x_1 < \dotsm < x_m$
(and write $x_0 = -\infty$).  Then for any $i \in [r]$,
\[
  \Pr{ \mu\left( x_{(i-1)(m/r)}, x_{i(m/r)} \right] < \frac{1-\delta}{r} } \leq
4 \cdot e^{-\frac{\delta^2 m}{32 r}} \,.
\]
\end{lemma}
\begin{proof}
We assume that $i-1 \leq r/2$. If $i-1 > r/2$ then we can repeat the following analysis with
the opposite ordering on the points in $X$. Write $x^* = x_{(i-1)\tfrac{m}{r}}$ and $\beta =
\mu(-\infty, x^*]$. First suppose that $(1-\delta/2)\frac{i-1}{r} < \beta <
(1+\delta/2)\frac{i-1}{r} \leq (1+\delta/2)/2$; we will bound the probability of this event later.

Let $t \in \bR$ be the point such that $\mu(x^*,t] = (1-\delta)/r$ (which must exist since $\mu$ is
continuous). Let $\eta = \frac{\delta}{1-\delta} \geq \delta$.  Write $X^* = \{ x \in X : x > x^*
\}$. The expected value of $|X^* \cap (x^*,t]|$ is $|X^*| \frac{1-\delta}{r(1-\beta)} =
\left(1-\frac{i-1}{r}\right)\frac{1-\delta}{r(1-\beta)}$, where the factor $1-\beta$ in the
denominator is due to the fact that each element of $X^*$ is sampled from $\mu$ conditional on being
larger than $x^*$. The event $\mu(x^*, x_{i(m/r)}] < (1-\delta)/r$ occurs if and only if $|X^* \cap
(x^*,t]| > m/r$, which occurs with probability
\begin{align*}
  \Pr{ |X^* \cap (x^*,t]| > \frac{m}{r} }
= \Pr{ |X^* \cap (x^*,t]| > m\left(1-\frac{(i-1)}{r}\right) \frac{1-\delta}{r(1-\beta)}(1+\eta) }
\end{align*}
where
\begin{align*}
  1 + \eta
      &= \frac{(1-\beta)}{(1-\delta)\left(1-\frac{i-1}{r}\right)} 
      \geq \frac{\left(1-(1+\delta/2)\frac{i-1}{r}\right)}{(1-\delta)\left(1-\frac{i-1}{r}\right)}
      = \frac{1}{1-\delta} \left(1 - \frac{(\delta/2)(i-1)}{r-(i-1)}\right)  \\
      &\geq \frac{1 - \delta/2 }{1-\delta}
      = 1 + \frac{\delta}{2(1-\delta)} \geq 1 + \delta/2 \,.
\end{align*}
Since the expected value satisfies
\[
|X^*|\frac{1-\delta}{r(1-\beta)} \geq \frac{m}{r}(1-\frac{i-1}{r}) \frac{2(1-\delta)}{1 - \delta/2}
\geq \frac{m}{r} (1-\delta/2) \geq \frac{m}{2r}\,,
\]
the Chernoff bound gives
\[
  \Pr{ |X^* \cap (x^*,t]| > \frac{m}{r} }
  \leq \exp{-\frac{\delta^2 |X^*|(1-\delta)}{3 \cdot 4 \cdot r(1-\beta)}}
  \leq e^{-\frac{\delta^2 m}{3 \cdot 4 \cdot 2 r}} \,.
\]
Now let $t \in \bR$ be the point such that $\mu(x^*,t] = (1+\delta)/r$. The expected value of $|X^*
\cap (x^*,t]|$ is now $|X^*|\frac{1+\delta}{r(1-\beta)}$. The event $\mu(x^*, x_{i(m/r)}] >
(1+\delta)/r$ occurs if and only if $|X^* \cap (x^*,t]| < m/r$, which occurs with probability
\[
  \Pr{ |X^* \cap (x^*,t]| < \frac{m}{r} }
  = \Pr{ |X^* \cap (x^*, t]| < m\left(1-\frac{i-1}{r}\right)\frac{1+\delta}{r(1-\beta)}(1-\eta) }
\]
where
\begin{align*}
  1 - \eta
  &= \frac{1-\beta}{(1+\delta)(1-\tfrac{i-1}{r})}
  \leq \frac{1 - (1+\delta/2)\tfrac{i-1}{r}}{(1+\delta)\left(1-\frac{i-1}{r}\right)}
  = \frac{1}{1+\delta}\left(1 + \frac{(\delta/2)(i-1)}{r-(i-1)}\right) \\
  &\leq \frac{1+\delta/2}{1+\delta}
  = 1 - \frac{\delta/2}{1+\delta}
  \leq 1 - \frac{\delta}{4} \,.
\end{align*}
The expected value satisfies $|X^*| \frac{1+\delta}{r(1-\beta)} > m/r$, so the Chernoff bound gives
\[
  \Pr{ |X^* \cap (x^*,t]| < \frac{m}{r} }
  \leq \exp{-\frac{\delta^2 |X^*|(1+\delta)}{2 \cdot 4^2 \cdot r(1-\beta)}}
  \leq e^{-\frac{\delta^2 m}{2 \cdot 4^2}} \,.
\]
It remains to bound the probability that $(1-\delta/2)\frac{i-1}{r} < \beta <
(1+\delta/2)\frac{i-1}{r}$. Define $t \in \bR$ such that $\mu(-\infty, t] =
(1+\delta/2)\frac{i-1}{r}$. $\beta = \mu(-\infty,x^*] \geq (1+\delta/2)\frac{i-1}{r}$ if and only if
$x^* > t$, i.e.~$|X \cap (-\infty,t]| < \frac{i-1}{r}$.  The expected value of $|X \cap
(-\infty,t]|$ is $m\frac{(1+\delta/2)(i-1)}{r}$, so for $\eta = \frac{\delta/2}{1+\delta/2} \geq
\delta/3$, the Chernoff bound implies
\begin{align*}
  \Pr{|X \cap (-\infty,t]| < m\frac{i-1}{r}}
  &= \Pr{|X \cap (-\infty,t]| < m\frac{(1+\delta/2)(i-1)}{r}(1-\eta) } \\
  &\leq e^{-\frac{\delta^2 m(1+\delta/2)(i-1)}{18 r}}
  \leq e^{-\frac{\delta^2 m}{18 r}} \,.
\end{align*}
Now define $t \in \bR$ such that $\mu(-\infty, t] = (1-\delta/2)\frac{i-1}{r}$. $\beta =
\mu(-\infty,x^*] \leq (1-\delta/2)\frac{i-1}{r}$ if and only if $x^* < t$, i.e.~$|X \cap
(-\infty,t]| > \frac{i-1}{r}$. The expected value of $|X \cap (-\infty,t]|$ is
$m\frac{(1-\delta/2)(i-1)}{r}$, so for $\eta = \frac{\delta}{2-\delta} \geq \delta/2$,
\begin{align*}
  \Pr{|X \cap (-\infty,t]| > m\frac{i-1}{r}}
  &= \Pr{|X \cap (-\infty,t]| > m\frac{(1-\delta/2)(i-1)}{r}(1+\eta) } \\
  &\leq e^{-\frac{\delta^2 m(1-\delta/2)(i-1)}{2 \cdot 4 r}}
  \leq e^{-\frac{\delta^2 m}{4^2 r}} \,.
\end{align*}
The conclusion then follows from the union bound over these four events.
\end{proof}

\begin{lemma}
\label{lemma:continuous uniform blocks}
Let $\mu = \mu_1 \times \dotsm \times \mu_d$ be a product distribution over $\bR^d$ where each
$\mu_i$ is continuous. Let $X$ be a random grid with length $m$ sampled from $\mu$, and let $\block
: \bR^d \to [r]^d$ be the $r$-block partition induced by $X$. Then
\[
  \Pru{X}{ \|\block(\mu) - \unif([r]^d) \|_\TV > \epsilon }
    \leq 4 rd \cdot e^{-\frac{\epsilon^2 m}{18 r d^2}}
\]
\end{lemma}
\begin{proof}
For a fixed grid $X$ and each $i \in [d]$, write $p_i : [r] \to [0,1]$ be the probability
distribution on $[r]$ with $p_i(z) = \mu_i(B_{i,z})$. Then $\block(\mu) = p_1 \times \dotsm \times
p_d$.

Let $\delta = \frac{4\epsilon}{3d}$.  Suppose that for every $i,j \in [d] \times [r]$ it holds that
$\frac{1+\delta}{r} \leq p_i(j) \geq \frac{1-\delta}{r}$.  Note that $d\delta = \frac{4\epsilon}{3}
\leq \ln(1+2\epsilon) \leq 2\epsilon$. Then for every $v \in [r]^d$,
\begin{align*}
\Pru{u \sim \mu}{\block(u) = v}
&= \prod_{i=1}^d p_i(v_i) \begin{cases}
\leq {(1+\delta)}^d r^{-d}
\leq e^{d\delta} r^{-d} \leq (1+2\epsilon)r^{-d} \\
\geq {(1-\delta)}^d r^{-d}
\geq (1-d\delta) r^{-d} \geq (1-2\epsilon) r^{-d} \,.
\end{cases}
\end{align*}
So
\[
\| \block(\mu) - \unif([r]^d) \|_\TV
= \frac{1}{2} \sum_{v \in [r]^d} | \Pru{u \sim \mu}{\block(u) = v} - r^{-d} |
  \leq \frac{1}{2} \sum_{v \in [r]^d} 2\epsilon r^{-d} = \epsilon \,.
\]
By \cref{lemma:uniform intervals} and the union bound, the probability that there is some $i \in
[d], j \in [r]$ that satisfies $p_i(j) < (1-\delta)/r$ is at most
$4 rd \cdot e^{-\frac{\epsilon^2 m}{18 r d^2}}$.
\end{proof}

\section{Testing Monotonicity}
\label{section:monotonicity}

\subsection{Testing Monotonicity on the Hypergrid}\label{section:warm-up}

A good introduction to downsampling is the following short proof of the main result of Black,
Chakrabarty, \& Seshadhri~\cite{BCS20}. In an earlier work,~\cite{BCS18}, they gave an
$O((d^{5/6}/\epsilon^{4/3})\poly\log(dn))$ tester for the domain $[n]^d$, and in the later work they
showed how to reduce the domain $[n]^d$ to $[r]^d$ for $r = \poly(d/\epsilon)$.

Our monotonicity tester will use as a subroutine the following tester for \emph{diagonal} functions.
For a hypergrid $[n]^d$, a \emph{diagonal} is a subset of points $\{ x \in [n]^d : x = v + \lambda
\vec 1, \lambda \in \bZ \}$ defined by some $v \in [n]^d$. A function $f : [n]^d \to \zo$ is a
\emph{diagonal function} if it has at most one 1-valued point in each diagonal. 

\begin{lemma}
\label{lemma:diagonal tester}
There is an $\epsilon$-tester with one-sided error and query complexity
$O\left(\tfrac{1}{\epsilon}\log^2(1/\epsilon)\right)$ for diagonal functions on $[n]^d$.
\end{lemma}
\begin{proof}
\newcommand{\diag}{\mathrm{diag}}
For each $t \in [n]$ let $D_t$ be the set of diagonals with length $t$. For any $x \in [n]^d$ let
$\diag(x)$ be the unique diagonal that contains $x$. For input $f : [n]^d \to \zo$ and any $x \in
[n]^d$, let $R(x) = \frac{|\{ y \in \diag(x) : f(y) = 1 \}|}{|\diag(x)|}$.

Suppose that $f$ is $\epsilon$-far from diagonal. Then $f$ must have at least $\epsilon n^d$
1-valued points; otherwise we could set each 1-valued point to 0 to obtain the constant 0 function.
Now observe
\begin{align*}
\Exu{x \sim [n]^d}{R(x)} 
&= \Exu{x \sim [n]^d}{ \sum_{t=1}^n \sum_{L \in D_t}
  \ind{\diag(x)=L} \frac{|\{ y \in L : f(y) = 1 \}|}{t}} \\
&= \sum_{t=1}^n \sum_{L \in D_t} \Pru{x \sim [n]^d}{x \in L} \frac{|\{ y \in L : f(y) = 1\}|}{t}
= \sum_{t=1}^n \sum_{L \in D_t} \frac{t}{n^d} \frac{|\{ y \in L : f(y) = 1\}|}{t} \\
&= \frac{1}{n^d} |\{ y \in [n]^d : f(y) = 1 \}| \geq \epsilon \,.
\end{align*}
For each $i$, define $A_i = \left\{ x \in [n]^d : \frac{1}{2^i} < R(x) \leq \frac{1}{2^{i-1}}
\right\}$.  Let $k = \log(4/\epsilon)$. Then
\begin{align*}
  \epsilon
  &\leq \Ex{R(x)}
  \leq \sum_{i=1}^{\infty} \frac{|A_i|}{n^d} \max_{x \in A_i} R(x)
  \leq \sum_{i=1}^{\infty} \frac{|A_i|}{n^d 2^{i-1}}
  \leq \sum_{i=1}^{k} \frac{|A_i|}{n^d 2^{i-1}} + \sum_{i=k+1}^{\infty} \frac{1}{2^{i-1}} \\
  &\leq \sum_{i=1}^{k} \frac{|A_i|}{n^d 2^{i-1}} + \frac{1}{2^{k-1}}
  \leq \sum_{i=1}^{k} \frac{|A_i|}{n^d 2^{i-1}} + \frac{\epsilon}{2} \\
  \implies \frac{\epsilon}{2} &\leq \sum_{i=1}^{k} \frac{|A_i|}{n^d 2^{i-1}} \,.
\end{align*}
Therefore there is some $\ell \in [k]$ such that $|A_\ell| \geq \frac{\epsilon n^d 2^{\ell-1}}{2k}$.

The tester is as follows.  For each $i \in [k]$:
\begin{enumerate}
\item Sample $p = \frac{k}{\epsilon 2^{i-2}}\ln(6)$ points $x_1,
\dotsc, x_p \sim [n]^d$.
\item For each $j \in [p]$, sample $q = 2^{i+2}\ln(12)$ points $y_1, \dotsc, y_q$ from $\diag(x_i)$ and reject
if there are two distinct 1-valued points in the sample.
\end{enumerate}
The query complexity of the tester is
$\sum_{i=1}^k 4^2 \ln(6)\ln(12) \frac{k}{\epsilon 2^i} 2^i = O\left(\frac{1}{\epsilon}
\log^2(1/\epsilon)\right)$.

The tester will clearly accept any diagonal function.  Now suppose that $f$ is $\epsilon$-far from
having this property, and let $\ell \in [k]$ be such that $|A_\ell| \geq \frac{\epsilon n^d
2^{\ell-2}}{k}$. On iteration $i=\ell$, the algorithm samples $p = \frac{k}{\epsilon
2^{\ell-2}}\ln(6)$ points $x_1, \dotsc, x_p$. The probability that $\forall j \in [p], x_j \notin
A_\ell$ is at most
\[
\left(1 - \frac{|A_\ell|}{n^d}\right)^p \leq \left(1 - \frac{\epsilon 2^{\ell-2}}{k}\right)^p
  \leq \exp{-\frac{\epsilon p 2^{\ell-2}}{k}} \leq 1/6 \,.
\]
Now assume that there is some $x_j \in A_\ell$, so that $R(x_j) > 2^{-\ell}$. Let $A,B \subset
\diag(x_j)$ be disjoint subsets that partition the 1-valued points in $\diag(x_i)$ into
equally-sized parts. Then for $y$ sampled uniformly at random from $\diag(x_j)$, $\Pr{y \in A},
\Pr{y \in B} \geq 2^{-(\ell+1)}$. The probability that there are at least 2 distinct 1-valued points
in $y_1, \dotsc, y_q$ sampled by the algorithm is at least the probability that one of the first
$q/2$ samples is in $A$ and one of the last $q/2$ samples is in $B$. This fails to occur with
probability at most $2(1-2^{-(\ell+1)})^{q/2} \leq 2e^{-q2^{-(\ell+2)}} \leq 1/6$. So the total
probability of failure is at most $2/6=1/3$.
\end{proof}

\begin{theorem}\label{thm:monotonicity tester grid}
There is a non-adaptive monotonicity tester on domain $[n]^d$ with one-sided error and query
complexity $\widetilde O\left(\frac{d^{5/6}}{\epsilon^{4/3}}\right)$.
\end{theorem}
\begin{proof}
Set $r = \ceil{4d/\epsilon}$, and assume without loss of generality that $r$
divides $n$.
Partition $[n]$ into $r$ intervals $B_i = \{(i-1)(n/r) + 1, \dotsc, i(n/r)\}$. For each $v \in[r]^d$
write $B_v = B_{v_i} \times \dotsm \times B_{v_d}$. Define $\block : [n]^d \to [r]^d$ where
$\block(x)$ is the unique vector $v \in [r]^d$ such that $x \in B_v$. Define
$\block^{-\downarrow}(v) = \min\{x \in B_v\}$ and $\block^{-\uparrow}(v) = \max\{x \in B_v \}$,
where the minimum and maximum are with respect to the natural ordering on $[n]^d$. For $f : [n]^d
\to \zo$, write $f^\block : [r]^d \to \zo, f^\block(v) = f(\block^{-\downarrow}(v))$. We may
simulate queries $v$ to $f^\block$ by returning $f(\block^{-\downarrow}(v))$. We will call $v \in
[r]^d$ a \emph{boundary block} if $f(\block^{-\downarrow}(v)) \neq f(\block^{-\uparrow}(v))$.

The test proceeds as follows: On input $f : [n]^d \to \zo$ and a block $v \in [r]^d$, define the
following functions:
\begin{align*}
  g : [n]^d \to \zo, \qquad g(x) &=
      \begin{cases}
        f^\block(\block(x)) &\text{ if } \block(x) \text{ is not a boundary block} \\
        f(x)                &\text{ if } \block(x) \text{ is a boundary block.}
      \end{cases} \\
  b : [r]^d \to \zo, \qquad b(v) &=
      \begin{cases}
        0 &\text{ if } v \text{ is not a boundary block} \\
        1 &\text{ if } v \text{ is a boundary block.}
      \end{cases} \\
  h : [r]^d \to \zo, \qquad h(v) &=
      \begin{cases}
        f^\block(v) &\text{ if } v \text{ is not a boundary block} \\
        0           &\text{ if } v \text{ is a boundary block.}
      \end{cases}
\end{align*}
Queries to each of these functions can be simulated by 2 or 3 queries to $f$. The tester performs:
\begin{enumerate}
\item Test whether $g=f$, or whether $\dist(f,g) > \epsilon/4$, using $O(1/\epsilon)$ queries.
\item Test whether $b$ is diagonal, or is $\epsilon/4$-far from diagonal, using Lemma
\ref{lemma:diagonal tester}, with $O\left(\tfrac{1}{\epsilon}\log^2(1/\epsilon)\right)$ queries.
\item Test whether $h$ is monotone or $\epsilon/4$-far from monotone, using the tester of Black,
Chakrabarty, \& Seshadhri with $\widetilde O\left(\frac{d^{5/6}}{\epsilon^{4/3}}\right)$ queries.
\end{enumerate}
\begin{claim}
If $f$ is monotone, the tester passes all 3 tests with probability 1.
\end{claim}
\begin{proof}[Proof of claim]
To see that $g=f$, observe that if $v=\block(x)$ is not a boundary block then
$f(\block^{-\downarrow}(v)) = f(\block^{-\uparrow}(v))$. If $f(x) \neq f^\block(\block(x))$ then
$f(x) \neq f(\block^{-\downarrow}(v))$ and $f(x) \neq f(\block^{-\uparrow}(v))$ while
$\block^{-\downarrow}(v) \preceq x \preceq \block^{-\uparrow}(v)$, and this is a violation of the
monotonicity of $f$. Therefore $f$ will pass the first test with probability 1.

To see that $f$ passes the second test with probability 1, observe that if $f$ had 2 boundary blocks
in some diagonal, then there are boundary blocks $u,v \in [r]^d$ such that $\block^{-\uparrow}(u)
\prec \block^{-\downarrow}(v)$. But then there is $x,y \in [n]^d$ such that $\block(x)=u,
\block(y)=v$ and $f(x) = 1, f(y) = 0$; since $x \preceq \block^{-\uparrow}(u) \prec
\block^{-\downarrow}(v) \preceq y$, this contradicts the monotonicity of $f$. So $f$ has at most 1
boundary block in each diagonal.

To see that $h$ is monotone, it is sufficient to consider the boundary blocks, since all other
values are the same as $f^\block$. Let $v \in [r]^d$ be a boundary block, so there exist $x,y \in
[n]^d$ such that $\block(x)=\block(y)$ and $f(x)=1,f(y)=0$. Suppose $u \prec v$ is not a boundary
block (if it is a boundary block then $h(u)=h(v)=0$). If $h(u)=1$ then
$f(\block^{-\downarrow}(u))=1$, but $\block^{-\downarrow}(u) \prec \block^{-\downarrow}(v) \preceq
y$ while $f(\block^{-\downarrow}(u)) > f(y)$, a contradiction. So it must be that $h(u)=0$ whenever
$u \prec v$. For any block $u \in [r]^d$ such that $v \prec u$, we have $0 = h(v) \leq h(u)$, so
monotonicity holds. Since the tester of Black, Chakrabarty, \& Seshadhri has one-sided error, the
test passes with probability 1.
\end{proof}

\begin{claim}
\label{claim:monotonicity grid claim 2}
If $g$ is $\epsilon/4$-close to $f$, $b$ is $\epsilon/4$-close to diagonal, and $h$ is
$\epsilon/4$-close to monotone, then $f$ is $\epsilon$-close to monotone.
\end{claim}
\begin{proof}[Proof of claim]
Let $h^\coarse : [n]^d \to \zo$ be the function $h^\coarse(x) = h(\block(x))$. Suppose that $f(x)
\neq h^\coarse(x)$. If $v = \block(x)$ is not a boundary block of $f$ then $h^\coarse(x) = h(v) =
f^\block(v) = g(x)$, so $f(x) \neq g(x)$. If $v$ is a boundary block then $h^\coarse(x) = h(v) = 0$
so $f(x)=1$, and $b(v) = 1$.

Suppose for contradiction that there are more than $\tfrac{\epsilon}{2} r^d$ boundary blocks $v \in
[r]^d$, so there are more than $\tfrac{\epsilon}{2} r^d$ 1-valued points of $b$. Any diagonal
function  has at most $d r^{d-1}$ 1-valued points. Therefore the distance of $b$ to diagonal is at
least
\[
r^{-d}\left(\frac{\epsilon}{2} r^d - d r^{d-1}\right)
= \frac{\epsilon}{2} - \frac{d}{r}
= \frac{\epsilon}{2} - \frac{\epsilon}{4} = \frac{\epsilon}{4} \,,
\]
a contradiction. So $f$ has at most $\tfrac{\epsilon}{2} r^d$ boundary blocks. Now
\begin{align*}
\dist(f,h^\coarse)
&= \dist(f,g) + \Pru{x \sim [n]^d}{f(x)=1, \block(x) \text{ is a boundary block}}
\leq \frac{\epsilon}{4} + r^{-d} \cdot \frac{\epsilon r^d}{2} = \frac{3}{4} \epsilon \,.
\end{align*}
Let $p : [r]^d \to \zo$ be a monotone function minimizing the distance to $h$, and let $p^\coarse :
[n]^d \to \zo$ be the function $p^\coarse(x) = p(\block(x))$. Then
\[
  \dist(h^\coarse, p^\coarse) = \Pru{x \sim [n]^d}{h(\block(x)) \neq p(\block(x))} = \Pru{v \sim
[r]^d}{ h(v) \neq p(v) } \leq \epsilon / 4 \,.
\]
Finally, the distance of $f$ to the nearest monotone function is at most
\[
  \dist(f,p^\coarse) \leq \dist(f, h^\coarse) + \dist(h^\coarse,p^\coarse)
  \leq \frac{3}{4}\epsilon + \frac{1}{4} \epsilon = \epsilon \,. \qedhere
\]
\end{proof}
These two claims suffice to establish the theorem.
\end{proof}

\subsection{Monotonicity Testing for Product Distributions}

The previous section used a special case of downsampling, tailored for the uniform distribution over
$[n]^d$.  We will call a product distribution $\mu = \mu_1 \times \dotsm \times \mu_d$ over $\bR^d$
\emph{continuous} if each of its factors $\mu_i$ are continuous (i.e.~absolutely continuous with
respect to the Lebesgue measure). The proof for discrete distributions is in \cref{section:finite
algorithms}.

\thmmonotonicity*

\begin{proof}
We follow the proof of \cref{thm:monotonicity tester grid}, with some small changes. Let $r =
\ceil{16 d/\epsilon}$. The tester first
samples a grid $X$ with length $m = O\left(\frac{r d^2}{\epsilon^2}\log(rd)\right)$ and constructs
the induced $(r+2)$-block partition, with cells labeled $\{0, \dotsc, r+1\}^d$. We call a block $v
\in \{0, \dotsc, r+1\}^d$ \emph{upper extreme} if there is some $i \in [d]$ such that $v_i = r+1$,
and we call it \emph{lower extreme} if there is some $i \in [d]$ such that $v_i = 0$ but $v$ is not
upper extreme. Call the upper extreme blocks $U$ and the lower extreme blocks $L$. Note that $[r]^d
= \{0, \dotsc, r+1 \}^d \setminus (U \cup L)$.

For each $v \in [r]^d$, we again define $\block^{-\uparrow}(v), \block^{-\downarrow}(v)$ as,
respectively, the supremal and infimal point $x \in \bR^d$ such that $\block(x) = v$. The algorithm
will ignore the extreme blocks $U \cup L$, which do not have a supremal or an infimal point.
Therefore it is not defined whether these blocks are boundary blocks.

By \cref{lemma:continuous uniform blocks}, with probability at least $5/6$, we will have $\|
\block(\mu) - \unif(\{0,\dotsc,r+1\}) \|_\TV \leq \epsilon/8$.
We define $b,h$ as before, with domain $[r]^d$. Define $g$ similarly but with domain $\bR^d$ and
values
\[
  g(x) = \begin{cases}
    1    &\text{ if } \block(x) \in U \\
    0    &\text{ if } \block(x) \in L \\
    f(x) &\text{ if } \block(x) \in [n]^d \text{ is a boundary block} \\
    f^\block(\block(x)) &\text{ otherwise.}
  \end{cases}
\]
If $f$ is monotone, it may now be the case $f \neq g$, but we will have $f(x) = g(x)$ for all $x$
with $\block(x) \in [r]^d$, where the algorithm will make its queries. The algorithm will test
whether $f(x)=g(x)$ on all $x$ with $\block(x) \in [r]^d$, or $\epsilon/8$-far from this property,
which can be again done with $O(1/\epsilon)$ samples. Note that if $f$ is $\epsilon/8$-close to
having this property, then
\begin{align*}
  \dist_\mu(f,g)
  &\leq \Pru{x \sim \mu}{\block(x) \notin [n]^d} + \epsilon/8 \\
  &\leq \frac{d(r+2)^{d-1}}{(r+2)^d} + \epsilon/8 + \| \block(\mu) - \unif([r]^d \cup U \cup L)
\|_\TV \\
  &\leq  \frac{\epsilon}{16} + \frac{\epsilon}{8} + \frac{\epsilon}{4} \leq \frac{\epsilon}{2} \,.
\end{align*}
The algorithm then procedes as before, with error parameter $\epsilon/2$. To test whether $g=f$, the
algorithm samples from $\mu$ and throws away any sample $x \in \bR^d$ with $\block(x) \notin [r]^d$.
It then tests $b$ and $h$ using the uniform distribution on $[r]^d$. It suffices to prove the
following claim, which replaces \cref{claim:monotonicity grid claim 2}.

\begin{claim}
If $g$ is $\epsilon/2$-close to $f$, $b$ is $\epsilon/16$-close to diagonal, and $h$ is
$\epsilon/8$-close to monotone, then $f$ is $\epsilon$-close to monotone.
\end{claim}
\begin{proof}[Proof of claim]
Let $p : [r]^d \to \zo$ be a monotone function minimizing the distance to $h$. Then $p(v) \neq h(v)$
on at most $\frac{\epsilon r^d}{8}$ blocks $v \in [r]^d$. Define $p^\coarse : \bR^d \to \zo$ as
$p^\coarse(x) = p(\block(x))$ when $\block(x) \in [r]^d$, and $p^\coarse(x) = g(x)$ when $\block(x)
\in U \cup L$. Note that $p^\coarse$ is monotone.

By the triangle inequality,
\[
  \dist_\mu(f,p^\coarse) \leq \dist_\mu(f, g) + \dist_\mu(g, p^\coarse) \,.
\]
From above, we know $\dist_\mu(f,g) \leq \epsilon/2$. To bound the second term, observe that since
$b$ is $\epsilon/16$-close to diagonal, there are at most
\[
\frac{\epsilon}{16} r^d + dr^{d-1} \leq \frac{\epsilon}{16} r^d + \frac{d}{r} r^d \leq
\frac{\epsilon}{16} r^d + \frac{\epsilon}{16} r^d = \frac{\epsilon}{8} r^d
\]
boundary blocks. Then observe that if $g(x) \neq p^\coarse(x)$ then $\block(x) \in [r]^d$ and either
$\block(x)$ is a boundary block, or $g(x) = f^\block(\block(x)) = h(\block(x))$ and $h(\block(x))
\neq p(\block(x))$. Then
\begin{align*}
  \dist_\mu(g, p^\coarse)
  &\leq \left(\frac{1}{(r+2)^d} \sum_{v \in [r]^d}
    \ind{ v \text{ is a boundary block, or } h(v) \neq p(v) } \right) \\
  &\qquad + \|\block(\mu) - \unif(\{0,\dotsc,r+1\}^d) \|_\TV \\
  &\leq \frac{\epsilon r^d}{8 r^d} + \frac{\epsilon r^d}{8 r^d} + \frac{\epsilon}{4}
  \leq \frac{\epsilon}{2} \,. \qedhere
\end{align*}
\end{proof}
\end{proof}

\section{Learning and Testing Functions of Convex Sets}\label{section:convex sets}

In this section we present our learning and testing results for functions of $k$ convex sets: an
agnostic learning algorithm, a sample-based distance approximator, and a sample-based one-sided
tester. All our algorithms will follow from more general results that actually hold for any class
$\cH$ with bounded $r$-block boundary size; this shows that bounded block-boundary size is
sufficient to guarantee learnability in product distributions.

Let $\cC$ be the set of functions $f : \bR^d \to \pmset$ such that $f^{-1}(1)$
is convex. Let $\cB_k$ be the set of all Boolean functions $h : \pmset^k \to \pmset$.
\begin{definition}[Function Composition]
For a set $\cH$ of functions $h : \bR^d \to \pmset$,
we will define the composition $\cB_k \circ \cH$ as the set of functions of the form
$f(x) = g(h_1(x), \dotsc, h_k(x))$ where $g \in \cB_k$ and each $h_i$ belongs to $\cH$.
\end{definition}

\begin{proposition}\label{prop:resolution of composed functions}
Let $\cH$ be any class of functions $\bR^d \to \pmset$ and fix any $r$. Then
$\bbs(\cB_k \circ \cH,r) \leq k \cdot \bbs(\cH,r)$.
\end{proposition}
\begin{proof}
If $f(\cdot) = g(h_1(\cdot),\dotsc,h_k(\cdot))$ is not constant on
$\block^{-1}(v)$ then one of the $h_i$ is not constant on that block. Therefore
$\bbs(f,r) \leq \sum_{i=1}^k \bbs(h_i,r) \leq k \cdot \bbs(\cH,r)$.
\end{proof}

\begin{lemma}
For any $r$, $\bbs(\cB_k \circ \cC,r) \leq 2dkr^{d-1}$.
\end{lemma}
\begin{proof}
We prove $\bbs(\cC,r) \leq 2dr^{d-1}$ by induction on $d$; the result will hold
by Proposition~\ref{prop:resolution of composed functions}. Let $\bbs(\cC,r,d)$
be the $r$-block boundary size in dimension $d$.  Recall that block $v \in
[r]^d$ is the set $B_v = B_{1,v_1} \times \dotsm \times B_{d,v_d}$ where
$B_{i,j} = (a_{i,j-1},a_{i,j}]$ for some $a_{i,j-1} < a_{i,j}$. Let $f \in \cC$.

For $d=1$, if there are 3 intervals $B_{1,i_1},B_{1,i_2},B_{1,i_3}$, $i_1 < i_2
< i_3$, on which $f$ is not constant, then within each interval the function
takes both values $\pmset$. Thus, there are points $a \in B_{1,i_1}, b \in
B_{1,i_2}, c \in B_{1,i_3}$ such that $f(a)=1,f(b)=-1,f(c)=1$, which is a
contradiction.

For each block $B_v$, let $A_v = \{a_{1,v_1}\} \times B_{2,v_2} \times \dotsc
\times B_{d,v_d}$ be the ``upper face''.  For $d > 1$, let $P \subseteq [r]^d$
be the set of non-constant blocks $B_v$ such that $f$ is constant on the upper
face and let $Q$ be the set of non-constant blocks that are non-constant on the
upper face, so that $\bbs(f,r,d) = |P|+|Q|$. We argue that $|P| \leq 2r^{d-1}$:
for a vector $w \in [r]^{d-1}$ define the line $L_w \define \{ v \in [r]^d \;|\;
\forall i > 1, v_i=w_i\}$. If $|P \cap L_w| \geq 3$ then there are $t,u,v \in
L_w$ with $t < u < v$ such that $f$ is constant on $A_t,A_u,A_v$ but
non-constant on $B_t,B_u,B_v$. Let $x,y,z$ be points in $B_t,B_u,B_v$
respectively such that $f(x)=f(y)=f(z)=1$. If $f$ is constant $-1$ on $A_t$ or
$A_u$ then there is a contradiction since the lines through $(x,y)$ and $(y,z)$
pass through $A_t,A_u$; so $f$ is constant 1 on $A_t,A_u$. But then there is a
point $q \in A_u$ with $f(q)=-1$, which is a contradiction since it is within the
convex hull of $A_t,A_u$. So $|L_w \cap P| < 3$; since there are at most
$r^{d-1}$ lines $L_w$, $|P| \leq 2r^{d-1}$.

To bound $|Q|$, observe that for each block $v \in Q, f$ is non-constant on the
plane $\{a_{1,v_1}\} \times \bR^{d-1}$, there are $(r-1)$ such planes, $f$ is
convex on each, and the $r$-block partition induces an $r$-block partition on
the plane where $f$ is non-constant on the corresponding block. Then, by
induction $|Q| \leq (r-1) \cdot \bbs(\cC,r,d-1) \leq 2(d-1)(r-1)r^{d-2}$. So
\[
  \bbs(\cC,r,d) \leq 2\left[ (d-1)(r-1)r^{d-2} + r^{d-1} \right] < 2dr^{d-1}
\,.\qedhere
\]
\end{proof}

The above two lemmas combine to show that $r^{-d} \cdot \bbs(\cB_k \circ \cC, r)
\leq r^{-d}(2dkr^{d-1}) = 2dk/r \leq \epsilon$ when $r = \ceil{2dk/\epsilon}$.

\subsection{Sample-based One-sided Tester}
First, we prove a one-sided sample-based tester for convex sets.

\thmconvextester*

\begin{proof}
We prove the result for continuous distributions. The proof for finite distributions is in
\cref{thm:finite algorithms}.

On input distribution $\mu$ and function $f$, let $r = \ceil{6d/\epsilon}$ so that $r^{-d}\cdot
\bbs(\cC,r) \leq \epsilon/3$.
\begin{enumerate}
\item Sample a grid $X$ of size $m = O\left(\frac{rd^2}{\epsilon^2}\log(rd/\epsilon)\right)$ large
enough that Lemma~\ref{lemma:continuous uniform blocks} guarantees
$\|\block(\mu)-\unif([r]^d)\|_\TV < \epsilon/9$ with probability $5/6$.
\item Take $q = O\left(\frac{r^d}{\epsilon}\right)$ samples $Q$ and accept if there exists $h \in
\cC$ such that $f(x) = h^\coarse(x)$ on all $x \in Q$ that are not in a boundary block of $h$.
\end{enumerate}
This tester is one-sided since for any $h \in \cC$, $h(x)=h^\coarse(x)$ for all $x \in Q$ that are
not in a boundary block, regardless of whether the $r$-block
decomposition induced by $X$ satisfies $\|\block(\mu)-\unif([r]^d)\|_\TV \leq \epsilon/3$. Now
suppose that $\dist_\mu(f,\cC) > \epsilon$, and suppose that $\|\block(\mu)-\unif([r]^d)\|_\TV \leq
\epsilon$. For $h \in \cC$, let $B_h \subseteq [r]^d$ be the set of non-constant blocks. If $\exists
h \in \cC$ such that $\Pru{x \sim \mu}{h^\coarse(x) \neq f(x) \wedge \block(x) \notin B_h} <
\epsilon/9$, then
\begin{align*}
  \dist_\mu(f,h^\coarse)
    &\leq \Pru{x \sim \mu}{\block(x) \in B_h}
      + \Pru{x \sim \mu}{h^\coarse(x) \neq f(x) \wedge \block(x) \notin B_h} \\
    &\leq r^{-d}\cdot\bbs(\cC,r) + \|\block(\mu)-\unif([r]^d)\|_\TV + \frac{\epsilon}{9} \\
    &\leq \left(\frac{1}{3}+\frac{2}{9}\right)\epsilon = \frac{5}{9} \cdot \epsilon \,.
\end{align*}
Therefore
\begin{align*}
  \dist_\mu(f,h)
    &\leq \dist_\mu(f,h^\coarse) + \dist_\mu(h^\coarse,h) \\
    &\leq \dist_\mu(f,h^\coarse) + r^{-d} \cdot \bbs(\cC,r) + \|\block(\mu)-\unif([r]^d)\|_\TV \\
    &\leq \frac{5}{9}\epsilon + \frac 1 3 \cdot \epsilon + \frac 1 9 \epsilon = \epsilon \,,
\end{align*}
a contradiction. So it must be that for every $h \in \cC, \Pr{f(x) \neq h^\coarse(x) \wedge x \notin
B_h} \geq \epsilon/9$. There are at most ${r^d \choose \tfrac{\epsilon}{3} r^d} \leq
(3e/\epsilon)^{\epsilon r^d / 3}$ choices of boundary set $B$. Because the 1-valued blocks must be
the convex hull of the boundary points, for each boundary set $B$ there are at most 2 choices of
function $h^\coarse$ with boundary $B$ (with a second choice occurring when the complement of
$h^\coarse$ is also a convex set with the same boundary). Therefore, by the union bound, the
probability that $f$ is accepted is at most
\[
    \left(\frac{3e}{\epsilon}\right)^{\frac{\epsilon}{3}r^d} \cdot {\left(1-\frac\epsilon
9\right)}^q \leq e^{\epsilon\left(\frac{r^d}{3} - \frac{q}{9}\right)} \,,
\]
which is at most $1/6$ for sufficiently large $q = O\left(r^d + \frac{1}{\epsilon}\right)$.
\end{proof}

\subsection{Sample-based Distance Approximator}
Our sample-based distance approximator follows from the following general result.
\begin{lemma}
For any set $\cH$ of functions $\bR^d \to \pmset$, $\epsilon > 0$, and $r$ satisfying $r^{-d} \cdot
\bbs(\cH,r) \leq \epsilon/3$, there is a sample-based distribution-free algorithm for product
distributions that approximates distance to $\cH$ up to additive error $\epsilon$ using
$O\left(\frac{r^d}{\epsilon^2} \right)$ samples.
\end{lemma}
\begin{proof}
On input distribution $\mu$ and function $f : \bR^d \to \zo$, let $r = 3dk/\epsilon$, then:
\begin{enumerate}
\item Sample a grid $X$ of size $m = O(\frac{rd^2}{\epsilon^2}\log\frac{rd}{\epsilon})$ large enough that
Lemma~\ref{lemma:continuous uniform blocks} guarantees $\|\block(\mu)-\unif([r]^d)\|_\TV <
\epsilon/3$ with probability $5/6$.
\item Let $\cH^\coarse$ be the set of all functions $h^\coarse$ where $h \in \cH$; note that
$|\cH^\coarse| \leq 2^{r^d}$.
\item Draw $q = O\left(\frac{r^d}{\epsilon^2}\right)$ samples $Q$ and output the distance on $Q$ to
the nearest function in $\cH^\coarse$.
\end{enumerate}
We argue that with probability at least $5/6$, $\cH^\coarse$ is an $\tfrac{5}{6}\epsilon$-cover of
$\cH$.  With probability at least $5/6$, $\|\block(\mu)-\unif([r]^d)\|_\TV < \epsilon/6$. Then by
Proposition~\ref{prop:distance to coarse}, for any $h \in \cH$,
\[
  \Pru{x \sim \mu}{h(x) \neq h^\coarse(x)} \leq r^{-d} \cdot \bbs(f,r)
    + \|\block(\mu)-\unif([r]^d)\|_\TV
  \leq \left(\frac{2}{3} + \frac{1}{6}\right)\epsilon
  = \frac{5}{6}\epsilon \,,
\]
so $\cH^\coarse$ is a $\tfrac 5 6 \epsilon$-cover; assume this event occurs.

Write $\dist_Q(f,g) \define \frac{1}{q} \sum_{x \in Q}\ind{f(x) \neq g(x)}$.
By the union bound and Hoeffding's inequality, with $q$ samples we fail to get an estimate of
$\dist_\mu(f,\cH^\coarse)$ up to additive error $\frac 1 6 \epsilon$ with probability at most
\[
  |\cH^\coarse| \cdot \max_{h^\coarse \in \cH^\coarse} \Pru{Q}{ \left|\dist_\mu(f,h^\coarse) -
\dist_Q(f,h^\coarse)\right| > \frac{1}{6}\epsilon}
    \leq |\cH^\coarse| \exp{-2\frac{q \epsilon^2}{36}} < \frac{1}{6}
\]
for appropriately chosen $q = O\left(\frac{1}{\epsilon^2} \log(|\cH^\coarse|)\right) =
O\left(\frac{r^d}{\epsilon^2}\right)$.  Assume this event occurs. We want to show that
$|\dist_Q(f,\cH^\coarse)-\dist_\mu(f,\cH)| \leq \epsilon$.  Let $h \in \cH$ minimize
$\dist_\mu(f,h)$ so $\dist_\mu(f,h)=\dist_\mu(f,\cH)$. Then
\begin{align*}
\dist_Q(f,\cH^\coarse)
  &\leq \dist_Q(f,h^\coarse) 
  \leq \dist_\mu(f,h^\coarse) + \frac{\epsilon}{6} \\
  &\leq \dist_\mu(f,h) + \dist_\mu(h,h^\coarse) + \frac{\epsilon}{6} 
  \leq \dist_\mu(f,\cH) + \epsilon \,.
\end{align*}
Now let $g \in \cH$ minimize $\dist_Q(f,g^\coarse)$ so
$\dist_Q(f,g^\coarse)=\dist_Q(f,\cH^\coarse)$. Then
\begin{align*}
  \dist_Q(f,\cH^\coarse)
    &= \dist_Q(f,g^\coarse) \geq \dist_\mu(f,g^\coarse) - \frac{\epsilon}{6}
    \geq \dist_\mu(f,h^\coarse) - \frac{\epsilon}{6} \\
    &\geq \dist_\mu(f,h) - \dist_\mu(h,h^\coarse) - \frac{\epsilon}{6}
    \geq \dist_\mu(f,h) - \epsilon \,,
\end{align*}
which concludes the proof.
\end{proof}

Applying the bound on $\bbs(\cB_k,r)$ we conclude:
\thmconvexdistance*

\subsection{Agnostic Learning}
We begin our learning results with an agnostic learning algorithm for functions of $k$ convex sets:
the class $\cB_k \circ \cC$.
For a distribution $\cD$ over $\bR^d \times \pmset$ and an $r$-block partition
$\block : \bR^d \to [r]^d$, define the distribution $\cD^\block$ over $[r]^d \times \pmset$
as the distribution of $(\block(x),b)$ when $(x,b) \sim \cD$.
\begin{lemma}\label{lemma:brute force}
Let $\cH$ be any set of functions $\bR^d \to \pmset$, let $\epsilon > 0$,
and suppose $r$ satisfies ${r^{-d} \cdot \bbs(\cH,r) \leq \epsilon/3}$. Then
there is an distribution-free agnostic learning algorithm for continuous product distributions
that learns $\cH$ in $O\left(\frac{r^d+rd^2\log(rd/\epsilon)}{\epsilon^2}\right)$ samples
and time.
\end{lemma}
\begin{proof}
On input distribution $\cD$:
\begin{enumerate}
\item Sample a grid $X$ of size $m =
O(\frac{rd^2}{\epsilon^2}\log(rd/\epsilon))$ large enough that Lemma~\ref{lemma:continuous uniform blocks} guarantees
$\|\block(\mu)-\unif([r]^d)\|_\mathsf{TV} < \epsilon/3$ with probability $5/6$,
where $\block : \bR^d \to [r]^d$ is the induced $r$-block partition.
\item Agnostically learn a function $g : [r]^d \to \pmset$ with error $\epsilon/3$
	and success probability $5/6$ using $O(r^d/\epsilon^2)$ samples
	from $\cD^\block$. Output the function $g \circ \block$.
\end{enumerate}
The second step is accomplished via standard learning results
(\cite{SB14} Theorem 6.8): the number of samples required for agnostic learning is
bounded by $O(1/\epsilon^2)$ multiplied by the logarithm of the number of functions in the class,
and the number of functions $[r]^d \to
\pmset$ is $2^{r^d}$. Assume that both steps succeed, which occurs
with probability at least $2/3$. Let $f \in \cH$ minimize $\Pru{(x,b)
\sim \cD}{f(x) \neq b}$.  By Proposition~\ref{prop:distance to coarse},
 \[
\Pru{x \sim \mu}{f(x) \neq f^\coarse}
  \leq r^{-d} \cdot \bbs(f,r) + \|\block(\mu)-\unif([r]^d)\|_\mathsf{TV}
  < 2\epsilon/3 \,.
\]
Then
\begin{align*}
	\Pru{(x,b) \sim \cD}{g(\block(x)) \neq b}
	&= \Pru{(v,b) \sim \cD^\block}{g(v) \neq b}
	\leq \Pru{(v,b) \sim \cD^\block}{f^\block(v) \neq b} + \epsilon/3 \\
	&= \Pru{(x,b) \sim \cD}{f^\coarse(x) \neq b} + \epsilon/3
	< \Pru{(x,b) \sim \cD}{f(x) \neq b} + \epsilon \,. \qedhere
\end{align*}
\end{proof}

Lemma~\ref{lemma:brute force} then gives the following result
for continuous product distributions, with the result for finite distributions
following from Theorem~\ref{thm:finite algorithms}.

\thmlearningconvexsets*

\section{Learning Functions of Halfspaces}\label{section:halfspaces}
A halfspace, or linear threshold function, is any function $h : \bR^d \to \pmset$ such that for some
$w \in \bR^d, t \in \bR, h(x) = \sign(\inn{w,x}-t)$, where $\sign(z) = 1$ if $z \geq 0$ and $-1$
otherwise. Let $\cH$ be the set of halfspaces. Recall that downsampling reduces learning $\cH$ in
$\bR^d$ to learning $\cH^\block$ over $[r]^d$, and $\cH^\block$ is \emph{not} the set of halfspaces
over $[r]^d$. Fortunately, agnostically learning a halfspaces $h$ is commonly done by giving a bound
on the degree of a polynomial $p$ that approximates $h$~\cite{KOS04,KOS08,KKMS08}, and we will
show that a similar idea also suffices for learning $\cH^\block$.  We first present a general
algorithm based on ``polynomial regression''\!\!, and then introduce the Fourier analysis necessary
to apply the general learning algorithm to halfspaces, polynomial threshold functions, and
$k$-alternating functions.

\subsection{A General Learning Algorithm}
The learning algorithm in this section essentially replaces step 2 of the brute
force algorithm (Lemma~\ref{lemma:brute force}) with the ``polynomial regression''
algorithm of Kalai \emph{et al.}~\cite{KKMS08}. Our general algorithm is
inspired by an algorithm of Canonne \emph{et al.}~\cite{CGG+19} for tolerantly
testing $k$-alternating functions over the uniform distribution on $[n]^d$; we
state the regression algorithm as it appears in~\cite{CGG+19}. For a set $\cF$
of functions, $\mathsf{span}(\cF)$ is the set of all linear combinations of
functions in $\cF$:

\begin{theorem}[\cite{KKMS08, CGG+19}]\label{thm:polynomial regression}
Let $\mu$ be a distribution over $\cX$,
let $\cH$ be a class of functions $\cX \to \pmset$ and $\cF$ a collection of
functions $\cX \to \bR$ such that for every $h \in \cH$, $\exists f \in
\mathsf{span}(\cF)$ where $\Exu{x \sim \mu}{{(h(x)-f(x))}^2} \leq \epsilon^2$.
Then there is an algorithm that, for any distribution $\cD$ over $\cX \times
\pmset$ with marginal $\mu$ over $\cX$, outputs a function $g : \cX \to \pmset$
such that $\Pru{(x,b) \sim \cD}{g(x)\neq b} \leq \inf_{h \in \cH} \Pru{(x,b)
\sim \cD}{g(x) \neq b} + \epsilon$, with probability at least $11/12$, using
at most $\poly(|\cF|, 1 /\epsilon)$ samples and time.
\end{theorem}

Our general learning algorithm will apply to any hypothesis class that has
small $r$-block boundary size, and for which there is a set of functions $\cF$
that approximately span the class $\cH^\block$. This algorithm is improved
to work for finite (rather than only continuous) product distributions in
Lemma~\ref{lemma:finite general algorithm}.
\begin{lemma}\label{lemma:general algorithm}
Let $\epsilon > 0$ and let $\cH$ be a set of measurable functions $f : \bR^d \to
\pmset$ that satisfy:
\begin{enumerate}
\item There is some $r = r(d,\epsilon)$ such that $\bbs(\cH,r) \leq \frac{\epsilon}{3} \cdot r^d$;
\item There is a set $\cF$ of functions $[r]^d \to \bR$ satisfying:
$\forall f \in \cH, \exists g \in \mathsf{span}(\cF)$ such that for $v \sim
[r]^d, \Ex{(f^\block(v)-g(v))^2} \leq \epsilon^2/4$.
\end{enumerate}
Let $n = \poly(|\cF|,1/\epsilon)$ be the sample complexity of the algorithm in
Theorem~\ref{thm:polynomial regression}, with error parameter $\epsilon/2$.  Then there is an
agnostic learning algorithm for $\cH$ on continuous product distributions over $\bR^d$, that uses
$O(\max(n^2,1/\epsilon^2) \cdot rd^2\log(dr))$ samples and runs in time polynomial in the sample
size.
\end{lemma}
\begin{proof}
We will assume $n > 1/\epsilon$.  Let $\mu$ be the marginal of $\cD$ on $\bR^d$.
For an $r$-block partition, let $\cD^\block$ be the distribution of
$(\block(x),b)$ when $(x,b) \sim \cD$. We may simulate samples from $\cD^\block$
by sampling $(x,b)$ from $\cD$ and constructing $(\block(x),b)$.  The algorithm
is as follows:
\begin{enumerate}
\item Sample a grid $X$ of length $m = O(rd^2n^2\log(rd))$ large enough that
	Lemma~\ref{lemma:continuous uniform blocks} guarantees
	$\|\block(\mu)-\unif([r]^d)\|_\mathsf{TV} < 1/12n$ with
	probability $5/6$.
Construct $\block : \bR^d \to [r]^d$ induced by $X$. We may construct the
$\block$ function in time $O(dm\log m)$ by sorting, and once constructed it
takes time $O(\log r)$ to compute.
\item Run the algorithm of Theorem~\ref{thm:polynomial regression} on a sample of $n$ points from
$\cD^\block$ to learn the class $\cH^\block$; that algorithm returns a function $g : [r]^d \to
\pmset$.  Output $g \circ \block$.\label{item:general algorithm learning}
\end{enumerate}
Assume that step 1 succeeds, which occurs with probability at least $5/6$.  By
condition 2, the algorithm in step 2 is guaranteed to work on samples
$(v,b) \in [r]^d \times \pmset$ where the marginal of $v$ is $\unif([r]^d)$; let
$\cD^\unif$ be the distribution of $(v,b)$ when $v \sim \unif([r]^d)$ and $b$
is obtained by sampling $(x,b) \sim (\cD \mid x \in \block^{-1}(v))$. The algorithm
of step~\ref{item:general algorithm learning} will succeed on $\cD^\unif$; we
argue that it will also succeed on the actual input $\cD^\block$ since
these distributions are close. Observe that for samples $(v,b) \sim
\cD^\unif$ and $(\block(x),b') \sim \cD^\block$, if $v=\block(x)$ then
$b,b'$ each have the distribution of $b'$ in $(x,b') \sim (\cD \mid
\block(x)=v)$. Therefore
\begin{align*}
  \|\cD^\unif - \cD^\block\|_\mathsf{TV}
  &= \|(v,b) - (\block(x),b')\|_\mathsf{TV}
  = \|v - \block(x)\|_\mathsf{TV} \\
  &= \|\block(\mu)-\unif([r]^d)\|_\mathsf{TV} 
  < \frac{1}{12n} \,.
\end{align*}
It is a standard fact that for product distributions $P^n,Q^n,
\|P^n-Q^n\|_\mathsf{TV} \leq n \cdot \|P-Q\|_\mathsf{TV}$; using this fact,
\[
\|(\cD^\unif)^n-(\cD^\block)^n\|_\mathsf{TV}
\leq n \cdot \|\cD^\unif - \cD^\block\|_\mathsf{TV}
< \frac{1}{12} \,.
\]
We will argue that step~\ref{item:general algorithm learning} succeeds with probability $5/6$;
i.e.~that with probability $5/6$,
\[
  \Pru{(v,b) \sim \cD^\block}{g(v) \neq b} < \inf_{h \in \cH} \Pru{(v,b) \sim
\cD^\block}{h^\block(v) \neq b} + \epsilon/2 \,.
\]
Let $E(S)$ be the event that success occurs given sample $S \in ([r]^d \times
\pmset)^n$. The algorithm samples $S \sim (\cD^\block)^n$ but the success guarantee of step~\ref{item:general algorithm learning} is for $(\cD^\unif)^n$; this step will still succeed with
probability $5/6$:
\begin{align*}
  \Pru{S \sim (\cD^\unif)^n}{E(S)}
  &\geq \Pru{S \sim (\cD^\block)^n}{E(S)} - \|(\cD^\unif)^n - (\cD^\block)^n\|_\mathsf{TV} \\
  &> \Pru{S \sim \cD^n}{E(S)} - \frac{1}{12} \geq \frac{11}{12}- \frac{1}{12} = \frac{5}{6}\,.
\end{align*}
Assume that each step succeeds, which occurs with probability at least $1-2\cdot(5/6)=2/3$.
By Proposition~\ref{prop:distance to coarse}, our condition 1, and the fact that $n > 1/\epsilon$,
we have for any $h \in \cH$ that
\begin{align*}
  \Pru{x \sim \mu}{h(x) \neq h^\coarse(x)}
  \leq r^{-d}\cdot \bbs(\cH,r) + \|\block(\mu)-\unif([r]^d)\|_\mathsf{TV}
  \leq \epsilon/3 + \frac{1}{12n} < \epsilon/2 \,.
\end{align*}
The output of the algorithm is $g \circ \block$, which for any $h \in \cH$
satisfies:
\begin{align*}
	\Pru{(x,b) \sim \cD}{g(\block(x)) \neq b}
	&= \Pru{(v,b) \sim \cD^\block}{g(v) \neq b}
	\leq \Pru{(v,b) \sim \cD^\block}{h^\block(v) \neq b} + \epsilon/2 \\
	&= \Pru{(x,b) \sim \cD}{h^\coarse(x) \neq b} + \epsilon/2 \\
	&\leq \Pru{(x,b) \sim \cD}{h(x) \neq b} + \Pru{x}{h(x) \neq h^\coarse(x)} + \epsilon/2 \\
	&< \Pru{(x,b) \sim \cD}{h(x) \neq b} + \epsilon \,.
\end{align*}
Then $\Pr{g(\block(x)) \neq b} \leq \inf_{h \in \cH}{h(x) \neq b} + \epsilon$, as desired.
\end{proof}
\subsection{Fourier Analysis on \texorpdfstring{${[n]}^d$}{[n]\string^d}}
We will show how to construct a spanning set $\cF$ to satisfy condition
2 of the general learning algorithm, by using noise sensitivity and the
Walsh basis. For any $n$, let $u \sim {[n]}^d$ uniformly
at random and draw $v \in {[n]}^d$ as follows: $v_i = u_i$ with probability
$\delta$, and $v_i$ is uniform in $[n]\setminus\{u_i\}$ with probability
$1-\delta$. The noise sensitivity of functions $[n]^d \to \pmset$
is defined as:
\[
  \mathsf{ns}_{n,\delta}(f) \define \Pru{u,v}{f(u) \neq f(v)} =
  \frac{1}{2}-\frac{1}{2}\cdot\Exu{u,v}{f(u)f(v)} \,.
\]
Note that we include $n$ in the subscript to indicate the size of the domain.
We will use $\mathsf{ns}_{r,\delta}(f)$ to obtain upper bounds on the spanning
set, and we will obtain bounds on $\mathsf{ns}_{r,\delta}$ by relating it
to $\mathsf{ns}_{2,\delta}$, for which many bounds are known.
For a function $f : {[n]}^d \to \pmset$, two vectors $u,v \in {[r]}^d$, and $x
\in \pmset^d$, define ${[u,v]}^x \in {[n]}^d$ as the vector with ${[u,v]}^x_i =
u_i$ if $x_i = 1$ and $v_i$ if $x_i=-1$. Then define $f_{u,v} : \pmset^d \to
\pmset$ as the function $f_{u,v}(x) = f({[u,v]}^x)$.  The next lemma is
essentially the same as the reduction in~\cite{BOW10}.
\begin{lemma}\label{lemma:noise sensitivity reduction}
Let $\cH$ be a set of functions $f : \bR^d \to \pmset$ such that for any linear
transformation $A \in \bR^{d \times d}$, the function $f \circ A \in
\cH$, and let $\block : \bR^d \to [r]^d$ be any $r$-block partition.
Let $\mathsf{ns}_{2,\delta}(\cH) = \sup_{f \in \cH}
\mathsf{ns}_{2,\delta}(f)$ where $\mathsf{ns}_{2,\delta}(f)$ is the
$\delta$-noise sensitivity on domain $\pmset^d$.  Then
$\mathsf{ns}_{r,\delta}(f^\block) \leq \mathsf{ns}_{2,\delta}(\cH)$.
\end{lemma}
\begin{proof}
Let $u \sim [r]^d$ and let $v$ be uniform among the vectors $[r]^d$ where
$\forall i, v_i \neq u_i$. Now let $x \sim \pmset^d$ uniformly at random and let
$y$ be drawn such that $y_i = x_i$ with probability $\delta$ and $y_i = -x_i$
otherwise. Then ${[u,v]}^x$ is uniform in ${[r]}^d$, because ${[u,v]}^x_i$ is
$u_i$ or $v_i$ with equal probability and the marginals of $u_i,v_i$ are
uniform. ${[u,v]}^y_i= {[u,v]}^x_i$ with probability $1-\delta$ and is otherwise
uniform in $[r] \setminus\{{[u,v]}^x_i\}$. Let $f : {[r]}^d \to \pmset$ and
$\delta \in [0,1]$.  Let $(u',v')$ be an independent copy of $(u,v)$
and observe that $\mathsf{ns}_{r,\delta}(f^\block) = \Pr{f^\block(u') \neq
f^\block(v')}$.  Now observe that $({[u,v]}^x,{[u,v]}^y)$ has the same
distribution as $(u',v')$, so:
\begin{align*}
\Exu{u,v}{\mathsf{ns}_{2,\delta}(f_{u,v})}
&= \Exu{u,v}{\Pru{x,y\sim_\delta x}{f({[u,v]}^x) \neq f({[u,v]}^y)}} \\
&= \Exu{u,v,{(x,y)}_\delta}{\ind{f({[u,v]}^x) \neq f({[u,v]}^y)}}  \\
&= \Exu{u',v'}{\ind{f(u') \neq f(v')}} = \mathsf{ns}_{r,\delta}(f^\block) \,.
\end{align*}
For any $u,v \in [r]^d$, define the function $\Phi_{u,v} : \pmset^d \to [r]^d$
by $\Phi_{u,v}(x) = \blockpoint([u,v]^x)$. This function maps $\pmset^d$
to a set $\{b_{1,i_1}, b_{1,j_1}\} \times \dotsm \times \{b_{d,i_d},b_{d,j_d}\}$
and can be obtained by translation and scaling, which is a linear transformation.
Therefore $f_{u,v} = f \circ \Phi_{u,v}^{-1}$, so we are guaranteed that
$f_{u,v} \in \cH$. So
\[
  \mathsf{ns}_{r,\delta}(f) = \Exu{u,v}{\mathsf{ns}_{2,\delta}(f_{u,v})}
  \leq \mathsf{ns}_{2,\delta}(\cH) \,. \qedhere
\]
\end{proof}

We define the Walsh basis, an orthonormal basis of functions $[n]^d \to
\bR$; see e.g.~\cite{BRY14}.
Suppose $n=2^m$ for some positive integer $m$.
For two functions $f,g : {[n]}^d \to \bR$, define the inner product $\inn{f,g} =
\bE_{x \sim {[n]}^d}[f(x)g(x)]$. The Walsh functions $\{\psi_0, \dotsc,
\psi_m\}, \psi_i : [n] \to \pmset$ can be defined by $\psi_0 \equiv 1$ and
for $i \geq 1$, $\psi_i(z) \define (-1)^{\mathsf{bit}_i(z-1)}$ where
$\mathsf{bit}_i(z-1)$ is the $i^\mathrm{th}$ bit in the binary representation of
$z-1$, where the first bit is the least significant (see e.g.~\cite{BRY14}). It is easy
to verify that for all $i,j \in \{0, \dotsc, m\}$, if $i \neq j$ then
$\inn{\psi_i,\psi_j}=0$, and $\bE_{x \sim [n]}[\psi_i(x)] = 0$ when $i \geq 1$.
For $S \subseteq [m]$ define $\psi_S = \prod_{i \in S} \psi_i$ and note that for
any set $S \subseteq [m], S \neq \emptyset$,
\begin{equation}
\label{eq:walsh function has mean 0}
\Exu{x \sim [n]}{\psi_S(x)} = \Exu{x \sim [n]}{\prod_{i \in S} \psi_i(x)}
= \Exu{x \sim [n]}{(-1)^{\sum_{i \in S}\mathsf{bit}_i(x-1)}} = 0
\end{equation}
since each bit is uniform in $\zo$, while $\psi_\emptyset \equiv 1$.  For $S,T
\subseteq [m]$,
\[
\inn{\psi_S,\psi_T} = \bE_{x \sim [n]}[\psi_S(x)\psi_T(x)] =
\bE_x[\psi_{S \Delta T}(x)] \,,
\]
where $S \Delta T$ is the symmetric difference, so
this is 0 when $S \Delta T \neq \emptyset$ (i.e.~$S \neq T$) and 1 otherwise;
therefore $\{\psi_S : S \subseteq [m]\}$ is an orthonormal basis of functions
$[n] \to \bR$. Identify each $S \subseteq [m]$ with the number $s \in \{0,
\dotsc, n-1\}$ where $\mathsf{bit}_i(s) = \ind{i \in S}$. Now for every $\alpha
\in \{0,\dotsc,n-1\}^d$ define $\psi_\alpha : [n]^d \to \pmset$ as
$\psi_\alpha(x) = \prod_{i=1}^d \psi_{\alpha_i}(x_i)$ where $\psi_{\alpha_i}$ is
the Walsh function determined by the identity between subsets of $[m]$ and the
integer $\alpha_i \in \{0, \dotsc, n-1\}$. It is easy to verify that the set
$\{\psi_\alpha : \alpha \in {\{0,\dotsc,n-1\}}^d\}$ is an orthonormal basis.
Every function $f : {[n]}^d \to \bR$ has a unique representation $f =
\sum_{\alpha \in \{0,\dotsc,n-1\}^d} \hat f(\alpha) \psi_\alpha$ where $\hat{f}(\alpha) = \inn{f,\psi_\alpha}$.

\newcommand{\stab}{\mathsf{stab}}
For each $x \in {[n]}^d$ and $\rho \in [0,1]$ define $N_\rho(x)$ as the
distribution over $y \in {[n]}^d$ where for each $i \in [d]$, $y_i = x_i$ with
probability $\rho$ and $y_i$ is uniform in $[n]$ with probability $1-\rho$.
Define $T_\rho f(x) \define \Exu{y \sim N_\rho(x)}{f(y)}$ and $\stab_\rho(f)
\define \inn{f,T_\rho f}$. For any $\alpha \in {\{0,\dotsc,n-1\}}^d$,
\begin{align*}
  T_\rho \psi_\alpha(x)
  &= \Exu{y \sim N_\rho(x)}{\psi_\alpha(y)}
  = \Exu{y \sim N_\rho(x)}{\prod_{i=1}^d \psi_{\alpha_i}(y_i)}
  = \prod_{i=1}^d \Exu{y_i \sim N_\rho(x_i)}{\psi_{\alpha_i}(y_i)}  \\
  &= \prod_{i=1}^d \left[ \rho \psi_{\alpha_i}(x_i)
    + (1-\rho)\Exu{z \sim [n]}{\psi_{\alpha_i}(z)} \right] \,.
\end{align*}
If $\alpha_i \geq 1$ then $\Exu{z \sim [n]}{\psi_{\alpha_i}(z)} = 0$; otherwise,
$\psi_1 \equiv 1$ so $\Exu{y_i \sim N_\rho(x_i)}{\psi_0(y_i)} = 1$. Therefore
\[
  T_\rho \psi_\alpha(x) = \rho^{|\alpha|} \psi_\alpha(x) \,,
\]
where $|\alpha|$ is the number of nonzero entries of $\alpha$; so
$\widehat{T_\rho f}(\alpha) = \inn{\psi_\alpha,T_\rho f} = \inn{T_\rho
\psi_\alpha, f} = \rho^{|\alpha|} \widehat f(\alpha)$. Since $T_\rho$ is a
linear operator,
\[
  \stab_\rho(f) = \inn{f,T_\rho f} = \sum_\alpha \rho^{|\alpha|} \hat{f}{(\alpha)}^2 \,.
\]
Note that for $f : \pmset^d \to \pmset$, $\stab_\rho(f)$ is the usual notion of
stability in the analysis of Boolean functions.
\begin{proposition}
For any $f : {[n]}^d \to \pmset$ and any $\delta,\rho \in [0,1]$,
$\mathsf{ns}_{n,\delta}(f) = \frac{1}{2} - \frac{1}{2} \cdot
\stab_{1-\frac{n}{n-1}\delta}(f)$.
\end{proposition}
\begin{proof}
For $v \sim N_\rho(u), v_i = u_i$ with probability $\rho +
\frac{1-\rho}{n}$, so in the definition of noise sensitivity, $v$ is distributed
as $N_\rho(u)$ where $(1-\delta) = \rho +\frac{1-\rho}{n}$, i.e.~$\delta = 1 -
\rho - \frac{1-\rho}{n} = (1-1/n) - \rho(1-1/n) = (1-\rho)(1-1/n)$; or, $\rho =
1 - \frac{n}{n-1}\delta$. By rearranging, we arrive at the conclusions.
\end{proof}
\begin{proposition}\label{prop:noise sensitivity to fourier degree}
For any $f : {[n]}^d \to \bR$ and $t = \frac{2}{\delta}$, $\sum_{\alpha : |\alpha|
\geq t} \hat f{(\alpha)}^2 \leq 2.32 \cdot \mathsf{ns}_{n,\delta}(f)$.
\end{proposition}
\begin{proof}
Following~\cite{KOS04}:
\begin{align*}
2\mathsf{ns}_{n,\delta}(f)
&= 1-\sum_\alpha {\left(1-\frac{n}{n-1}\delta\right)}^{|\alpha|} \hat f{(\alpha)}^2
\geq 1 - \sum_\alpha {\left(1-\delta\right)}^{|\alpha|} \hat f{(\alpha)}^2 \\
&= \sum_\alpha (1-{\left(1-\delta\right)}^{|\alpha|}) \hat f{(\alpha)}^2  
\geq \sum_{\alpha:|\alpha|\geq 2/\delta}
  (1-{\left(1-\delta\right)}^{|\alpha|}) \hat f{(\alpha)}^2 \\
&\geq \sum_{\alpha:|\alpha|\geq 2/\delta}
  (1-{\left(1-\delta\right)}^{2/\delta}) \hat f{(\alpha)}^2
\geq (1-e^{-2})\sum_{\alpha:|\alpha|\geq 2/\delta} \hat f{(\alpha)}^2 \,.
\end{align*}
The result now holds since $2/(1-e^{-2}) < 2.32$.
\end{proof}

\begin{lemma}\label{lemma:noise sensitivity to polynomial regression}
Let $\cH$ be a set of functions $[n]^d \to \pmset$ where $n$ is a power of 2,
let $\epsilon,\delta > 0$ such that $\forall h \in \cH,
\mathsf{ns}_{n,\delta}(h) \leq \epsilon^2/3$, and let $t = \ceil{\frac{2}{\delta}}$.
Then there is a set $\cF$ of functions ${[n]}^d \to \bR$ of size $|\cF| \leq
{(nd)}^t$, such that that for any $h \in \cH$, there is a function $p \in
\mathsf{span}(\cF)$ where $\Ex{{(h(x)-p(x))}^2} \leq \epsilon^2$.
\end{lemma}
\begin{proof}
Let $p = \sum_{|\alpha| < t} \hat f(\alpha) \phi_\alpha$.
Then by Proposition~\ref{prop:noise sensitivity to fourier degree},
\begin{align*}
  \Ex{{(p(x)-f(x))}^2}
  &= \Ex{{\left(\sum_{|\alpha| \geq t} \hat f(\alpha) \phi_\alpha(x)\right)}^2}
  = \Ex{\sum_{|\alpha| \geq t}\sum_{|\beta| \geq t}
    \hat f(\alpha) \hat f(\beta) \phi_\alpha(x) \phi_\alpha(x) } \\
  &= \sum_{|\alpha| \geq t}\sum_{|\beta| \geq t}
    \hat f(\alpha) \hat f(\beta) \inn{\phi_\alpha,\phi_\beta}
  = \sum_{|\alpha| \geq t} \hat f{(\alpha)}^2 \leq \epsilon^2 \,.
\end{align*}
Therefore $p$ is a linear combination of functions $\phi_\alpha = \prod_{i=1}^d
\phi_{\alpha_i}$ where at most $t$ values $\alpha_i \in \{0, \dotsc, n-1\}$ are
not 0. There are at most ${((n-1)d)}^t$ such products since for each
non-constant $\phi_{\alpha_i}$ we choose $i \in [d]$ and $\alpha_i \in [n-1]$.
We may take $\cF$ to be the set of these products.
\end{proof}

\subsection{Application}
To apply Lemma~\ref{lemma:general algorithm}, we must give bounds on
$\bbs(\cB_k \circ \cH,r)$ and the noise sensitivity:

\begin{lemma}\label{lemma:halfspace block boundary}
Fix any $r$. Then $\bbs(\cB_k \circ \cH,r) \leq dkr^{d-1}$.
\end{lemma}
\begin{proof}
Any halfspace $h(x) = \sign(\inn{w,x}-t)$ is unate, meaning there is a vector
$\sigma \in \pmset^d$ such that the function $h^\sigma \define
\sign(\inn{w,x^\sigma}-t)$, where $x^\sigma = (\sigma_1x_1, \dotsc,
\sigma_d x_d)$, is monotone. For any $r$-block partition $\block : \bR^d \to
[r]^d$ defined by values $\{a_{i,j}\}$ for $i \in [d], j \in [r-1]$, we can
define $\block^\sigma : \bR^d \to [r]^d$ as the block partition obtained by
taking $\{\sigma_i \cdot a_{i,j}\}$. The number
of non-constant blocks of $h$ in $\block$ is the same as that of $h^\sigma$ in
$\block^\sigma$, but $h^\sigma$ is monotone. Thus the bound on $\bbs$ for
monotone functions holds, so $\bbs(\cH,r) \leq dr^{d-1}$ by
Lemma~\ref{lemma:monotone block boundary}, and Lemma~\ref{prop:resolution of
composed functions} gives $\bbs(\cB_k \circ \cH,r) \leq dkr^{d-1}$.
\end{proof}

The bounds on noise sensitivity follow from known results for the hypercube.
\begin{proposition}\label{prop:noise sensitivity composition}
Let $h_1, \dotsc, h_k : {[n]}^d \to \pmset$ and let $g : \pmset^k \to \pmset$.
Let $f \define g \circ (h_1, \dotsc,h_k)$. Then $\mathsf{ns}_\delta(f) \leq
\sum_{i=1}^k \mathsf{ns}_\delta(h_i)$.
\end{proposition}
\begin{proof}
For $u,v$ drawn from ${[n]}^d$ as in the definition of noise sensitivity, the
union bound gives
$\mathsf{ns}_\delta(f) = \Pru{u,v}{f(u) \neq f(v)}
  \leq \Pru{u,v}{\exists i : h_i(u) \neq h_i(v)}
  \leq \sum_{i=1}^k \mathsf{ns}_\delta(h_i)$.
\end{proof}

\begin{lemma}\label{lemma:halfspace noise sensitivity}
Let $f = g \circ (h_1, \dotsc, h_k) \in \cB_k \circ \cH$. For any $r$-block
partition $\block : \bR^d \to [r]^d$, and any $\delta \in [0,1],
\mathsf{ns}_{r,\delta}(f^\block) = O(k \sqrt \delta)$.
\end{lemma}
\begin{proof}
It is known that $\mathsf{ns}_{2,\delta}(\cH) = O(\sqrt \delta)$ (Peres'
theorem~\cite{OD14}). Let $A$ be any full-rank linear transformation and let $h
\in \cH$, $h \circ A \in \cH$.  This holds since for some $w \in \bR^d, t \in
\bR$, $h(Ax) = \sign(\inn{w,Ax}-t) = \sign(\inn{Aw,x}-t)$, which is a halfspace.
Then Lemma~\ref{lemma:noise sensitivity reduction} implies
$\mathsf{ns}_{r,\delta}(h^\block) \leq \mathsf{ns}_{2,\delta}(\cH) = O(\sqrt
\delta)$ and we conclude with Proposition~\ref{prop:noise sensitivity
composition}.
\end{proof}

\thmhalfspaces*
\begin{proof}
Here we prove only the continuous distribution case. The finite case is proved in \cref{thm:finite
algorithms}.

For $r = \ceil{dk/\epsilon}$, we have $r^{-d}\cdot \bbs(\cB_k \circ \cH,r) \leq
\epsilon$ by Lemma~\ref{lemma:halfspace block boundary}, so condition 1 of
Lemma~\ref{lemma:general algorithm} holds. Lemma~\ref{lemma:halfspace noise
sensitivity} guarantees that for any $f \in \cB_k \circ \cH,
\mathsf{ns}_{r,\delta}(f^\block) = O(k \sqrt \delta)$. Setting $\delta =
\Theta(\epsilon^4/k^2)$ so that $\mathsf{ns}_{r,\delta}(f^\block) \leq
\epsilon^2/3$, we obtain via Lemma~\ref{lemma:noise sensitivity to polynomial
regression} a set $\cF$ of size $|\cF| \leq (rd)^{O(k^2/\epsilon^4)}$ satisfying
condition 2 of Lemma~\ref{lemma:general algorithm}. Then for $n =
\poly(|\cF|,1/\epsilon)$ we apply Lemma~\ref{lemma:general algorithm} to get an
algorithm with sample complexity
\[
  O\left(rd^2n^2\log(rd)\right)
  = O\left(\frac{d^3k}{\epsilon}\log(dk/\epsilon)\right)
  \cdot
\left(\frac{dk}{\epsilon}\right)^{O\left(\frac{k^2}{\epsilon^4}\right)} \,.
\]
The other time complexity follows from Lemma~\ref{lemma:brute force}.
\end{proof}

\section{Learning Polynomial Threshold Functions}\label{section:ptfs}
A \emph{degree-$k$ polynomial threshold function (PTF)} is a function $f : \bR^d
\to \pmset$ such that there is a degree-$k$ polynomial $p : \bR^d \to \bR$ in
$d$ variables where $f(x) = \sign(p(x))$; for example, a halfspaces is a
degree-1 PTF\@. Let $\cP_k$ be the set of degree-$k$
PTFs. As for halfspaces, we will give bounds on the noise sensitivity
and block boundary size and apply the general learning algorithm. The bound on
noise sensitivity will follow from known results on the hypercube~\cite{DHK+10},
but the bound on the block boundary size is much more difficult to obtain than
for halfspaces.

\subsection{Block-Boundary Size of PTFs}

A theorem of Warren~\cite{War68} gives a bound on the number of connected
components of $\bR^d$ after removing the 0-set of a degree-$k$ polynomial. This
bound (Theorem~\ref{thm:warren} below) will be our main tool.

\newcommand{\comp}{\mathsf{comp}}
A set $S \subseteq \bR^d$ is \emph{connected}\footnote{Here we are using the
fact that \emph{connected} and \emph{path-connected} are equivalent in $\bR^d$.}
if for every $s,t \in S$ there is a continuous function $p : [0,1] \to S$ such
that $p(0)=s,p(1)=t$. A subset $S \subseteq X$ where $X \subseteq \bR^d$ is a
\emph{connected component} of $X$ if it is connected and there is no connected
set $T \subseteq X$ such that $S \subseteq T$. Write $\comp(X)$ for the number of
connected components of $X$.

A function $\rho : \bR^d \to (\bR \cup \{*\})^d$ is called a \emph{restriction}
and we will denote $|\rho| = |\{i \in [d] : \rho(i) = *\}|$. The affine subspace
$A_\rho$ induced by $\rho$ is $A_\rho \define \{ x \in \bR^d \;|\; x_i = \rho(i)
\text{ if } \rho(i) \neq *\}$ and has affine dimension $|\rho|$.

For $n \leq d$, let $\cA_n$ be the set of affine subspaces
$A_\rho$ obtained by choosing a restriction $\rho$ with $\rho(i) = *$ when $i
\leq n$ and $\rho(i) \neq *$ when $i > n$, so in particular $\cA_d = \{\bR^d\}$.

Let $f : \bR^d \to \pmset$ be a measurable function and define the boundary of
$f$ as:
\[
  \partial f \define \{ x \in \bR^d \;|\; \forall \epsilon > 0, \exists y :
\|x-y\|_2 < \epsilon, f(y) \neq f(x) \} \,.
\]
This is equivalent to the boundary of the set of $+1$-valued points, and the
boundary of any set is closed.  Each measurable $f : \bR^d \to \pmset$ induces a
partition of $\bR^d \setminus
\partial f$ into some number of connected parts.
For a set $\cH$ of functions $f : \bR^d \to \pmset$ and $n \leq d$, write
\newcommand{\parts}{\mathsf{parts}}
\[
M(n) \define
  \max_{f \in \cH} \max_{A \in \cA_n} \comp(A \setminus \partial f) \,.
\]
For each $i \in [d]$ let $\cP_i$ be the set of hyperplanes of the form $\{x \in
\bR^d \;|\; x_i = a\}$ for some $a \in \bR$.
An \emph{$(r,n,m)$-arrangement} for $f$ is any set $A \setminus \left(\partial f
\cup \bigcup_{i=1}^m \bigcup_{j=1}^{r-1} H_{i,j}\right)$ where $A \in \cA_n$ and
$H_{i,j} \in \cP_i$ such that all $H_{i,j}$ are distinct. Write $R_f(r,n,m)$ for
the set of $(r,n,m)$-arrangements for $f$. Define
\[
  P_r(n,m) \define \max_{f \in \cH} \max\{ \comp(R) \;|\; R \in R_f(r,n,m) \}
\]
and observe that $P_r(n,0) = M(n)$.

\begin{proposition}\label{prop:bbs bound}
For any set $\cH$ of functions $f : \bR^d \to \pmset$ and any $r > 0$,
$\bbs(\cH, r) \leq P_r(d,d) - r^d$.
\end{proposition}
\begin{proof}
Consider any $r$-block partition, which is obtained by choosing values $a_{i,j}
\in \bR$ for each $i \in [d], j \in [r-1]$ and defining $\block : \bR^d \to
[r]^d$ by assigning each $x \in \bR^d$ the block $v \in [r]^d$ where $v_i$ is
the unique value such that $a_{i,v_i-1} < x_i \leq a_{i,v_i}$, where we define
$a_{i,0} = -\infty, a_{i,r} = \infty$. Suppose $v$ is a non-constant block, so
there are $x,y \in \block^{-1}(v) \setminus \partial f$ such that $f(x) \neq f(y)$.
Let $H_{i,j} = \{ x \in \bR^d \;|\; x_i = a_{i,j}\}$ and let $B = \partial f
\cup \bigcup_{i,j} H_{i,j}$. Consider the set $\bR^d \setminus B$. Since $x \notin \partial f$
there exists some small open ball $R_x$ around $x$ such that $\forall x' \in R_x, f(x') = f(x)$;
and since $x \in \block^{-1}(v)$, $R_x \cap \block^{-1}(v)$ is a set of positive
Lebesgue measure. Since $B$ has Lebesgue measure 0, we conclude that
$(R_x \cap \block^{-1}(v)) \setminus B$ has positive measure, so there
is $x' \in \block^{-1}(v) \setminus B$ with $f(x')=f(x)$. Likewise, there
is $y' \in \block^{-1}(v) \setminus B$ with $f(y')=f(y) \neq f(x')$. Therefore
$x',y'$ must belong to separate components, so $\block^{-1}(v)\setminus B$ is
partitioned into at least 2 components.  Meanwhile, each constant block
is partitioned into at least 1 component. So
\[
  P_r(d,d) \geq 2 \cdot (\#\text{ non-constant blocks}) + (\#\text{ constant
blocks}) = \bbs(\cH,r) + r^d \,. \qedhere
\]
\end{proof}

The following fact must be well-known, but not to us:
\begin{proposition}\label{prop:cutting bound}
Let $A$ be an affine subspace of $\bR^d$,
let $B \subset A$, and for $a \in \bR$ let $H = \{ x \in \bR^d \;|\; x_1 = a
\}$. Then
\[
  \comp(A \setminus (H\cup B)) - \comp(A \setminus B)
  \leq \comp(H \setminus B) \,.
\]
\end{proposition}
\begin{proof}
Let $G$ be the graph with its vertices $V$ being the components of $A \setminus
(H \cup B)$ and the edges $E$ being the pairs $(S,T)$ where $S,T$ are components
of $A \setminus (H \cup B)$ such that $\forall s \in S, s_1 < a$, $\forall t \in
T, t_1 > a$, and there exists a component $U$ of $A \setminus B$ such that $S,T
\subset U$. Clearly $\comp(A \setminus (H \cup B)) = |V|$; we will show that
$\comp(A \setminus B)$ is the number of connected components of $G$ and that
$|E| \leq \comp(H \setminus B)$. This suffices to prove the statement. We will
use the following claim:
\begin{claim}
Let $U$ be a connected component of $A \setminus B$.  If $S,T \in V$ and there
is a path $p : [0,1] \to U$ such that $p(0) \in S, p(1) \in T$ and either
$\forall \lambda,p(\lambda)_1 \leq a$ or $\forall \lambda, p(\lambda)_1 \geq a$,
then $S=T$.
\end{claim}
\begin{proof}[Proof of claim]
Assume without loss of generality that $p(\lambda)_1 \leq a$ for all $\lambda$.
Let $P = \{ p(\lambda) \;|\; \lambda \in [0,1]\}$. Since $U$ is open we can
define for each $\lambda$ a ball $B(\lambda) \ni p(\lambda)$ such that
$B(\lambda) \subset U$. Consider the sets $B_a(\lambda) \define \{ x \in
B(\lambda) \;|\; x_1 < a \}$, which are open, and note that for all
$\alpha,\beta \in [0,1]$, if $p(\alpha) \in B(\beta)$ then $B_a(\alpha)\cap
B_a(\beta) \neq \emptyset$ since $p(\alpha)_1,p(\beta)_1 \leq a$.

Assume for contradiction that there is $\lambda$ such that $B_a(\lambda)$ is not
connected to $S$ or $T$; then let $\lambda'$ be the infimum of all such
$\lambda$, which must satisfy $\lambda' > 0$ since $p(0) \in S$. For any
$\alpha$, if $p(\alpha) \in B(\lambda')$ and $B_a(\alpha)$ is connected to $S$
or $T$ then since $B_a(\lambda') \cap B_a(\alpha)$ it must be that
$B_a(\lambda)$ is connected as well; therefore $B_a(\alpha)$ is not connected to
either $S$ or $T$. But since $p$ is continuous, there is $\alpha < \lambda'$
such that $p(\alpha) \in B(\lambda')$, so $\lambda'$ cannot be the infimum,
which is a contradiction. Therefore every $\lambda$ has $B_a(\lambda)$ connected
to either $S$ or $T$. If $S \neq T$, this is a contradiction since there must
then be $\alpha,\beta$ such that $p(\alpha) \in B(\beta)$ but
$B_a(\alpha),B_a(\beta)$ are connected to $S,T$ respectively.  Therefore $S=T$.
\end{proof}
We first show that $\comp(A \setminus B)$ is the number of graph-connected
components of $G$.  Suppose that vertices $(S,T)$ are connected, so there is a
path $S = S_0, \dotsc, S_n = T$ in $G$. Then there are connected components
$U_i$ of $A \setminus B$ such that $S_{i-1},S_i \subset U_i$; so $S_i \subset
U_i \cap U_{i+1}$, which implies that $\bigcup_i U_i \subset A \setminus B$ is
connected.  Therefore we may define $\Phi$ as mapping each connected component
$\{S_i\}$ of $G$ to the unique component $U$ of $A \setminus B$ with $\bigcup_i
S_i \subset U$. $\Phi$ is surjective since for each component $U$ of $A
\setminus B$ there is some vertex $S$ (a component of $A \setminus (H \cup B)$)
such that $S \subseteq U$: this is $U$ itself if $U \cap H = \emptyset$. For
some connected component $U$ of $A \setminus B$, let $S,T \subseteq U$ be
vertices of $G$, and let $s \in S, t \in T$; since $U$ is connected, there is a
path $p : [0,1] \to U$ such that $s=p(0),t=p(1)$. Let $S=S_0, \dotsc, S_n = T$
be the multiset of vertices such that $\forall \lambda \in [0,1], \exists i :
p(\lambda) \in S_i$; let $\psi(\lambda) \in \{0, \dotsc, n\}$ be the index such
that $p(\lambda) \in S_{\psi(\lambda)}$, and order the sequence such that if
$\alpha < \beta$ then $\psi(\alpha) \leq \psi(\beta)$ (note that we may have
$S_i = S_j$ for some $i < j$ if $p(\lambda)$ visits the same set more than
once). Then for any $i$, $S_i,S_{i+1} \subseteq U$ since the path visits both
and is contained in $U$. If $S_i,S_{i+1}$ are on opposite sides of $H$, then
there is an edge $(S_i,S_{i+1})$ in $G$; otherwise, the above claim implies
$S_i=S_{i+1}$. Thus there is a path $S$ to $T$ in $G$; this proves that $\Phi$
is injective, so $\comp(A \setminus B)$ is indeed the number of graph-connected
components of $G$.

Now let $(S,T) \in E$, so there is a component $U$ of $A \setminus B$ such that
$S,T \subset U$. For any $s \in S, t \in T$ there is a continuous path $p_{s,t}
: [0,1] \to U$ where $p_{s,t}(0)=s,p_{s,t}(1)=t$. There must be some $z \in
[0,1]$ such that $p_{s,t}(z) \in H$, otherwise the path is a path in $\bR^d
\setminus B$ and $S=T$. Since $p_{s,t}(z) \in H \cap U$, so $p_{s,t}(z) \notin
B$, there is some component $Z \in \cC_H$ containing $p_{s,t}(z)$. We will map
the edge $(S,T)$ to an arbitrary such $Z$, for any choice $s,t,z$, and show that
it is injective.  Suppose that $(S,T),(S',T')$ map to the same $Z \in \cC_H$.
Without loss of generality we may assume that $S,S'$ lie on the same side of $H$
and that $\forall x \in S\cup S', x_1 < a$. Then there are $s \in S, s' \in S',
t \in T, t' \in T'$, and $z,z' \in [0,1]$ such that $p_{s,t}(z),p_{s',t'}(z') \in
Z$. Then since $Z$ is a connected component, we may take $z,z'$ to be the least
such values that $p_{s,t}(z),p_{s',t'}(z) \in Z$, and connected
$p_{s,t}(z),p_{s',t'}(z)$ by a path in $Z$ to obtain a path $q : [0,1] \to U$
such that $q(0)=s,q(1)=s'$, and $\forall \lambda, q(\lambda)_1 \leq a$. Then by
the above claim, $S=S'$; the same holds for $T,T'$, so the mapping is injective.
This completes the proof of the proposition.
\end{proof}

\begin{proposition}\label{prop:induction}
For any set $\cH$ of measurable functions $f : \bR^d \to \pmset$ and any $r > 1$,
\[
  P_r(n,m) \leq P_r(n,m-1) + (r-1) \cdot P_r(n-1,m-1) \,.
\]
\end{proposition}
\begin{proof}
Let $A \in \cA_n$ and $a_{i,j} \in \bR$, $i\in [m],j \in [r-1]$ such that the
number of connected components in $A \setminus B$, where $B = \partial f
\cup \bigcup_{i,j} H_{i,j}$ and $H_{i,j} = \{ x \in A \;|\; x_i = a_{i,j}\}$, is
$P_r(n,m)$. For $0 \leq k \leq r-1$ let
\[
B_k \define \partial f \cup
  \left(\bigcup_{i=1}^{m-1} \bigcup_{j=1}^{r-1} H_{i,j}\right) \cup
  \left(\bigcup_{j=1}^{k} H_{m,j}\right) \,,
\]
so that $B = B_{r-1}$ and $B_k = B_{k-1} \cup H_{m,k}$. Since $B_0$ is an
$(r,n,m-1)$-arrangement, $\comp(A \setminus B_0) \leq P_r(n,m-1)$. For
$k > 0$,  since $B_k$ is obtained from $B_{k-1}$ by
adding a hyperplane $H_{m,k}$, Proposition~\ref{prop:cutting bound} implies
\[
\comp(A \setminus B_k)
\leq \comp(A \setminus B_{k-1}) + \comp(H \setminus B_{k-1})
\leq \comp(A \setminus B_{k-1}) + P_r(n-1,m-1) \,,
\]
because $H \setminus B_{k-1}$ is an $(r,n-1,m-1)$-arrangement. Iterating $r-1$
times, once for each added hyperplane, we arrive at
\begin{align*}
  P_r(n,m)
    &= \comp(A \setminus B) \\
    &= \comp(A \setminus B_0) + \sum_{k=1}^{r-1}\left(
      \comp(A \setminus B_k) - \comp(A \setminus B_{k-1}) \right) \\
    &\leq P_r(n,m-1) + (r-1) P_r(n-1,m-1) \,. \qedhere
\end{align*}
\end{proof}

\begin{lemma}\label{lemma:bernoulli sum}
For any set $\cH$ of measurable functions $\bR^d \to \pmset$ and any $r$,
\[
  P_r(d,d) \leq (r-1)^d
	+ \sum_{i=0}^{d-1}  \binom{d}{i}  \cdot M(d-i) \cdot (r-1)^i \,.
\]
\end{lemma}
\begin{proof}
Write $s = r-1$ for convenience.
We will show by induction the more general statement that for any $m \leq n \leq
d$,
\[
	P_r(n,m) \leq \sum_{i=0}^m \binom{m}{i} \cdot M(n-i) \cdot s^i
\]
where we define $M(0) \define 1$. In the base case, note that
$P_r(n,0) = M(n)$.  Assume the statement holds for
$P_r(n',m')$ when $n' \leq n, m' < m$.  Then by
Proposition~\ref{prop:induction},
\begin{align*}
P_r(n,m)
&\leq P_r(n,m-1) + s \cdot P_r(n-1,m-1) \\
&\leq \sum_{i=0}^{m-1} {m-1 \choose i} \cdot M(n-i) \cdot s^i
    + \sum_{i=0}^{m-1} {m-1 \choose i} \cdot M(n-1-i) \cdot s^{i+1} \\
&\leq \sum_{i=0}^{m-1} {m-1 \choose i} \cdot M(n-i) \cdot s^i
    + \sum_{i=1}^{m} {m-1 \choose i-1} \cdot M(n-i) \cdot s^{i} \\
&= M(n) + M(n-m) \cdot s^m + \sum_{i=1}^{m-1}\left({m-1 \choose i} + {m-1
\choose i-1}\right) \cdot M(n-i) \cdot s^i \\
&= \sum_{i=0}^m \binom{m}{i} \cdot M(n-i) \cdot s^i \,. \qedhere
\end{align*}
\end{proof}

\begin{lemma}\label{lemma:bbs for low-components}
Let $\cH$ be a set of functions $f : \bR^d \to \pmset$ such that for some $k
\geq 1$, $M(n) \leq k^n$.  Then for any $\epsilon > 0$ and $r \geq
3dk/\epsilon$, $\bbs(\cH,r) < \epsilon \cdot r^d$.
\end{lemma}
\begin{proof}
Write $s = r-1$.  By Proposition~\ref{prop:bbs bound} and Lemma~\ref{lemma:bernoulli sum}, the
probability that $v$ is a non-constant block is
\begin{align*}
  \frac{\mathsf{bbs}(r)}{r^d}
  &\leq r^{-d}\left(P_r(d,d) - r^d\right)
  \leq r^{-d}\left(\sum_{i=0}^{d-1}\left[{d \choose i} \cdot M(d-i)
    \cdot s^i\right] + s^d - r^d\right) \\
  &\leq r^{-d} \sum_{i=0}^{d-1}{d \choose i} \cdot M(d-i) \cdot s^i \,.
\end{align*}
Split the sum into two parts:
\begin{align*}
&
  \sum_{i=0}^{\floor{d/2}}{d \choose i} \cdot \frac{M(d-i) \cdot s^i}{r^d}
+ \sum_{i=1}^{\ceil{d/2}-1} {d \choose i} \cdot \frac{M(i) \cdot s^{d-i}}{r^d} \\
&\qquad\leq
  \sum_{i=0}^{\floor{d/2}}{d \choose i} \cdot \frac{k^{d-i} \cdot r^i}{r^d}
+ \sum_{i=1}^{\ceil{d/2}-1} {d \choose i} \cdot \frac{k^i \cdot r^{d-i}}{r^d} \\
&\qquad\leq
  \sum_{i=0}^{\floor{d/2}}\frac{d^i k^{d-i} \cdot r^i}{r^d}
+ \sum_{i=1}^{\ceil{d/2}-1} \frac{d^i k^i \cdot r^{d-i}}{r^d}
\qquad\leq
  \sum_{i=0}^{\floor{d/2}}\frac{\epsilon^{d-i}}{3^{d-i}d^{d-i}}
+ \sum_{i=1}^{\ceil{d/2}-1} \frac{\epsilon^i}{3^i} \\
&\qquad\leq
  \frac{\epsilon}{3} + \sum_{i=2}^{\ceil{d/2}-1} \frac{\epsilon^i}{3^i}
  + \floor{d/2} \cdot \frac{\epsilon^{\ceil{d/2}}}{3^{\ceil{d/2}}d^{\ceil{d/2}}}
\qquad\leq
  \frac{\epsilon}{3} + \frac{\epsilon}{3} \sum_{i=1}^{\infty} \frac{\epsilon^i}{3^i}
  + \frac{\epsilon^{\ceil{d/2}}}{3^{\ceil{d/2}}}
\leq \epsilon \,. \qedhere
\end{align*}
\end{proof}

It is a standard fact that for degree-$k$ polynomials, $M(1) \leq k$, and a
special case of a theorem of Warren bounds gives a bound for larger dimensions:
\begin{theorem}[\cite{War68}]\label{thm:warren}
Polynomial threshold functions $p : \bR^d \to \pmset$ of degree $k$ have
$M(n) \leq 6(2k)^n$.
\end{theorem}
Since $M(1) \leq \sqrt{24} k$ and $6(2k)^n \leq (\sqrt{24}k)^n$, for $n > 1$,
Lemma~\ref{lemma:bbs for low-components} gives us:
\begin{corollary}\label{cor:ptf bbs bound}
For $r \geq 3 \sqrt{24} dk/\epsilon$, $r^{-d} \cdot \bbs(\cP_k,r) < \epsilon$.
\end{corollary}

\subsection{Application}

As was the case for halfspaces, our reduction of noise sensitivity on $[r]^d$ to
$\pmset^d$ requires that the class $\cP_k$ is invariant under linear
transformations:
\begin{proposition}\label{prop:linear transform ptfs}
For any $f \in \cP_k$ and full-rank linear transformation $A \in \bR^{d \times
d}$, $f \circ A \in \cP_k$.
\end{proposition}
\begin{proof}
Let $f(x) = \sign(p(x))$ where $p$ is a degree-$k$ polynomial and let $c_q
\prod_{i=1}^d x_i^{q_i}$ be a term of $p$, where $c \in \bR$ and $q \in
\bZ_{\geq 0}^d$ such that $\sum_i q_i \leq k$.  Let $A_i \in \bR^d$ be the
$i^\text{th}$ row of $A$. Then
\[
  \prod_{i=1}^d (Ax)_i^{q_i}
  = \prod_{i=1}^d \left(\sum_{j=1}^d A_{i,j}x_j\right)^{q_i}
  = p_q(x)
\]
where $p_q(x)$ is some polynomial of degree at most $\sum_{i=1}^d q_i \leq k$.
Then $p \circ A = \sum_q c_q p_q$ where $q$ ranges over $\bZ_{\geq 0}$ with
$\sum_i q_i \leq k$, and each $p_q$ has degree at most $k$, so $p \circ A$ is a
degree-$k$ polynomial.
\end{proof}
The last ingredient we need is the following bound of Diakonikolas \emph{et
al.}~on the noise sensitivity:
\begin{theorem}[\cite{DHK+10}]\label{thm:ptf noise sensitivity}
Let $f : \pmset^d \to \pmset$ be a degree-$k$ PTF\@. Then for any $\delta \in
[0,1]$,
\begin{align*}
  \mathsf{ns}_{2,\delta}(f) &\leq O(\delta^{1/2^k}) \\
  \mathsf{ns}_{2,\delta}(f) &\leq 2^{O(k)} \cdot
\delta^{1/(4k+2)}\log(1/\delta) \,.
\end{align*}
\end{theorem}
Putting everything together, we obtain a bound that is polynomial in $d$ for any
fixed $k,\epsilon$, and which matches the result of Diakonikolas \emph{et
al.}~\cite{DHK+10} for the uniform distribution over $\pmset^d$.


\thmptfs*

\begin{proof}
We prove the continuous case here; the finite case is proved in \cref{thm:finite algorithms}.

Let $r = \ceil{9dk/\epsilon}$, so that by Corollary~\ref{cor:ptf bbs bound},
condition 1 of Lemma~\ref{lemma:general algorithm} is satisfied. Due to
Proposition~\ref{prop:linear transform ptfs}, we may apply Theorem~\ref{thm:ptf
noise sensitivity} and Lemma~\ref{lemma:noise sensitivity reduction} to conclude
that for all $f \in \cP_k$
\begin{align*}
\mathsf{ns}_{r,\delta}(f^\block)
  &\leq O(\delta^{1/2^k}) \\
\mathsf{ns}_{r,\delta}(f^\block)
  &\leq 2^{O(k)} \cdot \delta^{1/(4k+2)}\log(1/\delta) \,.
\end{align*}
In the first case, setting $\delta = O(\epsilon^{2^{k+1}})$ we get
$\mathsf{ns}_{r,\delta}(f^\block) < \epsilon^2/3$, so by Lemma~\ref{lemma:noise
sensitivity to polynomial regression} we get a set $\cF$ of functions $[r]^d \to
\bR$ of size $|\cF| \leq (rd)^{O\left(\frac{1}{\epsilon^{2^{k+1}}}\right)}$
satisfying condition 2 of Lemma~\ref{lemma:general algorithm}. For $n =
\poly(|\cF|,1/\epsilon)$, Lemma~\ref{lemma:general algorithm} implies an
algorithm with sample size
\[
  O(rd^2n^2\log(rd))
  = O\left(\frac{d^3k}{\epsilon}\log(dk/\epsilon)\right) \cdot
    \left(\frac{kd}{\epsilon}\right)^{O\left(\frac{1}{\epsilon^{2^{k+1}}}\right)}
\,.
\]
In the second case, setting $\delta =
O\left(\left(\frac{2^{O(k)}\log(2^k/\epsilon)}{\epsilon^2}\right)^{4k+2}\right)$,
we again obtain $\mathsf{ns}(f^\block)_{r,\delta} \leq \epsilon^2/3$ and get
an algorithm with sample size
\[
  \left(\frac{kd}{\epsilon}\right)^{2^{O(k^2)}\left(\frac{\log(1/\epsilon)}{\epsilon^2}\right)^{4k+2}} \,.
\]
The final result is obtained by applying Lemma~\ref{lemma:brute force}.
\end{proof}

\section{Learning \& Testing \texorpdfstring{$k$}{k}-Alternating Functions}\label{section:k alternating}
A function $f : \cX \to \pmset$ on a partial order $\cX$ is \emph{$k$-alternating}
if for every chain $x_1 < \dotsc < x_{k+2}$ there is $i \in [k+1]$ such that
$f(x_i) = f(x_{i+1})$. Monotone functions are examples of 1-alternating functions. We consider
$k$-alternating functions on $\bR^d$ with the usual partial order: for $x,y \in
\bR^d$ we say $x < y$ when $x_i \leq y_i$ for each $i \in [d]$ and $x \neq y$.
Once again, we must bound the block boundary size, which has been done already by
Canonne~\emph{et al.}. We include the proof because their work does not share
our definition of block boundary size, and because we use it in our short proof
of the monotonicity tester in Section~\ref{section:warm-up}.
\begin{lemma}[\cite{CGG+19}]\label{lemma:monotone block boundary}
The $r$-block boundary size of $k$-alternating functions is at most $kdr^{d-1}$.
\end{lemma}
\begin{proof}
Let $f$ be $k$-alternating, let $\block : \bR^d \to {[r]}^d$ be any block
partition and let $v_1, \dotsc, v_m$ be any chain in ${[r]}^d$. Suppose that
there are $k+1$ indices $i_1, \dotsc, i_{k+1}$ such that $f$ is not constant on
$\block^{-1}(v_{i_j})$. Then there is a set of points $x_1, \dotsc, x_{k+1}$
such that $x_j \in \block^{-1}(v_{i_j})$ and $x_j \neq x_{j+1}$ for each $j \in
[k]$. But since $v_{i_1} < \dotsm < v_{i_{k+1}}$, $x_1 < \dotsm < x_{k+1}$ also,
which contradicts the fact that $f$ is $k$-alternating. Then every chain in
${[r]}^d$ has at most $k$ non-constant blocks, and we may partition ${[r]}^d$
into at most $dr^{d-1}$ chains by taking the diagonals $v + \lambda \vec 1$
where $v$ is any vector satisfying $\exists i : v_i=1$ and $\lambda$ ranges over
all integers.
\end{proof}
Canonne \emph{et al.}~also use noise sensitivity bound to obtain a spanning set
$\cF$; we quote their result.
\begin{lemma}[\cite{CGG+19}]\label{lemma:k-monotone grid basis}
There is a set $\cF$ of functions ${[r]}^d \to \bR$, with size
\[
|\cF| \leq
\exp{O\left(\frac{k\sqrt{d}}{\epsilon^2} \log(rd/\epsilon)\right)} \,,
\]
such that for any $k$-alternating function $h : {[r]}^d \to \pmset$, there is $g : {[r]}^d \to \bR$
that is a linear combination of functions in $\cF$ and $\Exu{x \sim {[r]}^d}{{(h(x)-g(x))}^2} \leq
\epsilon^2$.
\end{lemma}

\thmkalternating*
\begin{proof}
We prove the continuous case; for the finite case see \cref{thm:finite algorithms}.

Let $r = \ceil{dk/\epsilon}$ and let $\block : \bR^d \to [r]^d$ be any $r$-block
partition. By Lemma~\ref{lemma:monotone block boundary},  the
first condition of Lemma~\ref{lemma:general algorithm} is satisfied. Now let $f
\in \cH$ and consider $f^\block$. For any chain $v_1 < v_2 < \dotsm < v_m$ in
$[r]^d$, it must be $\blockpoint(v_1) < \blockpoint(v_2) < \dotsm <
\blockpoint(v_m)$ since every $x \in \block^{-1}(v_i), y \in \block^{-1}(v_j)$
satisfy $x < y$ when $v_i < v_j$; then $f$ alternates at most $k$ times on the
chain $\blockpoint(v_1) < \dotsm < \blockpoint(v_m)$ and, since $f^\block(v_i) =
f(\blockpoint(v_i))$, $f^\block$ is also $k$-alternating.  Therefore the set
$\cF$ of functions given by Lemma~\ref{lemma:k-monotone grid basis} satisfies
condition 2 of Lemma~\ref{lemma:monotone block boundary}, and we have $n =
\poly(|\cF|,1/\epsilon) =
\exp{O\left(\frac{k\sqrt{d}}{\epsilon^2}\log(rd/\epsilon)\right)}$. Applying
Lemma~\ref{lemma:general algorithm} gives an algorithm with sample complexity
\[
  O\left(rd^2n^2\log(rd)\right)
  = O\left(\frac{d^3k}{\epsilon} \log(dk/\epsilon) \cdot
	\left(\frac{dk}{\epsilon}\right)^{O\left(\frac{k \sqrt d}{\epsilon^2}\right)}\right)
  = \left(\frac{dk}{\epsilon}\right)^{O\left(\frac{k \sqrt d}{\epsilon^2}\right)} \,.
\]
	The other sample complexity follows from Lemma~\ref{lemma:brute force}.
\end{proof}

\thmtestingkalternating*

\begin{proof}
The following argument is for the continuous case, but generalizes to the finite case using the
definitions in \cref{section:discrete distributions}.

Let $\cH$ be the class of $k$-alternating functions.
Suppose there is a set $\cK \subset \cH$, known to the algorithm, that is a $(\tau/2)$-cover. Then,
taking a set $Q$ of $q=O(\tfrac{1}{\epsilon^2}\log |\cK|)$ independent random samples from $\mu$ and
using Hoeffding's inequality,
\begin{align*}
  \Pru{Q}{\exists h \in \cK : |\dist_Q(f,h)-\dist_\mu(f,h)| > \frac \tau 2 }
  &\leq |\cK| \cdot \max_{h \in \cK} \Pr{ |\dist_Q(f,h) - \dist_\mu(f,h)| > \frac \tau 2 } \\
  &\leq |\cK| \cdot 2\exp{-\frac{q\tau^2}{2}} < 1/6 \,.
\end{align*}
Then the tester accepts if $\dist_Q(f,\cK) < \epsilon_1 + \tau$ and rejects otherwise; we now
prove that this is correct with high probability. Assume that the above estimation is accurate,
which occurs with probability at least $5/6$.  If $\dist_\mu(f,\cH) \leq \epsilon_1$ then
$\dist_\mu(f,\cK) \leq \dist_\mu(f,h) + \dist_\mu(h,\cK) \leq \epsilon_1 + \tau/2$. Then for $g \in
\cK$ minimizing $\dist_\mu(f,g)$,
\[
  \dist_Q(f,\cK) \leq \dist_Q(f,g) < \dist_\mu(f,g) + \frac \tau 2
  \leq \epsilon_1 + \tau \,,
\]
so $f$ is accepted. Now suppose that $f$ is accepted, so $\dist_Q(f,\cK) < \epsilon_1 + \tau$. Then
\[
  \dist_\mu(f,\cH) \leq \dist_\mu(f,g) \leq \dist_Q(f,g) + \frac \tau 2 < \epsilon_1 +
\frac{3}{2}\tau = \epsilon_1 + \frac{3}{4}(\epsilon_2-\epsilon_1) \leq \epsilon_2  \,.
\]
What remains is to show how the tester constructs such a cover $\cK$.

Consider the learning algorithm of Theorem~\ref{thm:k-alternating} with error parameter $\tau/12$,
so $r = \ceil{12dk/\tau}$.  Let $X$ be the grid constructed by that algorithm and let $\block :
\bR^d \to [r]^d$ be the induced $r$-block partition. We may assume that with probability at least
$5/6$, $\|\block(\mu)-\unif([r]^d)\|_\TV < \tau/12$; suppose that this event occurs. The learner
then takes $m = \left(\frac{dk}{\tau}\right)^{O\left(\frac{k \sqrt d}{\tau^2}\right)}$ additional
samples to learn the class $\cH^\block$ with domain $[r]^d$. For every $f \in \cH$ the learner has
positive probability of outputting a function $h : [r]^d \to \zo$ with $\Pru{v}{h(v) \neq
f^\block(v)} < \tau/12$ (where $v$ is chosen from $\block(\mu)$). Let $\cK'$ be the set of possible
outputs of the learner; then $\cK'$ is a $(\tau/12)$-cover for $\cH^\block$.  Construct a set
$\cK^\block$ by choosing, for each $h \in \cK'$, the nearest function $g \in \cK$ with respect to
the distribution $\block(\mu)$. Then $\cK^\block$ is a $(\tau/6)$-cover, since for any function
$f^\block \in \cH^\block$, if $h \in \cK'$ is the nearest output of the learner and $g \in
\cK^\block$ is nearest $h$, then by the triangle inequality $f^\block$ has distance at most $\tau/6$
to $g$ with respect to $\block(\mu)$. Finally, construct a set $\cK \subset \cH$ by taking each
function $h \in \cH$ such that $h^\coarse=h$ and $h^\block \in \cK^\block$ (note that there exists
$h \in \cH$ such that $h^\coarse=h$ since $h^\coarse$ is $k$-alternating when $h^\block$ is
$k$-alternating). Then $\cK$ is a $(\tau/2)$-cover since for any $f \in \cH$, when $h \in \cK$
minimizes $\Pru{v \sim \block(\mu)}{f^\block(v) \neq h^\block}$,
\begin{align*}
  &\dist_\mu(f,\cK) \\
    &\qquad\leq \dist_\mu(f,f^\coarse) + \dist_\mu(f^\coarse, \cK) \\
    &\qquad\leq r^{-d}\cdot \bbs(\cH,r) + \Pru{v \sim \block(\mu)}{f^\block(v)\neq h^\block(v)}
      + 2 \|\block(\mu)-\unif([r]^d)\|_\TV \\
    &\qquad< \tau/6 + \tau/6 + 2\tau/12
    \leq \tau/2 \,.
\end{align*}
Now we bound the size of $\cK^\block$. Since there are $m$ samples and each sample $v \sim
\block(\mu)$ is in $[r]^d$, labelled by $\zo$, there are at most $(r^d)^m 2^m$ possible sample
sequences, so at most $(2r^d)^m$ outputs of the learner (after constructing $X$), so $|\cK^\block|
\leq (2r^d)^m$. Therefore, after constructing $X$, the tester may construct $\cK^\block$ and run the
above estimation procedure, with $q = O\left(\frac{1}{\tau^2} dm \log r \right) =
\left(\frac{dk}{\tau}\right)^{O\left(\frac{k \sqrt d}{\tau^2}\right)}$.
\end{proof}

\appendix

\section{Discrete Distributions}
\label{section:discrete distributions}

We will say that a distribution $\mu_i$ over $\bR$ is \emph{finite}
if it is a distribution over a finite set $X \subset \bR$. In this section, we
extend downsampling to work for finite product distributions: distributions $\mu =
\mu_1 \times \dotsm \times \mu_d$ such that all $\mu_i$ are finite.  As
mentioned in the introduction, our algorithms have the advantage that they do
not need to know in advance whether the distribution is continuous or finite,
and if they are finite they do not need to know the support. This is in
contrast to the algorithms of Blais \emph{et al.}~\cite{BOW10}, which work for
arbitrary finite product distributions but must know the support (since it
learns a function under the ``one-out-of-$k$ encoding'').  Our algorithms have
superior time complexity for large supports.

We begin with an example of a pathological
set of functions that illustrates some of the difficulties in the generalization.

\begin{example}
The \emph{Dirichlet function} $f : \bR \to \pmset$ is the function that
takes value $1$ on all rational numbers and value $-1$ on all irrational
numbers. We will define the \emph{Dirichlet class} of functions as the
set of all functions $f : \bR^d \to \pmset$ such that $f(x)=-1$
on all $x$ with at least 1 irrational coordinate $x_i$, and $f(x)$
is arbitrary for any $x$ with all rational coordinates. Since the
Lebesgue measure of the set of rational numbers is 0, in any continuous
product distribution, any function $f$ in the Dirichlet class satisfies
$\Pr{f(x) \neq -1} = 0$; therefore learning this class is trivial in any
continuous product distribution since we may output the constant $-1$
function. And $\bbs(f,r) = 0$ for this class since no block contains
a set $S$ of positive measure containing $1$-valued points. On the other
hand, if $\mu$ is a finitely supported product distribution, then it may
be the case that it is supported \emph{only} on points with all
rational coordinates. In that case, the Dirichlet class of functions is
the set of all functions on the support, which is impossible to learn when
the size of the support is unknown (since the number of samples will depend on
the support size). It is apparent that our former definition of $\bbs$
no longer suffices to bound the complexity of algorithms when we allow
finitely supported distributions.
\end{example}

Another difficulty arises for finitely supported distributions with small
support: for example, the hypercube $\pmset^d$. Consider what happens when we
attempt to sample a uniform grid, as in the first step of the algorithms above.
We will sample many points $x$ such that $x_1 = 1$ and many points such that
$x_1 = -1$. Essentially, the algorithm takes a small domain $\pmset^d$ and
constructs the larger domain $[r]^d$, which is antithetical to the downsampling
method. A similar situation would occur in large domains $[n]^d$ where some
coordinates have exceptionally large probability densities and are sampled many
times.  Our algorithm must be able to handle such cases, so we must redefine the
grid sampling step and block partitions to handle this situation. To do so, we
introduce \emph{augmented samples}: for every sample point $x \sim \mu$ we will
append a uniformly random value in $[0,1]^d$.

\subsection{Augmented Samples \& Constructing Uniform Partitions}

For \emph{augmented points} $\overline a, \overline b, \in \bR \times [0,1]$,
where $\overline a = (a,a'), \overline b = (b,b')$, we will define a total order by
saying $\overline a < \overline b$ if $a < b$, or $a=b$ and $a' < b'$. Define
interval $(\overline a, \overline b] \define \{ \overline c \;|\;
\overline a < \overline c \leq \overline b \}$. For convenience, when
$\overline a \in \bR \times [0,1]$ and $\overline a = (a,a')$ we will
write $\xi(\overline a) = a$. If
$\overline x \in \bR^d \times [0,1]^d$ is an augmented vector (i.e.~each
coordinate $\overline x_i$ is an augmented point), we will write
$\xi(\overline x) = (\xi(x_1), \dotsc, \xi(x_d))$; and when
$S \subseteq \bR^d \times [0,1]^d$ is a set of augmented points, we will
write $\xi(S) = \{ \xi(\overline x) \;|\; \overline x \in S \}$.

\newcommand{\oblock}{\overline \block}
\begin{definition}[Augmented Block Partition]
An \emph{augmented $r$-block partition} of $\bR^d$ is a pair of functions
$\oblock : \bR^d \times [0,1]^d \to [r]^d$ and $\blockpoint : [r]^d \to \bR^d$
obtained as follows. For each $i \in [d], j \in [r-1]$ let $\overline a_{i,j}
\in \bR \times [0,1]$ such that $\overline a_{i,j} < \overline a_{i,j+1}$ and
define $\overline a_{i,0} = (-\infty, 0), \overline a_{i,r} = (\infty,1)$. For
each $i \in [d], j \in [r]$ define the interval $B_{i,j} = (\overline
a_{i,j-1}, \overline a_{i,j}]$ and a point $\overline b_{i,j} \in \bR \times
[0,1]$ such that $\overline a_{i,j} \leq \overline b_{i,j} \leq \overline
a_{i,j+1}$. The function $\oblock : \bR^d \times [0,1]^d \to [r]^d$ is defined
by setting $\oblock(\overline x)$ to be the unique vector $v \in [r]^d$ such
that $\overline x_i \in B_{i,v_i}$ for each $i \in [d]$.  Observe that
$\oblock^{-1}(v) \define \{ \overline x : \oblock(x)=v \}$ is a set of
augmented points in $\bR^d \times [0,1]$ and that it is possible for two
augmented points $\overline x,\overline y$ to satisfy $\xi(\overline x) =
\xi(\overline y)$ while $\oblock(\overline x) \neq \oblock(\overline y)$.
The function $\blockpoint : [r]^d \to \bR^d$ is defined by setting
$\blockpoint(v) = (\xi(\overline b_{1,v_1}), \dotsc, \xi(\overline
b_{d,v_d}))$; note that this is a non-augmented point.
\end{definition}

\begin{definition}[Block Functions and Coarse Functions]
For a function $f : \bR^d \to \pmset$ we will define the functions $f^\block :
[r]^d \to \pmset$ as $f^\block \define f \circ \blockpoint$ and for each $z \in
[0,1]^d$ we will define $f^\coarse_z : \bR^d \to \pmset$ as $f^\coarse_z(x)
\define f^\block(\oblock(x,z))$. Unlike in the continuous setting,
$f^\coarse_z$ depends on an additional variable $z \in [0,1]^d$, which is
necessary because a single point $x \in \bR^d$ may be augmented differently to
get different $\oblock$ values.
For a distribution $\mu$ over $\bR^d$ define the augmented distribution
$\overline \mu$ over $\bR^d \times [0,1]^d$ as the distribution of $(x,z)$ when
$x \sim \mu$ and $z$ is uniform in $[0,1]^d$. For an augmented $r$-block
partition $\oblock : \bR^d \times [0,1]^d \to [r]^d$ we define the distribution
$\oblock(\mu)$ over $[r]^d$ as the distribution of $\oblock(\overline x)$ for
$\overline x \sim \overline \mu$.
\end{definition}

\begin{definition}[Augmented Random Grid]
An \emph{augmented random grid} $\overline X$ of length
$m$ is obtained by sampling $m$ augmented points $\overline x_1, \dotsc,
\overline x_m \sim \overline \mu$ and for each $i \in [d], j \in [m]$ defining
$\overline X_{i,j}$ to the be $j^\mathrm{th}$ smallest coordinate in dimension
$i$ by the augmented partial order. For any $r$ that divides $m$ we define an
augmented $r$-block partition depending on $\overline X$ by defining for each
$i \in [d], j \in [r-1]$ the points $\overline a_{i,j} = \overline X_{i,mj/r}$,
(and $\overline a_{i,0} = (-\infty,0), \overline a_{i,r} = (\infty,1)$), so
that the intervals are $B_{i,j} = (\overline X_{i,m(j-1)/r}, \overline
X_{i,mj/r}]$ for $j \in \{2, \dotsc, r-1\}$ and $B_{i,1} = ((-\infty,0),
\overline X_{i,m/r}], B_{i,r} = (\overline X_{i,m(r-1)/r}, (\infty,1)]$. We set
the points $\overline b_{i,j}$ defining $\blockpoint : [r]^d \to \bR^d$ to be
$\overline b_{i,j} = \overline X_{i,k}$ for some $\overline X_{i,k} \in
B_{i,j}$.  This is the augmented $r$-block partition \emph{induced} by
$\overline X$.
\end{definition}

\newcommand{\obbs}{\overline \bbs}
\begin{definition}[Augmented Block Boundary]\label{def:augmented block boundary}
For an augmented block partition $\oblock : \bR^d \times [0,1]^d \to [r]^d$, a
distribution $\mu$ over $\bR^d$, and a function $f : \bR^d \to \pmset$,
we say $f$ is \emph{non-constant} on an augmented block
$\oblock^{-1}(v)$ if there are sets $\overline S,\overline T \subset
\oblock^{-1}(v)$ such that $\mu(\xi(\overline S)),\mu(\xi(\overline T))
> 0$ and for all $s \in S, t \in T: f(s) = 1, f(t)=-1$.
For a function $f : \bR^d \to \pmset$ and a number $r$, we define
the augmented $r$-block boundary size $\obbs(f,r)$
as the maximum number of blocks on which $f$ is non-constant with respect to a
distribution $\mu$, where the maximum is taken over all augmented $r$-block
partitions.
\end{definition}

The augmented block partitions satisfy analogous properties to the previously-defined
block partitions:
\begin{lemma}\label{lemma:finite uniform blocks}
Let $\overline X$ be an augmented random grid with length $m$ sampled from a finite
product distirbution $\mu$, and let $\overline \block : \bR^d \times [0,1]^d \to [r]^d$
be the augmented $r$-block partition induced by $\overline X$. Then
\[
	\Pru{\overline X}{ \| \overline \block (\mu) - \unif([r]^d) \|_\mathsf{TV} > \epsilon }
	\leq 4rd \cdot \exp{-\frac{m\epsilon^2}{18 rd^2}} \,.
\]
\end{lemma}
\begin{proof}
Let $\mu_i$ be a finitely supported distribution with support $S \subset \bR$. Let
$\eta = \frac{1}{2} \min_{a,b \in S} |a-b|$. Let $\mu'_i$ be the distribution of
$x_i + \eta z_i$ where $x_i \sim \mu_i$ and $z_i \sim [0,1]$ uniformly
at random; note that $\mu'_i$ is a continuous distribution over $\bR$. For
$\overline x = (x,x'), \overline y = (y,y') \in \bR \times [0,1]$, observe that
$\overline x < \overline y$ iff $x + \eta x' < y + \eta y'$. Therefore,
\[
	\Pru{\overline x, \overline y \sim \overline \mu_i}{\overline x < \overline y}
	= \Pru{x,y \sim \mu'_i}{x < y} \,.
\]
By replacing each finitely supported $\mu_i$ with $\mu'_i$ we obtain a
continuous product distribution $\mu'$ such that $\oblock(\mu)$ is the
same distribution as $\block(\mu')$, so by Lemma~\ref{lemma:continuous
uniform blocks} the conclusion holds.
\end{proof}
\begin{proposition}\label{prop:finite distance to coarse}
For any continuous or finite product distribution $\mu$ over $\bR^d$,
any augmented $r$-block partition $\oblock : \bR^d \times [0,1]^d \to [r]^d$ constructed
from a random grid $\overline X$,
and any $f : \bR^d \to \pmset$,
\[
	\Pru{x \sim \mu, z \sim [0,1]^d}{f(x) \neq f^\coarse_z(x)}
	\leq r^{-d} \cdot \obbs(f,r) + \|\oblock(\mu) - \unif([r]^d)\|_\mathsf{TV} \,.
\]
\end{proposition}
\begin{proof}
The result for continuous product distributions holds by Proposition~\ref{prop:distance to coarse}
and the fact that $\bbs(f,r) \leq \obbs(f,r)$, so assume $\mu$ is a finite
product distribution, and let $S = \supp(\mu)$.

Suppose that for $(x,z)$ sampled from $\overline \mu$, $f(x) \neq
	f^\coarse_z(x)$, and let $v = \oblock(x,z)$. Then for $y =
	\blockpoint(v)$, $f(x) \neq f(y)$ and $x,y \in \xi(\oblock^{-1}(v))$.
	The points $x,y$ clearly have positive measure because $\mu$ is finite,
	so $v$ a non-constant block. Then
	\begin{align*}
		\Pru{x \sim \mu, z \sim [0,1]^d}{f(x) \neq f^\coarse_z(x)}
		&\leq \Pru{x,z}{\oblock(x,z) \text{ is non-constant}} \\
		&\leq \Pru{v \sim [r]^d}{v \text{ is non-constant}}
			+ \|\oblock(\mu) - \unif([r]^d)\|_\mathsf{TV} \,. \qedhere
	\end{align*}
\end{proof}

\subsection{Augmented Block-Boundary Size and Noise Sensitivity}
To obtain learning algorithms for $k$-alternating functions, functions of $k$ convex sets,
functions of $k$ halfspaces, and degree-$k$ PTFs, we must provide a bound on $\obbs$.

For a finite set $X \subset \bR^d$ and a function $f : \bR^d \to \pmset$, we will
call a function $f' : \bR^d \to \pmset$ a \emph{blowup} of $f$ (with respect to $X$)
if $\forall x \in X$ there exists an open ball $B_x \ni x$ where $\forall y \in B_x,
f'(y) = f(x)$.
We will call a set $\cH$ of functions $f : \bR^d \to \pmset$ \emph{inflatable} if
for every finite product set $X = X_1 \times \dotsm \times X_d$ and $f \in \cH$,
there exists $f' \in \cH$ that is a blowup of $f$ with respect to $X$.
\begin{proposition}\label{prop:obbs bound}
Let $\cH$ be a inflatable set of functions. Then $\obbs(\cH,r) \leq \bbs(\cH,r)$.
\end{proposition}
\begin{proof}
	Let $\oblock : \bR^d \times [0,1]^d \to [r]^d$ be an augmented $r$-block partition
	defined by parameters $\overline b_{i,j} \in \bR \times [0,1]$ for $i \in [d], j \in [r-1]$,
	and write $\overline b_{i,j} = (b_{i,j}, b'_{i,j})$. Let $X = X_1 \times \dotsm \times X_d$
	be any finite product set, and let $f \in \cH$; we will bound the number of
	non-constant blocks
	We construct a (non-augmented)
	$r$-block partition as follows. Let $\eta > 0$ be sufficiently small that:
	\begin{itemize}
		\item $\forall x \in X$, the rectangle $R_x \define [x_1,x_1+\eta]\times \dotsm \times [x_d,x_d+\eta]$ is contained within $B_x$,
		\item $\forall i \in [d], [x_i,x_i+\eta] \cap X_i = \{x_i\}$; and
		\item $\forall i \in [d], b_{i,j}+\eta < b_{i,j+1}$ unless $b_{i,j}=b_{i,j+1}$.
	\end{itemize}
	Such an $\eta$ exists since the number of constraints is finite. Then
	define $\block : \bR^d \to [r]^d$ by the parameters $c_{i,j} = b_{i,j}
	+ \eta \cdot b'_{i,j}$. Note that $c_{i,j} = b_{i,j} + \eta \cdot
	b'_{i,j} \leq b_{i,j} + \eta < b_{i,j+1} \leq c_{i,j+1}$. Let $v \in
	[r]^d$ and suppose that $f$ is non-constant on $\oblock^{-1}(v)$, so
	there are $\overline x, \overline y \in \oblock^{-1}(v) \cap (X \times
	[0,1]^d)$ such that $f(x) \neq f(y)$, where $\overline x = (x,x'),
	\overline y = (y,y')$, and $\forall i \in [d], x_i,y_i \in (b_{i,v_i-1},b_{i,v_i}]$ where we
	define $(b,b]=\{b\}$. Consider $\block^{-1}(v) =
	(c_{1,v_1-1},c_{1,v_1}] \times \dotsm \times (c_{d,v_d-1},c_{d,v_d}]$.

	Since $\overline x_i \in (\overline b_{i,v_i-1},\overline b_{i,v_i}]$,
	$x_i \in (b_{i,v_i-1},b_{i,v_i}]$ (where we define $(b,b]=\{b\}$)
	and $x'_i \in (b'_{i,v_i-1},b'_{i,v_i}]$. Therefore
	$x_i + \eta \cdot x'_i \leq b_{i,v_i} + \eta \cdot b'_{i,v_i} = c_{i,v_i}$
	and $x_i + \eta \cdot x'_i > b_{i,v_i-1} + \eta \cdot b'_{i,v_i-1} = c_{i,v_i-1}$
	so $x + \eta \cdot x' \in \block^{-1}(v)$. Also, $x + \eta \cdot x'$ is in
	the rectangle $R_x \subset B_x$ so there is a ball around $x + \eta \cdot x'$, containing
	only points with value $f'(x)=f(x)$. Likewise, there is a ball around $y + \eta \cdot y'$ inside
	$\block^{-1}(v)$ containing only points with value $f'(y)=f(y) \neq
	f'(x)$. Since these balls must intersect $\block^{-1}(v)$ on sets with
	positive measure (in the product of Lebesgue measures), $f'$ is
	non-constant on $\block^{-1}(v)$, which proves the statement.
\end{proof}

\begin{lemma}\label{lemma:k-alternating inflatable}
The set $\cA_k$ of $k$-alternating functions is inflatable.
\end{lemma}
\begin{proof}
Let $f \in \cA_k$ and let $X = X_1 \times \dotsm \times X_d$ be a finite set.
	where we use the standard ordering on $\bR^d$. Let $u \in \bR^d$. We
	claim that the set $\{x \in X : x \leq u\}$ has a unique maximum.
	Suppose otherwise, so there are $x,y \leq u$ that are each maximal. Let
	$x \wedge y = (\max(x_1,y_1), \dotsc, \max(x_d,y_d))$. Then $x \vee y
	\in X$ and $x \wedge y > x,y$ but $u \geq x \wedge y$, a contradiction.
	For every $u \in \bR^d$, write $u^\downarrow$ for this unique maximum.
	Let $\eta > 0$ be small enough that $\forall x \in X, (x + \eta \cdot
	\vec 1)^\downarrow = x$; such a value exists since $X$ is finite.
	Define the map $\phi(u) = (u + (\eta/2)\cdot \vec 1)^\downarrow$ and
	$\forall u \in \bR^d$, we define $f'(u) \define f(\phi(u))$, and argue
	that this satisfies the required properties. It is clear by our choice
	of $\eta$ that $f'(x) = f((x + (\eta/2)\cdot \vec 1)^\downarrow) =
	f(x)$.  Since $\phi$ is order-preserving (i.e.~if $u < v$ then $\phi(u)
	\leq \phi(v)$), $f'$ is $k$-alternating. Now consider the ball $B(x)
	\define \{y \in \bR^d : \|y-x\|_2 < \eta/2\}$. Since $|y_i-x_i| <
	\eta/2$ for all $y \in B(x)$, we have $\phi(y) =
	(y_1+\eta/2,\dotsc,y_d+\eta/2)^\downarrow \leq
	(x_1+\eta,\dotsc,x_d+\eta)^\downarrow = x$, and $\phi(y) \geq (x_1,
	\dotsc, x_d)^\downarrow =x$ so $f'(y)=f(\phi(y))=f(x)$.
\end{proof}

\begin{lemma}\label{lemma:convex set inflatable}
The set $\cC$ of indicator functions of convex sets is inflatable.
\end{lemma}
\begin{proof}
	Let $f : \bR^d \to \pmset$ indicate a closed convex set, let $S = f^{-1}(1)$ be
	this set, and write $\delta(x) \define \min\{\|x-y\|_2 : y \in S\}$ (this
	minimum exists since $S$ is closed). Let $X$ be any finite set and let
	$\delta = \min\{\delta(x) : x \in X\setminus S\}$. Consider $S' = \{x :
	\delta(x) \leq \delta/2 \}$, and let $f'$ be the indicator function
	for this set. Then $f'(x)=f(x)$ for all $x \in X$. Finally, $S'$ is
	closed, and it is convex since for any two points $x,y$, it is
	well-known that the function $\lambda \mapsto \delta(\lambda x +
	(1-\lambda)y)$ is convex for $\lambda \in [0,1]$.
\end{proof}

\begin{lemma}\label{lemma:halfspace inflatable}
The set $\cH$ of halfspaces is inflatable.
\end{lemma}
\begin{proof}
	It suffices to show that for any finite set $X$ (not necessarily a product set)
	and any halfspace $f(x) = \sign(\inn{w,x}-t)$, there is a halfspace
	$f'(x) =\sign(\inn{w,x}-t')$ such that $f(x)=f'(x)$ for all $x \in X$
	but $\inn{w',x}-t \neq 0$ for all $x \in X$; this is a commonly-used fact.
	Let $\delta = \min\{ -(\inn{w,x}-t) : \inn{w,x}-t < 0\}$. It must be the case
	that $\delta > 0$. Then $f'(x) = \sign(\inn{w,x}-t+\delta/2)$ satisfies the condition.
\end{proof}

\begin{lemma}\label{lemma:ptf inflatable}
The set $\cP_k$ of degree-$k$ PTFs is inflatable.
\end{lemma}
\begin{proof}
This follows from the above proof for halfspaces, since for any finite $X$ we may
map $x \in X$ to its vector $(x_1^k, x_2^k, \dotsc)$ of monomials, so that any polynomial
$p(x)$ is a linear threshold function in the space of monomials.
\end{proof}

For a set $\cH$ of functions $f : \bR^d \to \pmset$ and an augmented $r$-block partition
$\oblock : \bR^d \to [r]^d$, we will write $\overline \cH^\block \define \{ f^\block : f \in \cH\}$
for the set of block functions $f^\block : [r]^d \to \pmset$; note that this is not necessarily
the same set of functions as $\cH^\block$ defined for continuous distributions. We must
show that the same learning algorithms used above for learning $\cH^\block$
will work also for $\overline \cH^\block$. For the brute-force learning
algorithm of Lemma~\ref{lemma:brute force}, this is trivial, but for the
regression algorithm in Lemma~\ref{lemma:general algorithm} we must show that
there exists a set $\cF$ such that each $f^\block \in \overline \cH^\block$ is
close to a function $g \in \mathsf{span}(\cF)$. For functions of halfspaces and
PTFs, we used the bound on noise sensitivity, Lemma~\ref{lemma:noise
sensitivity reduction}, to construct a set $\cF$ of functions suitable for the
regression algorithm. The proof for that lemma works without modification for
augmented block partitions, so we have the following:

\begin{lemma}\label{lemma:finite noise sensitivity}
Let $\cH$ be any family of functions $f : \bR^d \to \pmset$ such that, for any linear
transformation $A : \bR^d \to \bR^d$, if $f \in \cH$ then $f \circ A \in \cH$. Let
$\oblock : \bR^d \times [0,1]^d \to [r]^d$ be any augmented $r$-block partition.
Let $\mathrm{ns}_{2,\delta}(\cH) \define \sup_{f \in \cH} \mathsf{ns}_{2,\delta}(f)$.
Then $\mathrm{ns}_{r,\delta}(f^\block) \leq \mathsf{ns}_{2,\delta}(\cH)$.
\end{lemma}

\subsection{Rounding the Output}

After learning a function $g : [r]^d \to \pmset$, we must output a function $g': \bR^d \to \pmset$.
In the continuous setting, we simply output $g \circ \block$. In the finite setting, we cannot simply
output $g \circ \oblock$ since $\oblock : \bR^d \times [0,1]^d \to [r]^d$ requires
an additional argument $z \in [0,1]^d$. For example, if the distribution $\mu$ is
a finitely supported distribution on $\pmset^d$, then for each point $x \in \pmset^d$
there may be roughly $(r/2)^d$ points $v \in [r]^d$ for which $(x,z) \in \oblock^{-1}(v)$
for an appropriate choice of $z \in [0,1]^d$, and these points $v$ may have
different values in $g$. The algorithm must choose a single value to output for each $x$.
We do so by approximating the function $x \mapsto \Exu{z}{g_z(x)}$ and then rounding it
via the next lemma.

\begin{lemma}\label{lemma:rounding}
Fix a domain $\cX$, let $\gamma : \cX \to [-1,1]$, and let $\epsilon > 0$. There is an algorithm
such that, given query access to $\gamma$ and sample access to any distribution $\cD$ over $\cX \times \pmset$, uses at most
$O\left(\log(1/\delta)/\epsilon^2\right)$ samples and queries
and with probability at least $1-\delta$
produces a value $t$ such that
	\[
	\Pru{(x,b)\sim \cD}{\sign(f(x)-t) \neq b}
	\leq \frac{1}{2} \Exu{(x,b) \sim \cD}{|f(x)-b|}
	+ \epsilon \,.
	\]
\end{lemma}
\begin{proof}
	Let $\cT$ be the set of functions $x \mapsto \sign(\gamma(x)-t)$ for
	any choice of $t \in [-1,1]$. We will show that the VC dimension of
	$\cT$ is 1.  Suppose for contradiction that two points $x,y \in \cX$
	are shattered by $\cT$, so in particular there are $s,t \in \bR$ such
	that $\sign(f(x)-s)=1$ and $\sign(f(y)-s)=-1$, while $\sign(f(x)-t)=-1$
	and $\sign(f(y)-t)=1$. Without loss of generality, suppose $s < t$. But
	then $\sign(f(y)-s) \geq \sign(f(y)-t)$, which is a contradiction.
	Therefore, by standard VC dimension arguments (\cite{SB14}, Theorem 6.8), using
	$O(\log(1/\delta)/\epsilon^2)$ samples and choosing $t$ to minimize the
	error on the samples, with probability at least $1-\delta$ we will
	obtain a value $t$ such that
	\[
		\Pru{(x,b)\sim \cD}{\sign(\gamma(x)-t) \neq b}
		\leq \Pru{(x,b) \sim \cD}{\sign(\gamma(x)-t^*) \neq b} + \epsilon
	\]
	where $t^*$ minimizes the latter quantity among all values $[-1,1]$.
	Since the algorithm can restrict itself to those values $t \in [-1,1]$
	for which $\gamma(x)=t$ for some $x$ in the sample, the value minimizing
	the error on the sample can be computed time polynomial in the number
	of samples. Next, we show that the minimizer $t^*$ satisfies the desired
	properties.
	Suppose that $t \sim [-1,1]$ uniformly at random.
	For any $y \in [-1,1], b \in \pmset$,
	\[
		\Pru{t}{\sign(y-t) \neq b} =
		\begin{cases}
			\Pru{t}{t > y} = \frac12|b-y| &\text{ if } b = 1 \\
			\Pru{t}{t \leq y} = \frac12|y-b| &\text{ if } b = -1 \,.
		\end{cases}
	\]
	Therefore
	\[
		\Exu{t \sim [-1,1]}{\Pru{(x,b)\sim \cD}{\sign(\gamma(x)-t) \neq b}}
		= \Exu{(x,b) \sim \cD}{\Pru{t}{\sign(f(x)-t) \neq b}}
		= \frac{1}{2} \Exu{(x,b) \sim \cD}{|\gamma(x)-b|} \,,
	\]
	so we can conclude the lemma with
	\[
		\Pru{(x,b) \sim \cD}{\sign(\gamma(x)-t^*) \neq b}
		\leq \frac12 \Exu{(x,b) \sim \cD}{|\gamma(x)-b|} \,. \qedhere
	\]
\end{proof}

\begin{lemma}
Let $\oblock : \bR^d \times [0,1]^d \to [r]^d$ be an augmented $r$-block partition.
There is an algorithm which, given $\epsilon,\delta > 0$,
query access to a function $g : [r]^d \to \pmset$ and sample access to
a distribution $\cD$ over $\bR^d \times \pmset$, outputs a function $g' : \bR^d \to \pmset$
such that, with probability $1-\delta$,
	\[
		\Pru{(x,b) \sim \cD}{g'(x) \neq b}
		\leq \Pru{(x,b) \sim \cD, z \sim [0,1]^d}{g(\oblock(x,z)) \neq b} + \epsilon \,,
	\]
	uses at most
	$O\left(\frac{d\log(r)}{\epsilon^2}\log\frac{1}{\delta}\right)$ samples
	and queries, and runs in time polynomial in the number of samples.
\end{lemma}
\begin{proof}
	For $z \in [0,1]^d$, write $g_z(x) = g(\oblock(x,z))$.
	For any $(x,b)$,
	\[
		|\Exu{z}{g_z(x)}-b| = |b \Pru{z}{g_z(x)=b} - b \Pru{z}{g_z(x) \neq b} - b|
		= | -2b \Pru{z}{g_z(x) \neq b} | = 2 \Pru{z}{g_z(x) \neq b} \,,
	\]
	so
	\[
		\frac12 \Exu{(x,b) \sim \cD}{|\Exu{z}{g_z(x)}-b|}
			= \Pru{(x,b) \sim \cD, z}{g_z(x) \neq b} \,.
	\]
	The algorithm will construct a function $\gamma(x) \approx \Exu{z}{g_z(x)}$ and then
	learn a suitable parameter $t$ for rounding $\gamma(x)$ to $\sign(\gamma(x)-t)$.

	First the algorithm samples a set $Z \subset [0,1]^d$ of size $m =
	\frac{2d\ln(r)\ln(1/\delta)}{\epsilon^2}$ and construct the function
	$\gamma(x) = \frac{1}{m} \sum_{z \in Z} g(\oblock(x,z))$.
	Fix $Z \subset [0,1]^d$ and suppose $x \in \bR^d$ satisfies $\gamma(x) \neq \Exu{z}{g_z(x)}$.
	Then there must be $w,z \in [0,1]^d$ such that $\oblock(x,z) \neq \oblock(x,w)$, otherwise
	$g_z(x)=g_w(x)$ for all $z,w$ so for all $w, \gamma(x) = g_w(x) = \Exu{z}{g_z(x)}$.
	There can be at most $r^d$ values of $x \in \bR^d$ for which $\exists z,w \in [0,1]^d :
	\oblock(x,z) \neq \oblock(x,w)$, so by the union bound and the Hoeffding bound,
	\begin{align*}
		\Pru{Z}{\exists x \in \bR^d : |\gamma(x)-\Exu{z}{g_z(x)}| > \epsilon}
		&\leq r^d \max_{x \in X} \Pru{Z}{ |\gamma(x)-\Exu{z}{g_z(x)}| > \epsilon} \\
		&\leq 2r^d \exp{-\frac{m\epsilon^2}{2}} < \delta \,.
	\end{align*}
	Therefore with probability at least $1-\delta/2$, $\gamma$ satisfies
	$|\gamma(x)-\Exu{z}{g_z(x)}| \leq \epsilon$ for all $x$. Suppose this occurs.
	Then
	\begin{align*}
	\frac{1}{2}\Exu{(x,b) \sim \cD}{|\gamma(x)-b|}
	&\leq \frac{1}{2}\Exu{(x,b) \sim \cD}{|\Exu{z}{g_z(x)}-b| + |\gamma(x)-\Exu{z}{g_z(x)}|} \\
	&\leq \Pru{(x,b) \sim \cD, z}{g_z(x) \neq b} + \frac{\epsilon}{2} \,.
	\end{align*}
	Now we apply Lemma~\ref{lemma:rounding} with error $\epsilon/2$, using $O(\log(1/\delta)/\epsilon^2)$ samples and polynomial time, to output a value $t$ such that with probability $1-\delta/2$,
	\[
		\Pru{(x,b) \sim \cD}{\sign(\gamma(x)-t) \neq b}
		\leq \frac{1}{2} \Exu{(x,b)}{|\gamma(x)-b|} + \frac{\epsilon}{2}
		\leq \Pru{(x,b) \sim \cD,z}{g_z(x) \neq b} + \epsilon \,. \qedhere
	\]
\end{proof}

\subsection{Algorithms for Finite Distributions}
\label{section:finite algorithms}

We now state improved versions of our monotonicity tester and two general learning algorithms: the
``brute force'' learning algorithm (Lemma~\ref{lemma:brute force}) and the ``polynomial regression''
algorithm (Lemma~\ref{lemma:general algorithm}).  Using these algorithms we obtain the same
complexity bounds as for continuous product distributions, but the algorithms can now handle finite
product distributions as well.

\thmmonotonicity*

\begin{proof}
The proof of \cref{thm:monotonicity} goes through as before, with $\block$ replaced by $\oblock$,
$\block^{-\downarrow}(v)$ replaced with $\oblock^{-\downarrow}(v)$ defined as the infimal element
$\overline{x}$ such that $\oblock(\overline x) = v$, and $\block^{-\uparrow}(v)$ defined as the
supremal element $\overline x$ such that $\oblock(\overline x) = v$, and $g$ redefined
appropriately.
\end{proof}

Next, we move on to the learning algorithms:
\begin{lemma}\label{lemma:finite brute force}
Let $\cH$ be any set of functions $\bR^d \to \pmset$, let $\epsilon > 0$,
and suppose $r$ satisfies $r^{-d} \cdot \obbs(\cH,r) \leq \epsilon/3$. Then
there is an agnostic learning algorithm for $\cH$ that uses
$O\left(\frac{r^d+rd^2\log(rd/\epsilon)}{\epsilon^2}\right)$ samples
and time and works for any distribution $\cD$ over $\bR^d \times
\pmset$ whose marginal on $\bR^d$ is a finite or continuous product distribution.
\end{lemma}
\begin{proof}
On input distribution $\cD$:
\begin{enumerate}
\item Sample a grid $\overline X$ of size $m =
O(\frac{rd^2}{\epsilon^2}\log(rd/\epsilon))$ large enough that Lemma~\ref{lemma:finite uniform blocks} guarantees
$\|\oblock(\mu)-\unif([r]^d)\|_\mathsf{TV} < \epsilon/3$ with probability $5/6$,
where $\oblock : \bR^d \times [0,1]^d \to [r]^d$ is the induced augmented $r$-block partition.
\item Agnostically learn a function $g : [r]^d \to \pmset$ with error $\epsilon/3$
	and success probability $5/6$ using $O(r^d/\epsilon^2)$ samples
	from $\cD^\block$.
\item Run the algorithm of Lemma~\ref{lemma:rounding} using
	$O\left(\frac{d\log r}{\epsilon^2}\right)$ samples to obtain $g'$
		and output $g'$.
\end{enumerate}
The proof proceeds as in the case for continuous distributions (Lemma~\ref{lemma:brute force}).
Assume all steps succeed, which occurs with probability at least $2/3$.
After step 3 we obtain $g : [r]^d \to \pmset$ such that, for any $h \in \cH$,
\[
	\Pru{(v,b) \sim \cD^\block}{g(v) \neq b} \leq \Pru{(v,b) \sim
	\cD^\block}{h^\block(v) \neq b} + \epsilon/3 \,.
\]
By Lemma~\ref{lemma:rounding} and Proposition~\ref{prop:finite distance to coarse},
the output satisfies,
\begin{align*}
  \Pru{(x,b) \sim \cD}{g'(x) \neq b}
  &\leq \Pru{(x,b) \sim \cD,z}{g(\oblock(x,z)) \neq b} + \epsilon/3 \\
  &\leq \Pru{(x,b) \sim \cD,z}{h^\block(\oblock(x,z)) \neq b} + 2\epsilon/3 \\
  &\leq \Pru{(x,b) \sim \cD}{h(x) \neq b} + \Pru{x,z}{h(x) \neq h^\coarse_z(x)} + 2\epsilon/3 \\
  &\leq \Pru{(x,b) \sim \cD}{h(x) \neq b} + \epsilon \,. \qedhere
\end{align*}
\end{proof}

We now state the general learning algorithm from Lemma~\ref{lemma:general algorithm},
improved to allow finite product distributions.
\begin{lemma}\label{lemma:finite general algorithm}
Let $\epsilon > 0$ and let $\cH$ be a set of measurable functions $f : \bR^d \to
\pmset$ that satisfy:
\begin{enumerate}
\item There is some $r = r(d,\epsilon)$ such that
	$\obbs(\cH,r) \leq \epsilon \cdot r^d$;
\item There is a set $\cF$ of functions $[r]^d \to \bR$ satisfying:
$\forall f \in \cH, \exists g \in \mathsf{span}(\cF)$ such that for $v \sim
[r]^d, \Ex{(f^\block(v)-g(v))^2} \leq \epsilon^2$.
\end{enumerate}
Let $n = \poly(|\cF|,1/\epsilon)$ be the sample complexity of the algorithm in
Theorem~\ref{thm:polynomial regression}.  Then there is an agnostic learning
algorithm for $\cH$ on finite and continuous product distributions over $\bR^d$, that
uses
$O(\max(n^2,1/\epsilon^2) \cdot rd^2\log(dr))$ samples and runs in time
polynomial in the sample size.
\end{lemma}
\begin{proof}
Let $\overline \cD$ be the augmented distribution, where $\overline x \sim \overline \cD$
is obtained by drawing $x \sim \cD$ and augmenting it with a uniformly random $z \in [0,1]^d$.
We will assume $n > 1/\epsilon$.  Let $\mu$ be the marginal of $\cD$ on $\bR^d$.
For an augmented $r$-block partition, let $\cD^\block$ be the distribution of
$(\oblock(\overline x),b)$ when $(\overline x,b) \sim \overline \cD$. We may
simulate samples from $\cD^\block$ by sampling $(x,b)$ from $\cD$ and
constructing $(\oblock(\overline x),b)$.  The algorithm is as follows:
\begin{enumerate}
\item Sample a grid $X$ of length $m = O(rd^2n^2\log(rd))$; by
	Lemma~\ref{lemma:finite uniform blocks}, this ensures that
	$\|\oblock(\mu)-\unif([r]^d)\|_\mathsf{TV} < 1/12n$ with
	probability $5/6$. Construct $\oblock : \bR^d \to [r]^d$ induced by $X$.
\item Run the algorithm of Theorem~\ref{thm:polynomial regression} on a sample
of $n$ points from $\cD^\block$; that algorithm returns a
function $g : [r]^d \to \pmset$.
\item Run the algorithm of Lemma~\ref{lemma:rounding} using
	$O\left(\frac{d\log r}{\epsilon^2}\right)$ samples to obtain $g'$
		and output $g'$.
\end{enumerate}
The proof proceeds as in the case for continuous distributions (Lemma~\ref{lemma:general
algorithm}).  Assume all steps succeed, which occurs with probability at least $2/3$.  After step 3
we obtain $g : [r]^d \to \pmset$ such that, for any $h \in \cH$,
\[
	\Pru{(v,b) \sim \cD^\block}{g(v) \neq b} \leq \Pru{(v,b) \sim
	\cD^\block}{h^\block(v) \neq b} + \epsilon/3 \,.
\]
By Lemma~\ref{lemma:rounding} and Proposition~\ref{prop:finite distance to coarse},
the output satisfies,
\begin{align*}
  \Pru{(x,b) \sim \cD}{g'(x) \neq b}
  &\leq \Pru{(x,b) \sim \cD,z}{g(\oblock(x,z)) \neq b} + \epsilon/3 \\
  &\leq \Pru{(x,b) \sim \cD,z}{h^\block(\oblock(x,z)) \neq b} + 2\epsilon/3 \\
  &\leq \Pru{(x,b) \sim \cD}{h(x) \neq b} + \Pru{x,z}{h(x) \neq h^\coarse_z(x)} + 2\epsilon/3 \\
  &\leq \Pru{(x,b) \sim \cD}{h(x) \neq b} + \epsilon \,. \qedhere
\end{align*}
\end{proof}

\begin{theorem}\label{thm:finite algorithms}
	There are agnostic learning algorithms for functions of convex sets,
	functions of halfspaces, degree-$k$ PTFs, and $k$-alternating functions
	achieving the sample and time complexity bounds in Theorems~\ref{thm:learning convex sets},~\ref{thm:halfspaces},~\ref{thm:ptfs}, and~\ref{thm:k-alternating}, that work for any finite or continuous product
	distribution over $\bR^d$.
\end{theorem}
\begin{proof}
	This follows from the same arguments as for each of those theorems,
	except with the bounds from Proposition~\ref{prop:obbs bound} and
	Lemmas~\ref{lemma:convex set inflatable},~\ref{lemma:halfspace inflatable},~\ref{lemma:ptf inflatable}, and~\ref{lemma:k-alternating inflatable} to
	bound $\obbs$; Lemma~\ref{lemma:finite noise sensitivity} to bound the
	noise sensitivity; and the improved general algorithms of Lemmas~\ref{lemma:finite brute force} and~\ref{lemma:finite general algorithm}.
\end{proof}

\section{Glossary}\label{glossary}

The notation $f(n) = \widetilde O(g(n))$ means that for some constant $c$, $f(n) = O(g(n)
\log^c(g(n)))$.

The \emph{natural ordering} on the set $[n]^d$ or $\bR^d$ is the partial order where for $x,y \in
\bR^d$, $x < y$ iff $\forall i \in [d], x_i \leq y_i$, and $x \neq y$.

\paragraph*{Property Testing.}
For a set $\cP$ of distributions over $X$ and a set $\cH$ of functions $X \to
\pmset$, a \textbf{distribution-free} property testing algorithm for $\cH$ under
$\cP$ is a randomized algorithm that is given a parameter $\epsilon > 0$. It has access to the input
probability distribution $\cD \in \cP$ via a \emph{sample oracle}, which returns an independent
sample from $\cD$. It has access to the input function $f : X \to \pmset$ via a \emph{query oracle},
which given query $x \in X$ returns the value $f(x)$.
A \textbf{two-sided} distribution-free testing algorithm must satisfy:
\begin{enumerate}
	\item If $f \in \cH$ then the algorithm accepts with probability at least $2/3$;
	\item If $f$ is $\epsilon$-far from $\cH$ with respect to $\mu$ then the algorithm
		rejects with probability at least $2/3$.
\end{enumerate}
A \textbf{one-sided} algorithm must accept with probability 1 when $f \in \cH$. An
\textbf{$(\epsilon_1,\epsilon_2)$-tolerant} tester must accept with probability at least $2/3$ when
$\exists h \in \cH$ such that $\Pru{x \sim \mu}{f(x) \neq h(x)} \leq \epsilon_1$ and reject when $f$
is $\epsilon_2$-far from $\cH$ with respect to $\mu$.

In the \textbf{query model}, the queries to the query oracle can be arbitrary. In the \textbf{sample
model}, the tester queries a point $x \in X$ if and only if $x$ was obtained from the sample oracle.

A tester in the query model is \textbf{adaptive} if it makes its choice of query based on the
answers to previous queries. It is \textbf{non-adaptive} if it chooses its full set of queries in
advance, before obtaining any of the answers.

The \textbf{sample complexity} of an algorithm is the number of samples requested from the sample
oracle. The \textbf{query complexity} of an algorithm is the number of queries made to the query
oracle.

\paragraph*{Learning.}
Let $\cH$ be a set of functions $X \to \pmset$ and let $\cP$ be a set of probability distributions
over $X$. A learning algorithm for $\cH$ under $\cP$ (in the \textbf{non-agnostic} or
\textbf{realizable}) model is a randomized algorithm that receives a parameter $\epsilon > 0$ and
has \emph{sample access} to an input function $f \in \cH$. Sample access means that the algorithm
may request an independent random example $(x,f(x))$ where $x$ is sampled from some input
distribution $\cD \in \cP$. The algorithm is required to output a function $g : X \to \pmset$ that,
with probability $2/3$, satisfies the condition
\[
  \Pru{x \sim \cD}{ f(x) \neq g(x) } \leq \epsilon \,.
\]

In the \textbf{agnostic} setting, the algorithm instead has sample access to an input distribution
$\cD$ over $X \times \zo$ whose marginal over $X$ is in $\cP$ (i.e.~it receives samples of the form
$(x,b) \in X \times \zo$). The algorithm is required to output a function $g : X \to \pmset$ that,
with probability $2/3$, satisfies the following condition: $\forall h \in \cH$,
\[
	\Pru{(x,b) \sim \cD}{g(x) \neq b} \leq \Pru{(x,b) \sim \cD}{h(x) \neq b} + \epsilon \,.
\]

A \textbf{proper} learning algorithm is one whose output must also satisfy $g \in \cH$; otherwise it
is \textbf{improper}.

\paragraph*{VC Dimension.} For a set $\cH$ of functions $X \to \pmset$, a set $S \subseteq X$
is \textbf{shattered} by $\cH$ if for all functions $f : S \to \pmset$ there is a function
$h \in \cH$ such that $\forall x \in S, h(x)=f(x)$. The \textbf{VC dimension} $\mathsf{VC}(\cH)$
of $\cH$ is the size of the largest set $S \subseteq X$ that is shattered by $\cH$.

\bibliographystyle{alpha}
\bibliography{references.bib}

\end{document}